\documentclass[a4paper]{article}
\advance\textwidth by 4cm
\advance\oddsidemargin by -2cm
\advance\textheight by 3cm
\advance\topmargin by -1.5cm

\newif\ifdraft\draftfalse
\newif\iffull\fulltrue
\usepackage{tablefootnote}
\usepackage{bcprules}
\usepackage{multicol}
\usepackage{paralist}
\usepackage{proof}
\usepackage{stmaryrd}
\usepackage{amsthm,amsmath,amssymb}
\usepackage{comment,url}
\usepackage{color}
\newtheorem{theorem}{Theorem}[section]

\newtheorem{corollary}[theorem]{Corollary}
\newtheorem{lemma}[theorem]{Lemma}
\theoremstyle{definition}
\newtheorem{definition}{Definition}[section]
\newtheorem{example}{Example}[section]
\newtheorem{remark}{Remark}[section]

\ifdraft
\newcommand\nk[1]{\textcolor{red}{[#1 -nk]}}
\else
\newcommand\nk[1]{}
\fi
\allowdisplaybreaks
\AtBeginDocument{}

\newcommand\sem[1]{\llbracket #1\rrbracket}
\newcommand\Assumes[2]{\mathtt{Assume}_{#1,#2}}
\newcommand\NewL[1]{\mathtt{new}\;#1}
\newcommand\EndL[1]{\mathtt{end}\;#1}
\newcommand\Tsize[1]{||#1||}
\newcommand\Op[1]{\mathtt{op}_{#1}}
\newcommand\Inj{\mathit{In}}
\newcommand\bty{\mathbf{b}}
\newcommand\rty{\rho}
\newcommand\nty{\tau}
\newcommand\fty{\sigma}
\newcommand\NTE{\Delta}
\newcommand\FTE{\Gamma}
\newcommand\FAIL{\mathtt{fail}}
\newcommand\Assume{\mathtt{assume}}
\newcommand\Assert{\mathtt{assert}}
\newcommand\toBPCF[1]{#1^\sharp}
\newcommand\emptyTE{\epsilon}
\newcommand\Ok{\mathtt{ok}}
\newcommand\Ng{\mathtt{ng}}
\newcommand\Def{\stackrel{\triangle}{=}}
\newcommand\mochi{\textsc{MoCHi}}

\newcommand\con{\sim}
\newcommand\exc[1]{\uparrow_{#1}}
\newcommand\excr[2]{\uparrow_{#1}^{#2}}
\newcommand\eval{\mathbin{\Downarrow}}
\newcommand\dom{\mathtt{dom}}
\newcommand\subsetTE{\subseteq}
\newcommand\codeof[1]{#1.\mathtt{code}}
\newcommand\envof[1]{#1.\mathtt{st}}

\newcommand\mvE{N}
\newcommand\Unpack[2]{{\mathtt{unpackE}_{#1}(#2)}}
\newcommand\unpack[2]{{\mathtt{unpack}_{#1}(#2)};}
\newcommand\unpackwo[2]{{\mathtt{unpack}_{#1}(#2)}}
\newcommand\pack[1]{\mathtt{pack}_{#1}}
\newcommand\packaux[1]{\mathtt{packsub}_{#1}}
\newcommand\Pack[1]{\mathtt{packE}_{#1}}

\newcommand\Clos[2]{\langle #2, #1\rangle}

\newcommand\proj[1]{\downarrow_{#1}}
\newcommand\splitTE[2]{\mathtt{splitTE}_{#2}(#1)}
\newcommand\BsplitTE[3]{\mathtt{splitTE}_{#2}^{\alpha}(#1)}
\newcommand\splitEnv[2]{\mathtt{splitE}_{#2}(#1)}

\newcommand\FV{\mathtt{FV}}
\newcommand\Tr{\Longrightarrow}
\newcommand\deref[1]{\,!{#1}}
\newcommand\BPCF{BPCF}
\newcommand\BPCFE{BPCF\(^{\mathit{exn}}\)}
\newcommand\BPCFAE{BPCF\(^{\mathit{alg}}\)}
\newcommand\BPCFL{BPCF\(^{\mathit{sym}}\)}
\newcommand\BPCFR{BPCF\(^{\mathit{ref}}\)}
\newcommand\BPCFRL{BPCF\(^{\mathit{ref, lin}}\)}
\newcommand\bPCFRL[1]{PCF\(^{\mathit{ref, lin}}_{#1}\)}
\newcommand\bPCF[1]{PCF\(_{#1}\)}
\newcommand\bt{\mathtt{b}}
\newcommand\Gensym{\mathtt{gensym}()}
\newcommand\Tsym{\mathtt{sym}}
\newcommand\rulesp{}
\newcommand\Up{\mathbin{\Uparrow}}

\newcommand\Trfun[3]{#1\stackrel{#3}{\leadsto}}
\newcommand\Tfuns[3]{#1\stackrel{#3}{\to}}
\newcommand\BTfuns[4]{#1\stackrel{#3,#4}{\to}}
\newcommand\Sharable{\mathtt{sharable}}
\newcommand\mkref{\texttt{ref}\;}

\newcommand\REF{\mathtt{ref}}
\newcommand\Tref[1]{#1\;\REF}
\newcommand\FTref[2]{#1\;\REF^{#2}}
\newcommand\BTref[2]{#1\;\REF^{#2}}
\newcommand\LTref[3]{#1\;\REF^{#2}_{#3}}
\newcommand\Tunit{\mathtt{unit}}

\newcommand\set[1]{\{#1\}}
\newcommand\seq[1]{\widetilde{#1}}
\newcommand\BOOL{\mathtt{bool}}
\newcommand\INT{\mathtt{int}}
\newcommand\IF{\mathtt{if}}
\newcommand\THEN{\mathtt{then}}
\newcommand\ELSE{\mathtt{else}}
\newcommand\LET{\mathtt{let}}
\newcommand\IN{\mathtt{in}}
\newcommand\LETREC{\mathtt{letrec}}
\newcommand\Band{\mathbin{\mathtt{\&\&}}}
\newcommand\AND{\;\mathtt{and}\;}
\newcommand\ifexp[2]{\IF\;#1\;\THEN\;#2\;\ELSE\;}
\newcommand\letexp[2]{\LET\;#1=#2\;\IN\;}
\newcommand\letrecexp[1]{\LETREC\;#1\;\IN\;}
\newcommand\fixexp[3]{\Y(\lambda #1.\lambda #2.#3)}
\newcommand\TRUE{\mathtt{true}}
\newcommand\NOT{\mathtt{not}}
\newcommand\FALSE{\mathtt{false}}
\newcommand\sty{\kappa}
\newcommand\COL{\mathbin{:}}
\newcommand\Y{\mathbf{Y}}
\newcommand\p{\vdash}
\newcommand\q{\dashv}
\newcommand\STE{\mathcal{K}}
\newcommand\hole{[\,]}
\newcommand\red{\longrightarrow}
\newcommand\reds{\red^*}
\newcommand\INC[1]{\mathtt{inc}\; c_{#1}}
\newcommand\DEC[1]{\mathtt{dec}\; c_{#1}}
\newcommand\GOTO[1]{\mathtt{goto}\; {#1}}
\newcommand\Mif[3]{\IF\;c_{#1}=0\;\THEN\;\GOTO{#2}\;\ELSE\;(\DEC{#1};\GOTO{#3})}
\newcommand\Mhalt{\mathtt{halt}}
\newcommand\Nat{\mathbf{Nat}}
\newcommand\EX[1]{\mathtt{E}_{#1}}
\newcommand\Ex[2]{\mathtt{E}_{#1}(#2)}
\newcommand\RAISE{\mathtt{raise}\;}
\newcommand\raiseexp[2]{\mathtt{raise}\;\Ex{#1}{#2}}
\newcommand\To{\Rightarrow}
\newcommand\tryexp[3]{\mathtt{try}\;#2\;\mathtt{with}\;\Ex{#1}{#3}\To}
\newcommand\encexn[1]{#1^\dagger}
\newcommand\enceff[1]{#1^\dagger}
\newcommand\RET{\mathtt{return}}

\newcommand\Handleexp[2]{\mathtt{with}\; #2\;\mathtt{handle}\;#1}

\newcommand\Thandler[2]{#1\To #2}

\newcommand\Effexp[3]{#1(#2; #3)}
\newcommand\alg{\sigma}
\newcommand\Alg{\mathtt{A}}
\newcommand\contvar{c}
\newcommand\rethandler[2]{\RET(#1)\mapsto #2}
\newcommand\effhandler[4]{#1(#2;#3)\mapsto #4}
\newcommand\UNIT{\mathtt{unit}}
\newcommand\Vunit{(\,)}

\newcommand\MATCH{\mathtt{match}}
\newcommand\WITH{\mathtt{with}}

\begin{document}

\title{On Decidable and Undecidable Extensions of Simply Typed Lambda Calculus}

\author{Naoki Kobayashi\\The University of Tokyo}

\maketitle

\begin{abstract}
  The decidability of the reachability problem for finitary PCF has been
  used as a theoretical basis for fully automated verification tools
  for functional programs. The reachability problem, however, often becomes
  undecidable for a slight extension of finitary PCF with side effects, such as
  exceptions, algebraic effects, and references, which hindered the extension
  of the above verification tools for supporting
  functional programs with side effects. In this paper,
  we first give simple proofs of the undecidability of four extensions of finitary PCF,
  which would help us understand and analyze the source of undecidability.
  We then focus on an extension with references, and give a decidable fragment
  using a type system.
  To our knowledge, this is the first non-trivial decidable fragment that
  features higher-order recursive functions containing reference cells.
\end{abstract}

\section{Introduction}

Higher-order model checking, or the model checking of  higher-order recursion schemes, has been proven to be decidable in 2006~\cite{Ong06LICS}, and
practical higher-order model checkers (which run fast
for many inputs despite the extremely high worst-case complexity) have
been developed around 2010's~\cite{Kobayashi09PPDP,Kobayashi11FOSSACS,Broadbent12ICALP,Kobayashi13JACM,Kobayashi13horsat,Ramsay14POPL}.
Since then, higher-order model checking has been applied to fully automated
verification of higher-order functional programs~\cite{Kobayashi09POPL,KTU10POPL,Ong11POPL,KSU11PLDI,SUK13PEPM}.

Among others, Kobayashi et al.~\cite{KSU11PLDI,SUK13PEPM} used
the decidability of the reachability problem for finitary PCF (i.e.,
the simply-typed \(\lambda\)-calculus with recursion and finite base types),
which is a direct consequence of the decidability of higher-order model
checking,\footnote{We are not sure whether it has been known before
  Ong's result~\cite{Ong06LICS}, but at least, the above-mentioned development of
  practical higher-order model checkers enabled the use of the result in practice.}
and combined it with predicate abstraction (which approximates a value in an infinite data domain
by a tuple of Booleans that represent whether certain predicates are satisfied~\cite{Graf1997,Ball01,KSU11PLDI})
and counterexample-guided abstraction
refinement, to obtain a fully automated model checker \mochi{} for a subset of OCaml,
just like software model checkers for first-order programs
have been developed based on the decidability of finite-state or pushdown model
checking~\cite{Ball02POPL,SLAM,BLAST}.

Although \mochi{}\footnote{\url{https://github.com/hopv/MoCHi}}~\cite{KSU11PLDI,SUK13PEPM} now supports a fairly large subset of OCaml, including exceptions, references, algebraic data types, and modules,
it often fails to verify correct programs using certain features. For example,
the recent version of \mochi{}
fails to verify even the following, trivially safe program:
\begin{verbatim}
    let x = ref 1 in assert(!x = 1).
\end{verbatim}
That is because \mochi{} approximates a given program
by replacing any read from reference cells with non-deterministic values;
the above program is thus replaced by 
\begin{verbatim}
    assert(Random.int(max_int) = 1)
\end{verbatim}
which may obviously lead to an assertion failure.
Whilst there would be some ad hoc ways to partially address
the problem (for example, in the case of the above program,
we can combine dataflow analysis to infer the
range of values stored in a reference cell),
a fundamental problem is that higher-order model checking is undecidable
for finitary PCF extended with references~\cite{DBLP:journals/jacm/JonesM78,DBLP:journals/apal/Ong04}.
  If we apply predicate abstraction
to a functional program with references, we would obtain a program of finitary PCF
extended with references, but no reasonable model checker is available to
decide the safety of the resulting program.
Similar problems exist for exceptions carrying
functions\footnote{In contrast, exceptions carrying only base-type values do not
  cause a problem; they can be transformed away by a CPS transformation~\cite{SUK13PEPM}.}
and algebraic effects
(which have been recently incorporated in OCaml 5),
as extensions of finitary PCF with those features are also
undecidable~\cite{DBLP:journals/lisp/Lillibridge99,DBLP:journals/pacmpl/LagoG24}.

In view of the above situation, the general goals of our research project are:
\begin{inparaenum}[(i)]
\item to understand and analyze the source of undecidability of various extensions
  of finitary PCF;
\item to find decidable fragments of those extensions; 
\item to use them for improving automated higher-order program verification tools like \mochi{}; and also 
\item possibly to redesign functional languages to help programmers write
  ``well-behaved'' programs that can be mapped (by predicate abstraction, etc.)
  to the decidable fragments.
\end{inparaenum}
As a first step towards the goals, this paper makes the following contributions.
\begin{asparaitem}
\item We provide simple proofs for the undecidability of (the reachability problem of)
  four simple extensions of finitary PCF: (a) finitary PCF with exceptions,
  (b) one with algebraic effects, (c) one with name creation and equality check,
  and (d) one with Boolean references. The undecidability of
  those extensions is either known
  (\cite{DBLP:journals/lisp/Lillibridge99} for (a);
  \cite{DBLP:journals/pacmpl/LagoG24} for (b); and
  \cite{DBLP:journals/jacm/JonesM78} (Theorem 3) and \cite{DBLP:journals/apal/Ong04} (Lemma 39)
  for (d)\footnote{Actually, these results for the extension with references are for call-by-name,
  while our proof is for call-by-value.})
  or probably folklore (for (c)), but the simplicity and uniformity (via
  encoding of Minsky machines) of our proofs may help us better understand and
  analyze the source of undecidability, and find decidable fragments. 
\item We give a decidable fragment of finitary PCF with Boolean references.
  This is obtained by controlling the usage of (function) closures containing Boolean references
  by using a kind of linear types.
  Our technique
   has been partially inspired by recent work on the use of ownerships for
  automated verification of first-order
  programs~\cite{DBLP:journals/toplas/MatsushitaTK21,DBLP:conf/esop/TomanSSI020},
  but to our knowledge, ours is the first to handle higher-order recursive functions.
  We allow closures to capture references and even other closures containing references, like
\begin{verbatim}
  let x = ref true 
  let rec f() = ... x := not(!x) ... 
  let rec g() = ... f() ... 
\end{verbatim}
Here, a reference cell \texttt{x} is captured by the function \texttt{f} (so that
the value of \texttt{x} changes each time \texttt{f} is called),
and \texttt{f} is further captured by \texttt{g} (so that the call of \texttt{g}
may also change the value of \texttt{x}).
The decidability is obtained by a translation from the fragment into finitary PCF.
The improvement of existing verification tools (such as \mochi{}) and
redesign of a functional language (stated above as the general research goals
(iii) and (iv)) are out of the scope of the present paper, but we hint on
how our result may be used to achieve such goals, and report preliminary experimental results
to show that would indeed be the case.
\end{asparaitem}

The rest of this paper is structured as follows.
Section~\ref{sec:pre} reviews Boolean PCF and Minsky machines.
Section~\ref{sec:undecidability} shows the undecidability of
the reachability problem for
closed terms of the four extensions of Boolean PCF.
Section~\ref{sec:reference} introduces a fragment of Boolean PCF with
references, and shows that the reachability problem is decidable for that
fragment.
Section~\ref{sec:exp} discusses applications of the result to automated verification
of higher-order functional programs with references and reports preliminary experiments.
Section~\ref{sec:rel} discusses related work, and Section~\ref{sec:conc}
concludes the paper.
 \section{Preliminaries}
\label{sec:pre}

This section reviews (call-by-value) Boolean PCF (the call-by-value simply-typed
\(\lambda\)-calculus with recursion and Booleans), whose extensions are
the target of study in this paper, and 
the halting problem for
Minsky machines~\cite{Minsky}, whose undecidability is used as
the main basis of our proofs of undecidability of various extensions
of Boolean PCF in Section~\ref{sec:undecidability}.
\subsection{Boolean PCF (\BPCF{})}
The syntax of types and expressions of \BPCF{} are given as follows.
\begin{align*}
  \sty \mbox{ (types) } &::= \bt \mid \sty_1\to\sty_2\qquad
  \bt \mbox{ (base types) }::= \UNIT\mid \BOOL\\
M \mbox{ (expressions) }&::=
    \FAIL\mid V \mid  x\mid M_1 M_2 \mid \fixexp{f}{x}{M} \mid \ifexp{M_0}{M_1}{M_2} \\V \mbox{ (values) }&::=  \Vunit\mid \TRUE\mid \FALSE \mid \lambda x.M
\end{align*}
The expression \(\FAIL\) is a special command that signals an error, which
is used as the target of reachability. The expression
\(\fixexp{f}{x}{M}\) denotes a recursive function such that \(f(x)=M\).
Other expressions are standard and would be self-explanatory.
We write \(\letexp{x}{M_1}{M_2}\) for \((\lambda x.M_2)M_1\), and
write \(M_1;M_2\) for \((\lambda x.M_2)M_1\) if \(x\) does not occur in \(M_2\).
We also often use the notation
\(\letrecexp{f_1\;\seq{x}_1=M_1\AND\cdots\AND f_n\;\seq{x}_n=M_n}{M}\)
for mutually recursive function definitions,
which is a derived form defined inductively (on the number of
mutually recursive functions) as:
\begin{align*}
  &  \letexp{f_1}{\Y(\lambda f_1.\lambda \seq{x}_1.\letrecexp{f_2\;\seq{x}_2=M_2\AND\cdots\AND f_n\;\seq{x}_n=M_n}{M_1})}\\
  &\cdots\\
  &  \letexp{f_n}{\Y(\lambda f_n.\lambda \seq{x}_n.\letrecexp{f_1\;\seq{x}_1=M_1\AND\cdots\AND f_{n-1}\;\seq{x}_{n-1}=M_{n-1}}{M_n})}
M
\end{align*}
Here, \(\seq{x}\) denotes a sequence of variables \(x_1,\ldots,x_k\), and \(\lambda \seq{x}.M\) denotes
\(\lambda x_1.\cdots\lambda x_k.M\). We assume \(\seq{x}_i\) is non-empty for each \(i\in\set{1,\ldots,n}\) in the derived form above. The notation \(\seq{\cdot}\), however,
may represent an empty sequence in general.

The typing rules and operational semantics, which are standard, are given in Figures~\ref{fig:typing}
and \ref{fig:os} respectively. We write \(\reds\) for the reflexive and transitive closure of
\(\red\).
We are interested in the following reachability problem for (various extensions of)
Boolean PCF.
\begin{definition}[reachability problem]
  The \emph{reachability problem} is the problem of deciding whether \(M\reds \FAIL\),
given a closed expression \(M\) such that \( \p M:\bt\).
\end{definition}
The reachability problem for Boolean PCF is decidable, which is obtained as
an immediate corollary of the decidability of higher-order model checking~\cite{Ong06LICS,Kobayashi13JACM}.

\begin{remark}
The reachability problem above should not be confused with
  the functional reachability studied by
  Ong and Tzevelekos~\cite{OngTzevelekos2009LICS}, which asks whether
  there is a context \(C\) such that \(C[M]\) reaches a given program point.
  The functional reachability is undecidable in general.
\end{remark}
    
\begin{remark}
  We consider the call-by-value version of Boolean PCF for the convenience of adding
  effectful operations later, but the decidability of the reachability problem
  is the same for the call-by-name version; notice that
  (type-preserving) CPS transformation makes the difference between call-by-value and call-by-name
  disappear. The complexity bound is different between call-by-value and call-by-name; see, e.g.,
  \cite{DBLP:conf/fossacs/Tsukada014}.
\end{remark}
\begin{figure}
\begin{multicols}{2}
\infrule{}{\STE\p \FAIL:\sty}
\infrule{}{\STE\p \Vunit:\UNIT}
\infrule{b\in \set{\TRUE,\FALSE}}{\STE\p b:\BOOL}
\infrule{\STE,x\COL\sty_1\p M:\sty_2}{\STE\p \lambda x.M:\sty_2}
\infrule{\STE(x)=\sty}{\STE \p x\COL\sty}
\infrule{\STE\p M_1:\sty_2\to\sty\andalso\STE\p M_2:\sty_2}{\STE\p M_1M_2:\sty}
\infrule{\STE,f\COL\sty_1\to\sty_2,x\COL\sty_1\p M:\sty_2}{\STE\p \Y(\lambda f.\lambda x.M):\sty_1\to\sty_2}
\infrule{\STE \p M_0:\BOOL\andalso \STE \p M_1:\sty\andalso \STE\p M_2:\sty}
        {\STE\p \ifexp{M_0}{M_1}{M_2}:\sty}
\end{multicols}
  \caption{Typing Rules for Boolean PCF}
  \label{fig:typing}
\end{figure}

\begin{figure}
  \begin{align*}
    & E\mbox{ (evaluation context) }::=
    \hole \mid E\,M\mid V\,E \mid \ifexp{E}{M_1}{M_2}
    \end{align*}
\begin{multicols}{2}
  \infax{E[(\lambda x.M)V]\red E[[V/x]M]}
  \infax{E[\ifexp{\TRUE}{M_1}{M_2}]\red E[M_1]}
  \infax{E[\FAIL]\red \FAIL}
  \infax{E[\ifexp{\FALSE}{M_1}{M_2}]\red E[M_2]}
\end{multicols}
  \infax{E[\Y(\lambda f.\lambda x.M)]\red E[\lambda x.[\Y(\lambda f.\lambda x.M)/f]M]}
  \caption{Reduction Semantics}
  \label{fig:os}
\end{figure}

\subsection{Minsky Machines}
This subsection reviews Minsky machines~\cite{Minsky}. We write \(\Nat\) for the set of natural numbers.
\begin{definition}[Minsky machine]
  A \emph{Minsky machine} is a finite map \(P\) from a finite set of natural numbers
  \(\set{0,\ldots,m}\) to the set \(I\) of the following three kinds of instructions.
  \begin{itemize}
  \item \(\INC{j};\GOTO{k}\) (where \(j\in\set{0,1}, k\in\set{0,\ldots,m}\)): increments
    counter \(c_j\), and jumps to the \(k\)-th instruction.
  \item \(\Mif{j}{k}{\ell}\) (where \(j\in\set{0,1}, k,\ell\in\set{0,\ldots,m}\)):
    jumps to the \(k\)-th instruction if \(c_j=0\); otherwise 
    decrements \(c_j\), and jumps to the \(\ell\)-th instruction.
    \item \(\Mhalt\): halts the machine.
  \end{itemize}
A \emph{configuration} of a Minsky machine is
  a triple \((i, n_0, n_1)\) where \(i\in\set{0,\ldots,m}\) and \(n_0,n_1\in\Nat\).
  The transition relation is defined by:
  \infrule{P(i)=\INC{j};\GOTO{k}\andalso n_j' = n_j+1\andalso n'_{1-j}=n_{1-j}}
          {(i,n_0,n_1)\red (k,n'_0,n'_1)}
  \infrule{P(i)=\Mif{j}{k}{\ell}\andalso n_j=0}
          {(i,n_0,n_1)\red (k,n_0,n_1)}
  \infrule{P(i)=\Mif{j}{k}{\ell}\quad n_j\ne 0\quad n_j'=n_j-1\quad n'_{1-j}=n_{1-j}}
          {(i,n_0,n_1)\red (\ell,n'_0,n'_1)}
\end{definition}

\begin{definition}[Minsky machine halting problem]
  The halting problem for Minsky machines is the problem of deciding whether
  \((0,0,0)\reds (j,n_1,n_2)\) for some \(j,n_1,n_2\) such that \(P(j)=\Mhalt\),
  given a Minsky machine \(P\).
\end{definition}
  
The halting problem for Minsky machines is undecidable~\cite{Minsky}.
 
\section{Undecidable Extensions of Boolean PCF}
\label{sec:undecidability}

This section considers four kinds of extensions of \BPCF{}:
those with exceptions, algebraic effects and handlers, name generation and equality, and
references. We provide simple proofs of the undecidability of those extensions,
by reduction from the halting problem for Minsky machines.

\subsection{Boolean PCF with Exceptions (\BPCFE{})}

We consider an extension with exceptions in this subsection.
We assume that there exists
an exception constructor \(\EX{\sty}\) for each type \(\sty\).\footnote{In a usual
programming language like OCaml, there can be multiple exception constructors of the same type
and they can be dynamically created, but
we consider here only one (statically allocated) constructor for each type,  for the sake of simplicity.}

We extend the syntax of expressions as follows.
\begin{align*}
  M ::= \cdots \mid \raiseexp{\sty}{M}\mid \tryexp{\sty}{M_1}{x}{M_2}
\end{align*}

Accordingly, we add the following typing rules.
\begin{multicols}{2}
\infrule{\STE\p M:\sty}{\STE\p \raiseexp{\sty}{M}:\sty'}
\infrule{\STE\p M_1:\sty\andalso \STE,x\COL\sty_x\p M_2:\sty}
        {\STE\p \tryexp{\sty_x}{M_1}{x}{M_2}:\sty}
\end{multicols}
The expression \(\raiseexp{\sty}{M}\) evaluates \(M\) to a value \(V\) (of type \(\sty\)), and then raise
an exception \(\Ex{\sty}{V}\).
The expression \(\tryexp{\sty}{M_1}{x}{M_2}\) evaluates \(M_1\), and if an exception \(\Ex{\sty}{V}\)
is raised, then \(M_2\) is evaluated, with \(x\) being bound to \(V\).
The operational semantics is extended as follows.
\begin{align*}
  E \mbox{ (evaluation contexts) }::= \cdots \mid \raiseexp{\sty}{E}\mid \tryexp{\sty}{E}{x}{M}
\end{align*}
\infax{E[\tryexp{\sty}{E'[\raiseexp{\sty}{V}]}{x}{M}]\red
  E[[V/x]M]\\
\mbox{ (if $E'$ does not contain a subexpression of the form
    \(\tryexp{\sty}{E''}{x}{M}\))}}

The reachability problem for the extended language is undecidable.
\begin{theorem}[\cite{DBLP:journals/lisp/Lillibridge99}]
  The reachability problem for Boolean PCF with exceptions, i.e.,
  the problem of deciding whether \(M\reds\FAIL\)
is undecidable.
\end{theorem}
Lillibridge~\cite{DBLP:journals/lisp/Lillibridge99}
has actually shown that the problem is
undecidable even for the fragment without recursion, by encoding the untyped \(\lambda\)-calculus. In his encoding,
exceptions carrying second-order functions (of type
\((\Tunit\to\Tunit)\to(\Tunit\to\Tunit)\)) were used.

Here we give an alternative proof of the theorem above by reduction
from the halting problem for Minsky machines. We use recursion, but
only exceptions carrying first-order functions of type \(\Tunit\to\Tunit\);
thus, the fragment of \BPCFE{} used here is incomparable with the one
used by Lillibridge~\cite{DBLP:journals/lisp/Lillibridge99}.

We encode a natural number \(n\) to an expression \(\encexn{n}\) of type \(\Tunit\to\Tunit\) as follows.
\begin{align*}
  \encexn{0} &= \lambda x.x\qquad
  \encexn{(n+1)} = \lambda x.\raiseexp{\Tunit\to\Tunit}{\encexn{n}}.
\end{align*}
In other words, a natural number \(n\) is encoded into the following function.
\[
\underbrace{\lambda x.\RAISE\EX{\Tunit\to\Tunit}(\lambda x.\RAISE\EX{\Tunit\to\Tunit}(\cdots
\lambda x.\RAISE\EX{\Tunit\to\Tunit}(}_n\lambda x.x)\cdots )).
\]
Using this, a Minsky machine \(\set{0\mapsto I_0,\ldots,m\mapsto I_m}\) can be encoded as
\begin{align*}
  \letrecexp{f_0\;c_0\;c_1 =M_0 \AND \cdots \AND f_m\;c_0\;c_1 = M_m}{f_0\;\encexn{0}\;\encexn{0}}
\end{align*}
where \(M_i\) is:
\begin{itemize}
\item \(\FAIL\), if \(I_i=\Mhalt\);
\item \(\letexp{c_j}{\lambda x.\raiseexp{\Tunit\to\Tunit}{c_j}}{f_k\;c_0\;c_1}\), if \(I_i = \INC{j};\GOTO{k}\); and
\item \(\tryexp{\Tunit\to\Tunit}{(c_j\Vunit; f_k\;c_0\;c_1)}{c_j}{f_\ell\;c_0\;c_1}\)\\ if \(I_i=\Mif{j}{k}{\ell}\).
\end{itemize}
In the last case, \(c_j\Vunit\) raises an exception just if (the value represented by)
\(c_j\) is non-zero; in that case,
the exception carries the value of \(c_j-1\), to which \(c_j\) is bound and \(f_\ell\) is executed.
If the value of \(c_j\) is \(0\), then no exception is raised and \(f_k\) is executed.
From this observation, it should be clear that \((i,n_0,n_1)\red (k,n_0',n_1')\) just if
\(f_i\;n_0\;n_1\) calls \(f_k\;n_0'\;n_1'\) directly.
Thus, the Minsky machine halts if and only if the term above reaches \(\FAIL\); hence
the reachability problem for Boolean PCF with exceptions is undecidable.
\begin{remark}
  Following the idea of Lillibridge~\cite{DBLP:journals/lisp/Lillibridge99}, we can
  also eliminate the use of recursion in the above encoding as follows.
  A recursive function \(\Y(\lambda f.\lambda x.M)\) of type \(\tau\) can be represented by:
  \begin{align*}
    & \LET\ \mathit{selfapp} = \lambda g.(\tryexp{(\Tunit\to\Tunit)\to\tau}{g\;\Vunit; \raiseexp{\Tunit}{{}}}{h}{h})g\; \IN\\
    & \LET\ \mathit{g} = \lambda \_.\raiseexp{(\Tunit\to\Tunit)\to\tau}{
      \lambda g.\lambda x.\letexp{f}{\mathit{selfapp}\;g}{M}}\; \IN\\
    & \mathit{selfapp}\;g
  \end{align*}
  It corresponds to \(\letexp{h}{\lambda g.\lambda x.(\letexp{f}{g\,g}M)}{h\,h}\) in the untyped \(\lambda\)-calculus.
  Here, intuitively, \(\mathit{selfapp}\) is intended to take
  a function \(h\) and computes the self-application \(h\,h\).
  To make it well-typed, however, the actual \(\mathit{selfapp}\) takes a closure \(g\) (of type
  \(\Tunit\to\Tunit\))
  that will raise an exception carrying \(h\), extracts \(h\) by invoking the closure
  and catching the exception, and then applies \(h\) to \(g\);
   the part \(\raiseexp{\Tunit}{}\) is a dummy expression that
   just adjusts the type and will never be executed.
\qed
\end{remark}

\begin{remark}
  In the above undecidability proof and the original proof of
  Lillibridge~\cite{DBLP:journals/lisp/Lillibridge99} (encoding the untyped \(\lambda\)-calculus),
  an exception carrying an exception-raising closure plays a key role.
  Indeed, the reachability problem is decidable if exception values are restricted to base types;
see, e.g., \cite{Kobayashi13JACM,SUK13PEPM}.
  Another, a little weaker restriction to retain the decidability would be to restrict
  values carried by exceptions to pure functions (that do not raise exceptions) and base-type values,
  or to restrict exceptions to ``checked'' ones~\cite{DBLP:journals/lisp/Lillibridge99}
  (where the type of a function contains those of exceptions it may raise;
  for example, a Boolean function that may raise an exception carrying a closure of type \(\tau\)
  is expressed by \(\BOOL\stackrel{\tau}{\to}\BOOL\)). With such restrictions,
  exceptions can be transformed away by an extension of the CPS transformation,
  where an exception handler is passed as another continuation~\cite{SUK13PEPM}.
  \qed
\end{remark}
 
\subsection{An Extension with Algebraic Effects and Handlers (\BPCFAE{})}

This section considers \BPCFAE{}, an extension of Boolean PCF with
algebraic effects~\cite{DBLP:journals/corr/PlotkinP13,DBLP:journals/pacmpl/LagoG24}.
The model checking problem for an extension of Boolean PCF with algebraic
effects has recently been shown to be
undecidable~\cite{DBLP:journals/pacmpl/LagoG24}, by encoding PCF (with natural numbers).
An alternative proof of the undecidability given below is essentially just a streamlined version
of their encoding, specialized for Minsky machines.

We extend the syntax of \BPCF{} expressions as follows.
\begin{align*}
  M \mbox{ (expressions) }&::= \cdots \mid \Effexp{\alg}{V_1}{V_2}\mid
  \Handleexp{M}{H}\\
  H \mbox{ (handlers) }&::= \set{\rethandler{x}{M_0}, \effhandler{\alg_1}{x}{\contvar}{M_1},
    \ldots,\effhandler{\alg_\ell}{x}{\contvar}{M_\ell}}
\end{align*}
Here, \(\alg, \alg_1,\ldots,\alg_\ell\) ranges over a set of (the names of) algebraic
effects. We assume a type assignment function \(\Alg\), which maps
each algebraic effect \(\alg\) to a type of the form \(\bt_1\to \bt_2\).\footnote{We restrict effects to operations on base types, following \cite{DBLP:journals/pacmpl/LagoG24}. If we allow effects on higher-order values, then
  the undecidability follows immediately since exceptions can be directly encoded.}
The expression \(\Effexp{\alg}{V_1}{V_2}\) applies the algebraic effect operation
\(\alg\) to \(V_1\), where a continuation function \(V_2\) is passed as an additional
parameter; the actual effect of the operation is defined by a surrounding handler
\(H\). The expression \(\Handleexp{M}{H}\) evaluates \(M\) and
the handler \(H\) is invoked when an algebraic effect operation is invoked during
the evaluation. Unlike exception handlers, the control may be returned to
the main expression \(M\), by invoking the continuation.
In \(\Handleexp{M}{H}\), we assume that \(H\) contains a handler for every \(\alg_i\);
this does not lose generality, as we can always add a ``forwarder'' handler
\(\effhandler{\alg_i}{x}{\contvar}{\alg_i(x,\contvar)}\) to ensure this condition.

The typing rules (in addition to those of \BPCF{})
for deep handlers
are given in Figure~\ref{fig:alg-typing}.
(For shallow handlers~\cite{DBLP:conf/aplas/HillerstromL18}, the type of \(c\) in the last rule
should be replaced with \(\bt'_i\to\sty_1\).)
In the figure, \(\Alg\) gives the types of the operand and result of the
algebraic effect operation;
for the sake of simplicity, we assume that they are globally fixed.
For simplicity, we have also omitted the distinction between value
types and computation types, as opposed to standard type systems
for algebraic
effects~\cite{DBLP:journals/corr/PlotkinP13,DBLP:journals/pacmpl/LagoG24}.
A handler is given a type of the form
\(\Thandler{\sty_1}{\sty_2}\); it means that the main
expression is expected to return a value of 
type \(\sty_1\) (unless an algebraic effect occurs),
and the handler transforms it to a value of type \(\sty_2\); see the last rule
in the figure.\footnote{
In the type system of \cite{DBLP:journals/pacmpl/LagoG24}, a handler type also carries information
  about the types of algebraic effects handled and raised by the handler;
  we can omit them thanks to the assumption that they are globally fixed
  as described by \(\Alg\).}
\begin{figure}
  \infrule{\Alg(\alg)=\bt_1\to \bt_2\andalso \STE\p V_1:\bt_1 \andalso
    \STE\p V_2: \bt_2\to \sty}
          {\STE\p \Effexp{\alg}{V_1}{V_2}: \sty}
  \infrule{\STE \p H:\Thandler{\sty_1}{\sty_2}\andalso \STE\p M:\sty_1}
          {\STE \p \Handleexp{M}{H}: \sty_2}
        \infrule{\STE,x\COL\sty_1 \p M_0:\sty_2\\
            \Alg(\alg_i)=\bt_i\to \bt'_i  \andalso
            \STE,x\COL \bt_i,\contvar\COL \bt'_i\to\sty_2\p M_i:\sty_2 \mbox{ (for each $i\in
            \set{1,\ldots,\ell}$)}}
          {\STE \p \set{\rethandler{x}{M_0},
              \effhandler{\alg_1}{x}{\contvar}{M_1},
    \ldots,\effhandler{\alg_\ell}{x}{\contvar}{M_\ell}}: \Thandler{\sty_1}{\sty_2}}
  \caption{Typing Rules for Algebraic Effects}
  \label{fig:alg-typing}
\end{figure}

The operational semantics is extended as given in Figure~\ref{fig:alg-os}.
Here, the rule \rn{R-DH} is used in the case of the deep handler semantics,
and \rn{R-SH} is used in the case of the shallow handler semantics.
The rule \rn{R-Ret} is for the case where the handled expression evaluates
to a value; in this case, the return handler \(M_0\) is invoked.
The rules \rn{R-DH} and \rn{R-SH} are for the case where an algebraic
effect is invoked. In the deep handler semantics, the partial continuation including
the handler \(H\) is passed as the continuation argument,
while in the shallow handler semantics,
the partial continuation without the handler is passed as the continuation argument.

For example, let \(H=\set{\rethandler{x}{x}, \effhandler{\alg}{x}{\contvar}
  {\contvar(\FALSE)\Band\contvar(\TRUE)}}\)
where \(\Band\) is the Boolean AND operation. When an effect \(\alg\) occurs,
the handler calls the continuation twice, with \(\FALSE\) and \(\TRUE\) as arguments.
Then, \(M \Def \Handleexp{\NOT(\alg(\Vunit,\lambda z.z))}{H}\) is reduced as follows in
the shallow handler semantics.
\begin{align*}
  M &\red [\Vunit/x, \lambda y.{\NOT((\lambda z.z)y)}/\contvar](\contvar(\FALSE)\Band\contvar(\TRUE))\\
  &\equiv \NOT(\FALSE)
  \Band \NOT(\TRUE)\reds  \TRUE\Band \NOT(\TRUE) \reds  \FALSE.
\end{align*}
In the case of the deep handler semantics, \(M\) is reduced to
\[[\Vunit/x, \lambda y.\Handleexp{\NOT((\lambda z.z)y)}{H}/\contvar](\contvar(\FALSE)\Band\contvar(\TRUE))\] on the first line, but the final result is the same in this example.

\begin{figure}
  \begin{align*}
    E\mbox{ (evaluation context) }::=
    \cdots \Effexp{\alg}{E}{V}\mid \Handleexp{E}{H}\\
    \end{align*}
  \infax[R-Ret]{E[\Handleexp{V}{\set{\rethandler{x}{M_0},\ldots
}}]\red
    E[[V/x]M_0]}
  \infrule[R-DH]
  {H = \set{
         \ldots,  \effhandler{\alg_i}{x}{\contvar}{M_i},\ldots}}
          {E[\Handleexp{F[\Effexp{\alg_i}{V_1}{V_2}]}
      {H
}]\red
    E[[V_1/x,\lambda y.\Handleexp{F[V_2\,y]}{H}/\contvar]M_i]}
  \infrule[R-SH]
  {H = \set{
         \ldots,  \effhandler{\alg_i}{x}{\contvar}{M_i},\ldots}}
          {E[\Handleexp{F[\Effexp{\alg_i}{V_1}{V_2}]}
      {H
}]\red
    E[[V_1/x,\lambda y.F[V_2\,y]/\contvar]M_i]}
  \caption{Operational Semantics for Algebraic Effects. The meta-variable
    \(F\) denotes an evaluation context without a sub-expression of the form
   \(\Handleexp{E}{H}\) and \(F\) does not capture \(y\).}
  \label{fig:alg-os}
\end{figure}

\subsubsection{Encoding of Minsky machines with shallow handlers}
For algebraic effect with shallow handlers (when the reduction relation is
defined with the rule \rn{R-SH} instead of \rn{R-DH}),
we can encode natural numbers as follows, following
the idea of Dal Lago and Ghyselen~\cite{DBLP:journals/corr/PlotkinP13,DBLP:journals/pacmpl/LagoG24}.\footnote{Actually,
  we can instead encode \(0\) as \(\lambda x.x\) in the case of shallow handlers;
  we chose the current encoding for the sake of uniformity with the encoding for
  deep handlers.}
\begin{align*}
& \enceff{0} = \lambda x.\Effexp{\alg_Z}{x}{\lambda y.y}\qquad
    \enceff{(n+1)}= \lambda x.\Effexp{\alg_S}{x}{\enceff{n}}
\end{align*}
This encoding is essentially the same as the one used in the previous subsection, where
a natural number was encoded as a chain of closures to raise exceptions.

A Minsky machine \(\set{0\mapsto I_0,\ldots,m\mapsto I_m}\) can be encoded as
\begin{align*}
  \letrecexp{f_0\;c_0\;c_1 =M_0 \AND \cdots \AND f_m\;c_0\;c_1 = M_m}{f_0\;\encexn{0}\;\encexn{0}}
\end{align*}
where \(M_i\) is:
\begin{itemize}
\item \(\FAIL\), if \(I_i\) is \(\Mhalt\);
\item \(\letexp{c_j}{\lambda x.\Effexp{\alg_S}{x}{c_j}}{f_k\;c_0\;c_1}\), if \(I_i\) is \(\INC{j};\GOTO{k}\); and
\item \(\Handleexp{c_j\Vunit}{H_i}\) if \(I_i\) is \(\Mif{j}{k}{\ell}\),
  where
  \begin{align*}
    &  H_i = \set{\rethandler{x}{x},\effhandler{\alg_Z}{x}{\contvar}{f_k\;c_0\;c_1},  \effhandler{\alg_S}{x}{c_j}{f_\ell\;c_0\;c_1}}
  \end{align*}
\end{itemize}
In the last case, 
we can check whether the value represented by \(c_j\) is \(0\), by invoking it and checking whether
\(\alg_Z\) or \(\alg_S\) occurs, and let the handler \(H_i\) perform appropriate operations.
If \(\alg_z\) occurs, \(f_k\) is called. If \(\alg_S\) occurs, then \(c_j\) is bound to the continuation argument,
which is the current value of \(c_j\) minus one, and \(f_\ell\) is called.
The return handler in \(H_i\) above will never be used.

From the above observation, it should be clear that the Minsky machine halts, just if the above term of
\BPCFAE{} reaches \(\FAIL\) in the shallow handler semantics. Thus,
the reachability problem for \BPCFAE{} with shallow handlers is undecidable.

\subsubsection{Encoding of Minsky machines with  deep handlers}

In the case of deep handlers, we need a trick (proposed by \cite{DBLP:journals/pacmpl/LagoG24})
to apply the effect handler \(H_i\) only to the first effect.
To this end, we prepare ``primed versions'' of the algebraic effects \(\alg_S\) and \(\alg_Z\),
and convert the representation of a natural number:
\[n^\dagger :=  \lambda x.\alg_S(x; \lambda x.\alg_S(x; \cdots \lambda x.\alg_S(x; \lambda x.\alg_Z(x, \lambda x.x))\cdots))\]
to
\[n^\flat := \lambda x.\alg_{S'}(x; \lambda x.\alg_S(x; \cdots \lambda x.\alg_S(x; \lambda x.\alg_Z(x, \lambda x.x))\cdots)),\]
before calling it to perform case analysis.

A Minsky machine \(\set{0\mapsto I_0,\ldots,m\mapsto I_m}\) is now encoded as
\begin{align*}
  \letrecexp{f_0\;c_0\;c_1 =M_0 \AND \cdots \AND f_m\;c_0\;c_1 = M_m}{f_0\;\encexn{0}\;\encexn{0}}
\end{align*}
where \(M_i\) is:
\begin{itemize}
\item \(\FAIL\), if \(I_i\) is \(\Mhalt\);
\item \(\letexp{c_j}{\lambda x.\Effexp{\alg_S}{x}{c_j}}{f_k\;c_0\;c_1}\), if \(I_i\) is \(\INC{j};\GOTO{k}\); and
\item \(\Handleexp{\mathit{unwrap}\;c_j\;\Vunit}{H_i}\) \\if \(I_i\) is \(\Mif{j}{k}{\ell}\),
  where
  \begin{align*}
    & \mathit{unwrap} = \lambda c.\lambda x.\Handleexp{c\Vunit}{H_{\mathit{unwrap}}}\\
    &  H_i = \set{\rethandler{x}{x},\effhandler{\alg_{Z'}}{x}{\contvar}{f_k\;c_0\;c_1},  \effhandler{\alg_{S'}}{x}{c_j}{f_\ell\;c_0\;c_1}}\\
    &  H_{\mathit{unwrap}} = \set{\rethandler{x}{x},\effhandler{\alg_{Z}}{x}{\contvar}{\Effexp{\alg_{Z'}}{x}
        {\lambda y.\Handleexp{\contvar\,y}{H_{\mathit{wrap}}}}},\\&\qquad\qquad\qquad\qquad\qquad\qquad
          \effhandler{\alg_{S}}{x}{c}{\Effexp{\alg_{S'}}{x}{\lambda y.\Handleexp{\contvar\,y}{H_{\mathit{wrap}}}}}}\\
    &  H_{\mathit{wrap}} = \set{\rethandler{x}{x}, \effhandler{\alg_{Z'}}{x}{\contvar}{\Effexp{\alg_{Z}}{x}{\contvar}}, \effhandler{\alg_{S'}}{x}{c}{\Effexp{\alg_{S}}{x}{\contvar}}}
  \end{align*}
\end{itemize}
The function \(\mathit{unwrap}\) has the effect of converting \(n^\dagger\) to \(n^\flat\).
To that end, it applies the handler \(H_{\mathit{unwrap}}\), which
replaces \(\alg_S\) and \(\alg_Z\) with their primed versions. Furthermore,
inside the handler \(H_{\mathit{unwrap}}\), 
the handler \(H_{\mathit{wrap}}\) is applied to the continuation, which re-replaces primed effects
\(\alg_{S'}\) and \(\alg_{Z'}\) with unprimed versions. Thus, overall,
\(\mathit{unwrap}\) replaces only the outermost effects with their primed versions. The handler \(H_i\) then
performs case analysis on whether the outermost effect is \(\alg_{S'}\) or \(\alg_{Z'}\).

Let us check how the above encoding works for the case \(I_i\) is a conditional command. We have the following reductions for \(f_i\,c_0\, c_1\), where \(c_0= (n+1)^\dagger = \lambda x.\Effexp{\alg_S}{x}{n^\dagger}\).
\begin{align*}
  f_i\,c_0\, c_1 & \red \Handleexp{\mathit{unwrap}\;c_0\;\Vunit}{H_i}\\
  & \red \Handleexp{(\Handleexp{c_0\Vunit}{H_{\mathit{unwrap}}})}{H_i}\\
  & \red \Handleexp{(\Effexp{\alg_{S'}}{\Vunit}{\lambda y.\Handleexp{(\Handleexp{n^\dagger\,y}{H_{\mathit{unwrap}}})}{H_{\mathit{wrap}}}})}{H_i}\\
  &\equiv \Handleexp{\Effexp{\alg_{S'}}{\Vunit}{n^\dagger}}{H_i} \mbox{ (\(H_{\mathit{unwrap}}\) and \(H_{\mathit{wrap}}\) are canceled out)}\\
  &\red f_\ell\;n^\dagger\;c_1.
\end{align*}

Thus, the Minsky machine halts, just if the above term of
\BPCFAE{} with deep handlers reaches \(\FAIL\). Therefore,
the reachability problem for \BPCFAE{} with deep handlers is also undecidable.

\begin{remark}
  We have shown the undecidability of shallow handlers first because it is easier
  to obtain, and then shown that of deep handlers
  by using Dal Lago and Ghyselen's encoding of shallow handlers using deep handlers.
  For the latter, we can alternatively use the encoding by
  Hillerstr{\"{o}}m and Lindley~\cite{DBLP:conf/aplas/HillerstromL18}.
  \end{remark}
 
\subsection{Extension with Symbol Generation}
\label{sec:BPCFL}
We now consider \BPCFL{}, an extension of \BPCF{} with primitives for generating new symbols and
testing the equality of symbols.
The syntax of expressions and types is extended as follows.
\begin{align*}
  M \mbox{ (expressions) }::= \cdots \mid \Gensym{}\mid M_1=M_2 \qquad \qquad
  \tau \mbox{ (types) } ::= \cdots \mid \Tsym.
\end{align*}
We also add the following typing rules.
\begin{multicols}{3}
\infrule{}{\STE\p \Gensym{}:\Tsym}
\infrule{\STE\p M_1:\Tsym \andalso \STE\p M_2:\Tsym}{\STE\p M_1=M_2:\BOOL}
\end{multicols}

The expression \(\Gensym{}\) generates a new symbol, and \(M_1=M_2\) tests the equality of
the two symbols (obtained by evaluating) \(M_1\) and \(M_2\).
\BPCFL{} can be considered a minimal extension of \BPCF{} for writing symbol manipulation programs.
It will also be used in the next subsection to show the undecidability of an extension of \BPCF{} with references.

We extend the reduction semantics accordingly.
To that end, we assume a denumerable set of labels, ranged over by a meta-variable \(\ell\).
We further extend the syntax of expressions by \(M::=\cdots \mid \ell\) (where \(\ell\) may occur only
at run-time). 
The reduction semantics in Figure~\ref{fig:os} is extended with the following rules
(where \(E\) is also extended by \(E ::= \cdots \mid E=M\mid \ell=E\)).
\infax{E[\Gensym{}]\red E[\ell]\qquad \mbox{($\ell$ does not occur in \(E\))}}
\begin{multicols}{2}
  \infax{E[\ell=\ell]\red E[\TRUE]} 
  \infax{E[\ell_1=\ell_2]\red E[\FALSE]\qquad \mbox{(if $\ell_1\ne \ell_2$)}}
\end{multicols}

  To encode a Minsky machine, we represent a number as a symbol,
  and keep and update a function to decrement a number.
  A Minsky machine \(P = \set{0\mapsto I_0,\ldots,m\mapsto I_m}\) is encoded into the following
  term \(P^\dagger\) of \BPCFL{}.
  \newcommand\lzero{\ell_{\mathit{zero}}}
  \newcommand\iszero{\mathit{iszero}}
  \newcommand\inc{\mathit{inc}}
  \newcommand\dec{\mathit{dec}}
  \newcommand\next{\mathit{next}}
  \begin{align*}
    &    \letexp{\lzero}{\Gensym{}}\letexp{\iszero }{\lambda x.x=\lzero}\letexp{\dec}{\lambda x.\lzero}\\
&    \letexp{\inc}{\lambda (x,d).\letexp{y}{\Gensym{}}{(y, \lambda z.\ifexp{z=y}{x}{d(z)})}}\\
    &  \letrecexp{f_0\;d\;c_0\;c_1 =M_0 \AND \cdots \AND f_m\;d\;c_0\;c_1 = M_m}\\
    & {f_0\;\dec\;\encexn{0}\;\encexn{0}}
\end{align*}
where \(M_i\) is:
\begin{itemize}
\item \(\FAIL\), if \(I_i\) is \(\Mhalt\);
\item \(\letexp{(c_j,d)}{\inc(c_j,d)}{f_k\;d\;c_0\;c_1}\), if \(I_i\) is \(\INC{j};\GOTO{k}\); and
\item \(\ifexp{\iszero\;{c_j}}{f_k\;d\;c_0\;c_1}{\letexp{c_j}{d(c_j)}{f_n\;d\;c_0\;c_1}}\) \\ if
    \(I_i\) is \(\Mif{j}{k}{n}\).
\end{itemize}

The program above first creates a symbol \(\lzero\) which represents \(0\), and defines the function \(\iszero\) accordingly.
Initially, the function \(\dec\) to decrement a number is set to the constant function that always returns (the symbol representing) \(0\),
because at this moment, the only symbol
that represents a number is \(\lzero\), and \(0-1\) is defined as \(0\).
The function \(\inc\) takes a pair consisting of the number \(x\) to increment, and the current decrement function \(d\)
as an argument, and returns a pair consisting of a symbol that represents \(x+1\) and the new decrement function \(d'\).
As the symbol for  \(x+1\), a fresh symbol is always generated even if \(x+1\) has been computed before,
and \(d'\) remembers \(x\) as the predecessor of the created symbol. Thus, except for \(0\), there may be more than
one symbol that represents the same number; that is fine, as the only zero-test is performed in a Minsky machine
to compare numbers.

The function \(f_i\), corresponding to the program point \(i\) of the Minsky machine, now takes the current decrement function \(d\)
as an additional argument. When \(I_i\) is \(\INC{j};\GOTO{k}\), then \(M_i\) calls
the function \(\inc\) to obtain \(c_j+1\) and the updated decrement function, and proceeds to call \(f_k\).
In the last case, \(\iszero\) is called to check whether \(c_j\) represents \(0\), and if not, the current decrement function \(d\) is called.

Based on the intuition above, it should be clear that a Minsky machine \(P\) halts if and only if \(P^\dagger\reds \FAIL\). Thus, we have:
\begin{theorem}
\label{th:BPCFL-is-undecidable}
The reachability problem for \BPCFL{} is undecidable.
\end{theorem}

 \subsection{Extension with Boolean References}
\label{sec:BPCFR}
Let us now consider \BPCFR{}, an extension of \BPCF{} with Boolean references.
The syntax of expressions and types is extended as follows.
\begin{align*}
  M ::= \cdots \mid \mkref{M}\mid \deref{M}\mid M_1:=M_2 \qquad\qquad
  \tau ::= \cdots \mid \Tref{\BOOL}.
\end{align*}
We add the following typing rules.
\begin{multicols}{3}
\infrule{\ \\\STE\p M:\BOOL}{\STE\p \mkref{M}:\Tref{\BOOL}}
\infrule{\ \\\STE\p M:\Tref{\BOOL}}{\STE\p \deref{M}:{\BOOL}}
\infrule{\STE\p M_1:\Tref{\BOOL}\\ \STE\p M_2:\BOOL}{\STE\p M_1:=M_2:{\Tunit}}
\end{multicols}
We omit (standard) rules for reductions.

The reachability problem for \BPCFR{} is also undecidable.
To see it, it suffices to observe that \BPCFL{} can be encoded into \BPCFR{}.
The types and terms of \BPCFL{} can be encoded as follows.
\begin{align*}
&\Tsym^\dagger = \Tref{\BOOL} \qquad (\Gensym{})^\dagger = \mkref{\TRUE}\\
  &(x=y)^\dagger = x:=\FALSE; y:=\TRUE; !x
\end{align*}
A symbol is represented as a Boolean reference, and the generation of a fresh symbol is emulated by the creation of
a new reference cell.
To test the equality between two references \(x\) and \(y\),
it suffices to change the value of \(x\) and test whether the contents of the values of \(x\) and \(y\) remain the
same.\footnote{This idea for encoding the equality test can be at least traced back to \cite{10.5555/309656.309671}.}

As a corollary of Theorem~\ref{th:BPCFL-is-undecidable}, we have:
\begin{theorem}
\label{th:BPCFR-is-undecidable}
The reachability problem for \BPCFR{} is undecidable.
\end{theorem}

\begin{remark}
  Conversely, we can encode \BPCFR{} into \BPCFL{}.
  To that end, it suffices to represent a reference as a symbol of \BPCFL{},
  and apply a store-passing transformation, where 
  a store can be represented as a function that maps symbols (representing references) to
  their current values.
  Thus, primitives on references can be encoded as follows.
  \begin{align*}
&  (\mkref{M})^\dagger =  \lambda s.
\letexp{(x,s')}{M^\dagger\,s}\letexp{y}{\Gensym{}} {(y, \lambda z.\ifexp{y=z}{x}{s'(z)})}\\
&(!M)^\dagger =  
\lambda s.
\letexp{(x,s')}{M^\dagger\,s}{(s'(x), s')}\\
&  (M_1 := M_2)^\dagger =  \lambda s.
\letexp{(x,s')}{M_1^\dagger\,s}\letexp{(y,s'')}{M_2^\dagger\,s'}
\\&\qquad\qquad\qquad\qquad\qquad\qquad\qquad
{(\Vunit, \lambda z.\ifexp{x=z}{y}{s''(z)})}.  \qquad\qquad\qquad\qquad\qed
\end{align*}
\end{remark}
  
\section{Decidable References}
\label{sec:reference}
We now focus on \BPCFR{}, and introduce a decidable fragment of it, called \BPCFRL{}.
Actually, we introduce a language \bPCFRL{\bty} parameterized by a base type \(\bty\in\set{\BOOL,\INT}\);
\BPCFRL{} is just \bPCFRL{\BOOL}.
Before formally defining \bPCFRL{\bty}, we first give an overview of the ideas in Section~\ref{sec:overview}.
We then define \bPCFRL{\bty} in Section~\ref{sec:typing}, and show,
in Section~\ref{sec:trans},
a type-based translation from \bPCFRL{\bty} to \bPCF{\bty}, where \bPCF{\bty} is PCF parameterized by
the base type \(\bty\); \bPCF{\BOOL} is \BPCF{}, and \bPCF{\INT} is the ordinary PCF (with integers).
The correctness of the translation for the case \(\bty=\BOOL\) yields the decidability of \BPCFRL{},
and the case \(\bty=\INT\) yields an automated verification technique as discussed in Section~\ref{sec:exp}.
We discuss applications and further extensions of \bPCFRL{\bty} in Section~\ref{sec:disc},
and report preliminary experimental results on an application to automated verification
in Section~\ref{sec:exp}.
\subsection{Overview}
\label{sec:overview}
To recover the decidability, we impose the following restrictions by a type system.
\begin{asparaenum}
\item As in Rust~\cite{DBLP:journals/cacm/JungJKD21} and Consort~\cite{DBLP:conf/esop/TomanSSI020}, 
  we ensure that when a mutable reference exists,
  there are no other references for the same memory cell at each point of time during the evaluation.
  For example, \(M_{\Ok_1}\) given below is allowed, but not \(M_{\Ng_1}\).
  \begin{align*}
  &M_{\Ok_1} \Def \letexp{x}{\mkref{\TRUE}}\letexp{y}{x} (y:=\NOT(\deref{y});!y)\\
    & M_{\Ng_1}\Def\letexp{x}{\mkref{\TRUE}}\letexp{y}{x} (y:= \NOT(\deref{x});!x).
    \end{align*}
  In the latter, both \(x\) and \(y\) point to the same reference cell at the same time, which violates the restriction.
  In contrast, in the former expression, only \(y\) is used in the body of the let-expression (although it is also in the lexical
  scope of \(x\)).

  This restriction allows us to represent a reference as the value stored in the cell pointed to by
the reference.\footnote{In RustHorn~\cite{DBLP:journals/toplas/MatsushitaTK21},
a reference is represented as a pair consisting of the current value stored in the reference
cell and the future value; since we do not consider borrowing, we need not consider the future
value.}
For example, \(M_{\Ok_1}\) above can be translated to:
  \[ \letexp{x}{{\TRUE}}\letexp{y}{x} (\letexp{y}{\NOT({y})}y).\]
  Notice that \(\mkref{\TRUE}\) has just been replaced by \(\TRUE\), the value of the reference cell being created,
  and that the update operation \(y:=\NOT(\deref{y})\) has been replaced by the let-expression
  \(\letexp{y}{\NOT(y)}\cdots\). At the point of executing \(y:=\NOT(\deref{y})\), there is no other active
  reference for the cell pointed to by \(y\) (the use of \(x\) is disallowed); thus it is safe to replace the reference \(y\)
  with the value stored in the reference cell. Note that if we apply the same transformation to the invalid program
  \(M_{\Ng_1}\), then we would obtain a wrong program:
  \[\letexp{x}{{\TRUE}}\letexp{y}{x} (\letexp{y}{\NOT(x)}{x}),\]
  which would return \(\TRUE\) instead of \(\FALSE\).

\item We allow a closure to contain references and other closures, but such a closure
  is also treated linearly; there cannot be more than one variable that holds the closure
  at each point of time.
  For example, \(M_{\Ok_2}\) given below is allowed, but not \(M_{\Ng_2}\).
  \begin{align*}
& M_{\Ok_2}\Def\letexp{x}{\mkref{\TRUE}}\letexp{f}{\lambda z.(x:=\NOT(\deref{x});!x)}{(f();f())}\\&
 M_{\Ng_2}\Def\letexp{x}{\mkref{\TRUE}}\letexp{f}{\lambda z.(x:=\NOT(\deref{x});!x)}\letexp{g}{f}{(f();g())}.
    \end{align*}
  In the former, a closure containing a reference cell is created, and it is exclusively referred to by \(f\).
  In the latter, the same closure is also referred to by \(g\), and called through both \(f\) and \(g\),
  which violates the restriction.
  The expression
  \begin{align*}
   &\qquad M_{\Ok_3}\Def\letexp{f}{\lambda z.(\letexp{x}{\mkref{\TRUE}}x:=\NOT(\deref{x});!x)}\letexp{g}{f}{(f();g())}.
    \end{align*}
    is allowed, because the closure does not contain any reference cells, although a reference cell is
    created when the closure is called.
    As in the following expression, a closure may also capture other closures.
\begin{align*}
M_{\Ok_4}\Def\, &\letexp{x}{\mkref{\TRUE}}  
\letexp{y}{\mkref{\TRUE}}\\  
&\letexp{f}{\lambda z.(x := \NOT(x); !x)} \letexp{g}{\lambda u.f()\Band !y}{(g();g())}.
\end{align*}
The expression \(M_{\Ng_4}\) obtained by
replacing \(g();g()\) with \(f();g()\), however, violates the restriction, as \(f\) is already owned by \(g\), so that
\(f\) cannot be used in the main body of the let-expression.

The above restriction allows us to 
represent a closure as a pair consisting of the store,
which expresses the values of the reference cells captured by the closure,
and the code of the closure.
For example, \(M_{\Ok_2}\) can be translated to the following \BPCF{} program (extended with
tuples).
\begin{align*}
&  \letexp{x}{\TRUE}
  \letexp{f}{\Clos{\lambda (y,x).\letexp{x'}{\NOT(x)}{x'}}{x}}\\&
  \letexp{(\_,\envof{f})}{\codeof{f}(\Vunit,\envof{f})}\codeof{f}(\Vunit,\envof{f})
\end{align*}
Here, \(f\) is represented as a pair consisting of a store holding the current value of
\(x\), and code that takes an argument for the original function \(f\) and the value of \(x\);
for readability, we write \(\Clos{\cdot}{\cdot}\) instead of \((\cdot,\cdot)\), for a pair
representing a closure, and write \(\envof{f}\) and \(\codeof{f}\) respectively for
the first element (which is the store captured by the closure) and the second element (which is the code).
We treat \(\envof{f}\) and \(\codeof{f}\) as if they were variables,
and write \(\letexp{(\_,\envof{f})}{M_1}M_2\) for
\(\letexp{(\_,h)}{M_1}\letexp{f}{\Clos{\codeof{f}}{h}}M_2\).
The code part of function \(f\) now takes a pair consisting of the original argument and
the current store, and returns a pair consisting of a return value and an updated store.
In translating \(f();f()\), 
the first call \(f()\) has been replaced by \(\codeof{f}(\Vunit,\envof{f})\), and for the second call,
the store part of \(f\) is updated.
So, our translation is a kind of store-passing translation, but only a local store is passed upon
a function call, instead of the global store.

If we apply the same translation to the invalid term \(M_{\Ng_2}\), we would obtain:
\begin{align*}
&  \letexp{x}{\TRUE}
  \letexp{f}{\Clos{\lambda (y,x').\letexp{x''}{\NOT(x')}{x''}}{x}}\letexp{g}{f}\\&
  \letexp{(\_,\envof{f})}{\codeof{f}(\Vunit,\envof{f})}\codeof{g}(\Vunit,\envof{g}),
\end{align*}
which would return a wrong value \(\FALSE\). This is because the store is duplicated
by \(\letexp{g}{f}\cdots\), and the update of the value of \(x\)
by the call of \(f\) is not reflected to \(g\).

In contrast, the program \(M_{\Ok_3}\) is fine as the closure \(f\) does not contain a reference cell. It is translated to:
\begin{align*}
&  
  \letexp{f}{\Clos{\lambda (z,\_).\letexp{x}{\TRUE}\letexp{x'}{\NOT(x)}{x'}}{()}}\letexp{g}{f}\\&
  \letexp{(\_,\envof{f})}{\codeof{f}(\Vunit,\envof{f})}\codeof{g}(\Vunit,\envof{g}),
\end{align*}
which correctly returns \(\FALSE\). Here, the duplication of \(f\)'s store by \(\letexp{g}{f}\cdots\)
is not problematic as the duplicated local store for \(f\) is empty, containing no information.

  \item A closure is allowed to capture only a statically bounded number of reference cells.
    The functions \(f\) and \(g\) in \(M_{\Ok_4}\) above are fine, as they can be
    statically inferred to contain one and two references respectively; we will embed that information into
    function types, and assign types \(\Tfuns{\UNIT}{}{1}{\BOOL}\) and \(\Tfuns{\UNIT}{}{2}{\BOOL}\)
    respectively to \(f\) and \(g\).
    In contrast, the decrement function \(\dec\) in the encoding of a Minsky machine in Section~\ref{sec:BPCFR}
    (obtained indirectly through the encoding into \BPCFL{} in Section~\ref{sec:BPCFL}) 
    may contain an unbounded number of reference cells, hence violating the restriction.

    This restriction ensures that the local store of a closure can be represented as a tuple of
a fixed size, so that the resulting \BPCF{} program is well-typed.
For example, \(M_{\Ok_4}\) above can be translated to:
\begin{align*}
&\letexp{x}{\TRUE}
\letexp{y}{\TRUE}\\
&\letexp{f\COL \BOOL\times (\UNIT\times \BOOL\to \BOOL\times \BOOL)\\&\qquad
}{\Clos{\lambda (z,x).(\letexp{x'}{\NOT(x)}(x',x'))}{x}} \\&
\letexp{g\COL\BOOL^2\times (\UNIT\times \BOOL^2 \to \BOOL\times \BOOL^2)\\&\quad }{\Clos{\lambda (z, (x,y)).\letexp{(r,x')}{\codeof{f}((),x)}(r \Band y, (x',y))}{(\envof{f},y)}}\\&
  \letexp{(\_,\envof{g})}{\codeof{g}(\Vunit,\envof{g})}\codeof{g}(\Vunit,\envof{g}).
\end{align*}

For clarify, we have annotated \(f\) and \(g\) with their types.
In \(M_{\Ok_4}\), the function \(g\) contains a reference cell \(y\) and a closure \(f\), which in turn contains
a reference cell \(x\). Thus, \(g\) contains two reference cells in total.
We thus encode \(g\) into a \BPCF{} term
of type \(\BOOL^2\times (\UNIT\times \BOOL^2 \to \BOOL\times \BOOL^2)\),
where the first component represents the current values of the two reference cells,
and the second component represents a function that takes an argument and the current values of
the reference cells, and returns a pair consisting of the output of \(g\) and the updated values
of the reference cells.
When \(g\) calls \(f\), it passes to \(f\) a part of \(g\)'s store (i.e., \(x\)), which represents the store of \(f\), 
and updates it after the call.
If we do not bound the number of reference cells contained in a closure,
then we would need recursive types (like \(\BOOL\;\mathtt{list}\)) to represent the  store of a closure,
disabling the encoding into \BPCF{}.

\end{asparaenum}
We enforce the restrictions above through a type system.
To make the type system as simple as possible,
we consider neither lifetime and borrowing mechanisms of Rust~\cite{DBLP:journals/pacmpl/0002JKD18,DBLP:journals/toplas/MatsushitaTK21}, nor
fractional ownerships of Consort~\cite{DBLP:conf/esop/TomanSSI020}.

As explained above, our translation is a kind of store-passing transformation, which, instead of maintaining
a global store, has a function closure maintain its own part of the store.
To make that possible, we restrict copying of closures (by restriction (2)), and the number of reference cells
owned by each closure (by restriction (3)).

\begin{remark}
One may consider the restrictions above too strong,
but the Rust language imposes a restriction similar to (2) for
a closure with a ``unique immutable borrow''\footnote{\url{https://doc.rust-lang.org/reference/types/closure.html}.}. As for (3), even if the size of the
store captured by a closure in a source program (which may use integers and other data structures)
is unbounded, one may apply predicate abstraction to ensure
that the resulting program of \BPCFR{} satisfies the restriction (3). \qed
\end{remark}

\subsection{Syntax and Type System of \bPCFRL{\bty}}
\label{sec:typing}
We now formalize a type system that enforces the restrictions stated above.

We first define the syntax of the language \bPCFRL{\bty} parameterized by the base type \(\bty\in\set{\BOOL,\INT}\).
\bPCFRL{\BOOL} is also called \BPCFRL{}.
The syntax of \bPCFRL{\bty} is given by:
\begin{align*}
  M ::=\,& c \mid x \mid \Op{k}(x_1,\ldots,x_k)\mid\, \deref{x} \mid \mkref{x} \mid x:=y \mid \letexp{x}{M_1}{M_2} \\
   \mid\,& \ifexp{x}{M_1}{M_2} \mid \lambda x.M \mid \fixexp{f}{x}{M}\mid f\,x
\end{align*}
Here, the meta-variable \(c\) denotes either the unit value \(\Vunit\) or a constant of type \(\bty\). We sometimes write \(c^\bty\) for a constant of type \(\bty\).
The symbol \(\Op{k}\) denotes a \(k\)-ary operation on constants of type \(\bty\).
In \bPCFRL{\INT}, \(0\) is treated as \(\FALSE\) and non-zero integers are
treated as \(\TRUE\) in the conditional expression
\(\ifexp{x}{M_1}{M_2}\).
Note that in \bPCFRL{\bty}, each value must first be assigned to a variable before being used.
Ignoring the type system introduced below, \bPCFRL{\BOOL} (without nested references) has the same expressive power as \BPCF{};
any \BPCF{} expression can be normalized to a (possibly ill-typed) \bPCFRL{\BOOL} expression.

The sets of types and type environments are inductively defined by:
\begin{align*}
  &  \rty \mbox{ (base and reference types) } ::= \bty \mid \Tref{\rty}\\ &
  \nty \mbox{ (normal types) }::= \UNIT\mid \rty\mid \Tfuns{\nty_1}{}{n}{\nty_2}\\&
  \fty \mbox{ (full types) }::= \nty \mid  \Trfun{\nty_1}{}{\Delta}{\nty_2} \\  &
  \NTE \mbox{ (normal type environments) }::= x_1\COL\nty_1,\ldots,x_k\COL\nty_k\\&
  \FTE \mbox{ (full type environments) }::= x_1\COL\fty_1,\ldots,x_k\COL\fty_k.
\end{align*}
Here, \(n\) and \(k\) range over the set of non-negative integers.
As shown above, types are stratified into three levels.
The meta-variable \(\rty\) denotes the base type \(\bty\) or
(nested) reference types.
The type \(\Tfuns{\nty_1}{}{n}{\nty_2}\) represents a function closure containing \(n\) reference cells.
We have another kind of function type \(\Trfun{\nty_1}{}{\NTE}{\nty_2}\), where
\(\NTE\) describes reference cells captured by the closure; this type will be used in typing recursive functions
and explained later. The stratification of types ensures that
\(\Trfun{\nty_1}{}{\Delta}{\nty_2}\) occurs only at the top-level,
disallowing types like \((\Trfun{\nty_1}{}{\NTE}{\nty_2})\to \BOOL\)
and \(\Trfun{\nty_1}{}{f\COL\Trfun{\nty_1'}{}{\NTE}{\nty_2'}}{\nty_2}\).
We consider a type environment \(\Gamma\) as a sequence (rather than a set), and require that \(x_1,\ldots,x_k\) are mutually different.
We do not have dependent types, but the representation of a type environment as a sequence allows us to represent
an environment conforming to the type environment as an (ordered) tuple. We write \(\emptyTE\) for the empty type environment,
and write \(\Gamma(x)=\tau\) if \(x\COL\tau\) occurs in \(\Gamma\).
Full types (full type environments, resp.) are often called just types (type environments, resp.). 

We consider a type judgment of the form \(\Gamma\p M:\tau\q \Gamma'\), where
 \(\Gamma\) and \(\Gamma'\)
(which we call pre- and post-type environments respectively) describe the type of each variable before and
after the evaluation of \(M\) respectively. For example, we have
\[x\COL \Tref{\BOOL}, y\COL\Tref{\BOOL} \p \letexp{z}{x}{z}:\Tref{\BOOL} \q y\COL\Tref{\BOOL}.\]
Note that the post-type environment does not contain the binding for \(x\). This is because a copy of the reference \(x\)
has been created and returned as the final value. Thus we no longer allow \(x\) to be used; recall the first restriction
listed at the beginning of this section.
Actually, we also have \(x\COL \Tref{\BOOL}, y\COL\Tref{\BOOL} \p x:\Tref{\BOOL} \q y\COL\Tref{\BOOL}\),
as a copy of \(x\) is being created by the expression. In contrast, we have
\(x\COL {\BOOL}, y\COL\Tref{\BOOL} \p x:{\BOOL} \q x\COL\BOOL, y\COL\Tref{\BOOL}\);
the type binding for \(x\) remains in the post-type environment, as Boolean values can be shared.
In general, the presence of a binding \(x\COL\fty\) means that there is an ownership for using \(x\) as a value of type \(\fty\),
and its absence means there is no such ownership. Some type bindings in the pre-type environment \(\Gamma\)
may disappear in the post-type environment \(\Gamma'\) as in the above example, and \(\Gamma\) and \(\Gamma'\) are the same
except that. We can consider more flexible ownership types where ownerships in \(\Gamma\) may increase or decrease partially,
as in Rust and Consort~\cite{DBLP:conf/esop/TomanSSI020}, but we avoid them to keep the formalization simple; such extensions will be discussed in Section~\ref{sec:disc}.
We sometimes write \(\p M:\tau\q \) for \(\emptyTE\p M:\tau\q \emptyTE\).
As a well-formed condition of a type judgment \(\Gamma\p M:\tau\q \Gamma'\), we require
that,  whenever a type binding of the form
\(f\COL\Trfun{\nty_1}{}{\Delta}{\nty_2}\) occurs in \(\Gamma\),
\(\Delta\subseteq \Gamma\) holds, i.e.,  \(\Delta\) is a subsequence of \(\Gamma\).

The typing rules for \bPCFRL{\bty} are shown in Figure~\ref{fig:typing-BPCFRL}.
We explain main rules below, while also introducing notations used in the figure.
The rules \rn{T-Unit} and \rn{T-Const} for constants are straightforward; the type
environment does not change before and after the evaluation.
In the rule \rn{T-Var} for variables, \(\Gamma\exc{x}\) denotes the type environment
obtained by removing \(x\) if \(\Gamma(x)\) is the type of a non-sharable value. More precisely,
\(\Gamma\exc{x}\) is defined by:
\begin{align*}
  & \emptyTE\excr{x}{\fty} \,=\, \emptyTE\qquad
(\Gamma,y\COL\fty)\exc{x}\,=\, \left\{\begin{array}{ll}
  \Gamma & \mbox{if $x=y$ and $\neg\Sharable(\fty)$}\\
  \Gamma,x\COL\fty & \mbox{otherwise}
  \end{array}\right.
\end{align*}
Here, \(\Sharable(\fty)\) if \(\fty\) is the type of a sharable value, i.e., if 
$\fty$ is a base type, or of the forms \(\Tfuns{\tau_1}{}{0}{\tau_2}\).
Thus, we have 
\[x\COL \Tref{\BOOL}, y\COL\Tref{\BOOL} \p x:\Tref{\BOOL} \q y\COL\Tref{\BOOL},\]
but not 
\[x\COL \Tref{\BOOL}, y\COL\Tref{\BOOL} \p x:\Tref{\BOOL} \q x\COL\Tref{\BOOL},y\COL\Tref{\BOOL},\]
for the reason explained already.
We also use a variation \(\Gamma\excr{x}{\fty}\) of \(\Gamma\exc{x}\), which
denotes the type environment obtained from \(\Gamma\)
by removing the binding on \(x\) if \(\neg\Sharable(\fty)\), and
\(\Gamma\) if \(\Sharable(\fty)\).

\begin{figure}
\typicallabel{T-Var}
\begin{multicols}{2}
\infrule[T-Fail]{}{\Gamma \p \FAIL:\bt\q \Gamma}
\infrule[T-Unit]{}{\Gamma \p \Vunit:\UNIT\q \Gamma}
\infrule[T-Const]{}{\Gamma \p c^\bty:\bty\q \Gamma}
  \infrule[T-Var]{\Gamma(x)=\tau}
          {\Gamma\p x:\tau\q \Gamma\exc{x}}
\rulesp
  \infrule[T-Op]{\Gamma(x_i)=\bty\mbox{ for each $i\in\set{1,\ldots,k}$}}
  {\Gamma\p \Op{k}(x_1,\ldots,x_k):\bty\q\Gamma}
  \rulesp
\infrule[T-Deref]{\Gamma(x)=\Tref{\rty}}
        {\Gamma \p !x:{\rty} \q \Gamma\excr{x}{\rty}}
  \rulesp
\infrule[T-Mkref]{\Gamma(x)=\rty}
        {\Gamma \p \mkref{x}: \Tref{\rty}\q \Gamma\exc{x}}

  \rulesp
  \infrule[T-Assign]{\Gamma(y)=\Tref{\rty}\andalso \Gamma(x)=\rty}
        {\Gamma\p y:=x:\Tunit\q\Gamma\exc{x}}

        \rulesp
 \end{multicols}
\infrule[T-Let]{\Gamma \p M_1: \tau' \q \Gamma''\andalso
  \Gamma'', x\COL\tau'\p M_2:\tau\q \Gamma'}
        {\Gamma \p \letexp{x}{M_1}{M_2}: \tau\q \Gamma'\setminus{x}}
  \rulesp
  \infrule[T-If]
      {\Gamma(x)=\bty\andalso
      \Gamma\p M_1:\tau\q \Gamma'
      \andalso \Gamma\p M_2:\tau\q \Gamma'}
      {\Gamma\p \ifexp{x}{M_1}{M_2}:\tau\q \Gamma'}
      \rulesp
  \infrule[T-Fun]
      {
        \Delta, x\COL\tau_1\p M:\tau_2\q \Delta,x\COL\tau_1
    \andalso (\Delta,\Gamma') = \splitTE{\Gamma}{\FV(\lambda x.M)}
            }
      {\Gamma \p \lambda x.M:\Tfuns{\tau_1}{\tau_1'}{|\Delta|}\tau_2\q
              \Gamma'}
\rulesp
  \infrule[T-RFun]
      {\Delta, f\COL\Trfun{\tau_1}{\tau_1'}{\Delta}{\tau_2}, x\COL\tau_1\p
    M: \tau_2\q \Delta,
    f\COL\Trfun{\tau_1}{\tau_1'}{\Delta}{\tau_2}, x\COL\tau_1\andalso
    (\Delta,\Gamma')=\splitTE{\Gamma}{\FV(\fixexp{f}{x}{M})}}
      {\Gamma\p
        \fixexp{f}{x}{M}:\Tfuns{\tau_1}{\tau_1'}{|\Delta|}{\tau_2}\q \Gamma'}
\rulesp
  \infrule[T-App]
     {\Gamma(f)=\Tfuns{\tau_1}{\tau_1'}{m}{\tau_2}\andalso \Gamma(x)=\tau_1}
     {\Gamma
       \p f\,x:\tau_2\q \Gamma}

  \rulesp
  \infrule[T-RApp]
      {\Gamma(f)=\Trfun{\tau_1}{\tau_1'}{\Delta}{\tau_2}\andalso \Gamma(x)=\tau_1
      \andalso \Delta \subsetTE \Gamma\setminus\set{f,x}}
     {\Gamma
       \p f\,x:\tau_2\q \Gamma}

\caption{Typing Rules for \bPCFRL{\bty}}
\label{fig:typing-BPCFRL}
\end{figure}

In the rule \rn{T-Deref} for dereference, the binding on \(x\) is removed from
the post-type environment if \(\neg\Sharable(\rty)\), to avoid the duplication
of the ownership for the contents of \(x\). For example,
if \(\rty=\Tref{\bty}\), then we have
\(x\COL\Tref{\rty}\p !x:\rty\q \epsilon\), but not
\(x\COL\Tref{\rty}\p !x:\rty\q x\COL\Tref{\rty}\).
The latter is problematic because the value of \(!!x\) would be accessible 
through both the values of \(!x\) and \(x\).
Similarly, in \rn{T-Mkref}, and \rn{T-Assign}, the type binding on \(x\)
may be removed depending on the type \(\rty\). These rules are rather restrictive
for nested references; the restriction can be relaxed by
the extension with the borrowing mechanism discussed in Section~\ref{sec:disc}.

The rules \rn{T-Let} and \rn{T-If} are almost standard, except that, in rule \rn{T-Let},
the post-type environment of \(M_1\) is used for typing \(M_2\).
In \rn{T-Let}, \(\Gamma\setminus{x}\) denotes the type environment obtained by removing
the binding on \(x\) if there is any. 
Before explaining the typing rules for functions, we give some examples of typing.
\begin{example}
  Let us reconsider \(M_{\Ok_1}\) in Section~\ref{sec:overview}.
  It is typed as follows.
  \[\footnotesize
  \infer{\p M_{\Ok_1}:\BOOL\q }
        {\p \mkref{\TRUE}:\Tref{\BOOL}\q &
          \infer{x\COL\Tref{\BOOL}\p \letexp{y}{x} (y:=\NOT(\deref{y});!y):\BOOL\q }
                {x\COL\Tref{\BOOL}\p x:\Tref{\BOOL}\q &
                  \infer{y\COL\Tref{\BOOL}\p y:=\NOT(\deref{y});!y:\BOOL\q y\COL\Tref{\BOOL}}{\cdots}
                } 
          }
        \]
        Here, \(\NOT(\deref{y})\) abbreviates 
        \(\letexp{u}{!y}{\NOT(u)}\).
        Notice that in the body of \(\letexp{y}{x}\cdots\), only the type binding on \(y\) (i.e., \(y\COL\Tref{\BOOL}\))
        is available.
        Therefore, \(M_{\Ng_1}\) is not typable, because \(y:= \NOT(\deref{x});!x\) is not typable under
        \(y\COL\Tref{\BOOL}\). \qed
\end{example}

We now explain the remaining typing rules for functions.
In the rule \rn{T-Fun} for creating a function, 
\(|\FTE|\) is defined by:
\begin{align*}
  &  |\FTE| = \sum_{x\COL\fty} |\fty|
  \qquad |\UNIT|=|\bty|=|\Trfun{\tau_1}{\tau_1'}{\NTE}{\tau_2}|=0 \qquad |\Tref{\rty}|=\Tsize{\rty}\qquad
  |\Tfuns{\nty_1}{}{n}{\nty_2}|=n\\&
  \Tsize{\UNIT}=0\qquad \Tsize{\bty}=1\qquad \Tsize{\Tref{\rty}}=\Tsize{\rty}.
\end{align*}
It represents the size of the store for representing an environment that respects \(\NTE\).
As defined above, it is the sum of the size of values stored directly (in reference cells)
or indirectly (though closures). For example,
\(|x\COL\Tref{\Tref{\bty}}, y\COL\Tfuns{\nty_1}{}{2}{\nty_2}, z\COL\bty| = 1+2+0=3\).
\(\FV(M)\) denotes the set of free variables in \(M\), and \(\splitTE{\Gamma}{S}\), defined below,
splits the type environment \(\Gamma\) into one on the variables in \(S\), and the remaining type environment.
\begin{align*}
 \splitTE{\emptyTE}{S} &= (\emptyTE,\emptyTE)\\
   \splitTE{\Gamma,x\COL\fty}{S} &= \left\{
  \begin{array}{ll}
     ((\Gamma_1,x\COL\fty), \Gamma_2) & \mbox{if $x\in S$ and \(\neg\Sharable(\fty)\)}\\
     ((\Gamma_1,x\COL\fty), (\Gamma_2,x\COL\fty)) & \mbox{if $x\in S$ and \(\Sharable(\fty)\)}\\
     (\Gamma_1, (\Gamma_2,x\COL\fty)) & \mbox{if $x\not\in S$}
  \end{array}
  \right.\\&\qquad
  \mbox{(where \((\Gamma_1,\Gamma_2)=\splitTE{\Gamma}{S}\))}
\end{align*}
For example, \(\splitTE{x\COL\BOOL,y\COL\Tref{\BOOL},z\COL\BOOL}{\set{x,y}} = ((x\COL\BOOL,y\COL\Tref{\BOOL}), (x\COL\BOOL,z\COL\Tref{\BOOL}))\).
Using \(\texttt{splitTE}\), the rule \rn{T-Fun} splits the current type environment \(\Gamma\) into
the type environment \(\NTE\) required for typing the function, and the rest of the type environment \(\FTE'\).
We then type the body of \(M\) under \(\NTE,x\COL\tau_1\). The function has type \(\Tfuns{\tau_1}{\tau_1'}{|\NTE|}\tau_2\),
where \(|\NTE|\) gives the size of the store captured by the closure.\footnote{We could instead assign
\(\Tfuns{\tau_1}{\tau_1'}{n}\tau_2\) for any \(n\ge |\NTE|\); that point will be discussed later
when discussing the translation rule.}

\begin{example}
  Let us recall \(M_{\Ok_2}\). The body of the definition of \(f\) is typed as follows.
  \[
  \infer{  x\COL\Tref{\BOOL}\p\lambda z.(x:=\NOT(\deref{x});!x):
    \Tfuns{\UNIT}{}{1}{\BOOL}\q \emptyTE}
   {\infer{x\COL\Tref{\BOOL},z\COL\UNIT\p x:=\NOT(\deref{x});!x:\BOOL\q
        x\COL\Tref{\BOOL},z\COL\UNIT}{\cdots}}
   \]
   The post-type environment in the conclusion ensures that \(x\) is no longer directly accessed,
   as it is owned by the function \(f\). As \(f\) has type 
   \(\Tfuns{\UNIT}{}{1}{\BOOL}\), it cannot be freely copied.
   In contrast, \(f\) in \(M_{\Ok_3}\) is typed as follows.
\[
\infer{  
\emptyTE  \p \lambda z.(\letexp{x}{\mkref{\TRUE}}x:=\NOT(\deref{x});!x):
  \Tfuns{\UNIT}{}{0}{\BOOL}\q \emptyTE}
{\infer{  z\COL\UNIT
  \p \letexp{x}{\mkref{\TRUE}}x:=\NOT(\deref{x});!x:\BOOL\q
  z\COL\UNIT}
  {z\COL\UNIT,x\COL\Tref{\BOOL}
  \p x:=\NOT(\deref{x});!x:\BOOL\q
  z\COL\UNIT,x\COL\Tref{\BOOL}}
  }
\]
As \(f\) has type \(\Tfuns{\UNIT}{}{0}{\BOOL}\), \(f\) can be
freely copied. \qed
\label{ex:typing-function}  
\end{example}

The rule \rn{T-RFun} for typing recursive functions is similar to \rn{T-Fun},
except that for typing the function body \(M\), \(f\) is given a special function type \(\Trfun{\tau_1}{\tau_1'}{\NTE}{\tau_2}\).
(Thus, \rn{T-Fun} may be viewed as a special case, where \(f\) is not used in \(M\).)
The type \(\Trfun{\tau_1}{\tau_1'}{\NTE}{\tau_2}\) describes a function from \(\tau_1\) to \(\tau_2\), where
an environment described by \(\NTE\) may be accessed during the evaluation of the function body.
Unlike a function of type \(\Tfuns{\tau_1}{\tau_1'}{|\NTE|}{\tau_2}\), it does not own the store corresponding to \(\NTE\).
This is for allowing a recursive function like
\[
\fixexp{f}{x}{\ifexp{*}{!y}{(y:=\NOT(!y); f\,x)}},
\]
which accesses the reference cell \(y\) both inside the recursive call of \(f\) and in
\(y:=\NOT(!y)\).
(Here, \(*\) denotes a non-deterministic Boolean.)

We have two rules for function applications: \rn{T-App} for normal function calls
and \rn{T-RApp} for recursive function calls from the inside of a recursive function.
In \rn{T-App}, we just require that the argument type of \(f\) and the type of
the actual argument \(x\) match. In \rn{T-RApp}, we additionally require
\(\Delta\subsetTE \Gamma\setminus\set{f,x}\), which means that \(\Delta\) is a
subsequence of \(\Gamma\setminus\set{f,x}\). That condition ensures that the environment
described by \(\Delta\) is available when the recursive function \(f\) is called.
Note here that \(\Delta\) should not share the binding on \(x\); see Example~\ref{ex:illtyped}
\begin{example}
  The recursive function
  \(\fixexp{f}{x}{\ifexp{*}{!y}{(y:=\NOT(!y); f\,x)}}\) mentioned above
  is typed as follows, where
  \(\Gamma = y\COL\Tref{\BOOL},f\COL \Trfun{\UNIT}{}{y\COL\Tref{\BOOL}}{\BOOL},
       x\COL\UNIT\).
  \[\footnotesize
  \infer{y\COL\Tref{\BOOL}\p
    \fixexp{f}{x}{\ifexp{*}{!y}{(y:=\NOT(!y); f\,x)}}:
    \Tfuns{\UNIT}{}{1}{\BOOL}
    \q y\COL\Tref{\BOOL}}
   {\infer
     {\Gamma
       \p \ifexp{*}{!y}{(y:=\NOT(!y); f\,x)}:\BOOL
    \q \Gamma}
     {\Gamma
       \p {!y}:\BOOL
       \q \Gamma
       &
    \infer{\Gamma
       \p y:=\NOT(!y); f\,x:\BOOL
       \q \Gamma}
          {\Gamma
       \p y:=\NOT(!y):\BOOL
       \q \Gamma & \Gamma
       \p f\,x:\BOOL
       \q \Gamma}
     }
   }
  \] \qed
\end{example}

\begin{example}
  \label{ex:illtyped}
  The type judgment \(x\COL\Tref{\BOOL}\p \fixexp{f}{y}{\ifexp{*}{x:=!y}{f(x)}}:\Tfuns{\Tref{\BOOL}}{}{1}{\UNIT}\q \emptyTE\) does NOT hold. For this to be typable, the body of the function must be typed under:
    \begin{align*}
    \Gamma \Def  x\COL\Tref{\BOOL},f\COL\Trfun{\Tref{\BOOL}}{}{\Delta}{\UNIT},y\COL\Tref{\BOOL},
    \end{align*}
    where \(\Delta = x\COL\Tref{\BOOL}\).
    Therefore, the function call \(f(x)\) must also be typed under the same type environment \(\Gamma\).
    However, \(\Delta\) is not a subsequence of \(\Gamma\setminus\set{f,x} = y\COL\Tref{\BOOL}\).
    In fact, the function above is problematic, because, upon the recursive call \(f(x)\), the ownership for \(x\) is duplicated (as \(y\) becomes  an alias of \(x\)).
\qed
\end{example}

\subsection{Translation from  \bPCFRL{\bty} to \bPCF{\bty}}
\label{sec:trans}
We now formalize the translation from \bPCFRL{\bty} to \bPCF{\bty}.
Thanks to the type-based usage control of closures containing reference cells,
we can represent a closure as a pair consisting of
a part of the store captured by the closure and the code,
as sketched in Section~\ref{sec:overview}.

We define a type-based transformation relation
\(\Gamma\p M:\tau\q \Gamma'\Tr \mvE\),
which means that a \bPCFRL{\bty} term \(M\) typed by
\(\Gamma\p M:\tau\q \Gamma'\) can be transformed to a \bPCF{\bty}
term \(\mvE\).

The transformation relation is defined by the rules in
Figure~\ref{fig:tr-BPCFRL}.
Except the part ``\(\Tr \mvE\)'', each rule \rn{Tr-xx} is the same as
the corresponding typing rule \rn{T-xx}. Thus, for every typable
term \(M\) of \bPCFRL{\bty}, the output \(\mvE\) is defined by induction on the type derivation for \(M\).

\begin{figure}
  \begin{multicols}{2}
\infrule[Tr-Fail]{}{\Gamma \p \FAIL:\bt\q \Gamma\Tr (\FAIL,\pack{\Gamma})}
\infrule[Tr-Unit]{}{\Gamma \p \Vunit:\UNIT\q \Gamma\Tr (\Vunit,\pack{\Gamma})}
\infrule[Tr-Const]{}{\Gamma \p c^\bty:\bty\q \Gamma\Tr (c, \pack{\Gamma})}
  \rulesp
  \infrule[Tr-Var]{\Gamma(x)=\tau}
          {\Gamma\p x:\tau\q \Gamma\exc{x}
    \Tr (x, \pack{\Gamma\exc{x}})}
  \rulesp
\infrule[Tr-Deref]{\Gamma(x)=\Tref{\rty}}
        {\Gamma \p !x:{\rty} \q \Gamma\excr{x}{\rty}
        \Tr (x,\pack{\Gamma\excr{x}{\rty}})}
  \rulesp
\infrule[Tr-Mkref]{\Gamma(x)= \rty}
        {\Gamma \p \mkref{x}: \Tref{\rty}\q \Gamma\exc{x}
          \Tr (x, \pack{\Gamma\exc{x}})}

  \rulesp
  \end{multicols}
    \infrule[Tr-Op]{\Gamma(x_i)=\bty\mbox{ for each $i\in\set{1,\ldots,k}$}}
  {\Gamma\p \Op{k}(x_1,\ldots,x_k):\bty\q\Gamma\Tr (\Op{k}(x_1,\ldots,x_k),\pack{\Gamma})}
  \rulesp
  \infrule[Tr-Assign]{
  \Gamma(y)=\Tref{\rty}\andalso \Gamma(x)=\rty}
        {\Gamma\p y:=x:\Tunit\q\Gamma\exc{x}\Tr
          \letexp{y}{x}(\Vunit, \pack{\Gamma\exc{x}})}
  \rulesp
\infrule[Tr-Let]{\Gamma \p M_1: \tau' \q \Gamma''\Tr \mvE_1\andalso
  \Gamma'', x\COL\tau'\p M_2:\tau\q \Gamma' \Tr \mvE_2}
        {\Gamma \p \letexp{x}{M_1}{M_2}: \tau\q \Gamma'\setminus{x} \Tr
          \letexp{(x,h_1)}{\mvE_1}\unpack{\Gamma''}{h_1}\mvE_2}
  \rulesp
  \infrule[Tr-If]
      {\Gamma(x)=\bty\andalso 
      \Gamma\p M_1:\tau\q \Gamma'\Tr \mvE_1
      \andalso \Gamma\p M_2:\tau\q \Gamma'\Tr \mvE_2}
      {\Gamma\p \ifexp{x}{M_1}{M_2}:\tau\q \Gamma'\Tr
       \ifexp{x}{\mvE_1}{\mvE_2}}
      \rulesp
  \infrule[Tr-Fun]
      {
        \Delta, x\COL\tau_1\p M:\tau_2\q \Delta,x\COL\tau_1\Tr \mvE
    \andalso (\Delta,\Gamma') = \splitTE{\Gamma}{\FV(\lambda x.M)}
            }
      {\Gamma \p \lambda x.M:\Tfuns{\tau_1}{\tau_1'}{|\Delta|}\tau_2\q
              \Gamma'\\\Tr
              (\Clos{\lambda (x, h).\unpack{\Delta}{h}
                \letexp{(r,h')}{\mvE} \\
                \unpack{\Delta,x\COL\tau_1}{h'}(r,x,\pack{\Delta})}{\pack{\Delta}}, \pack{\Gamma'})
      }

\rulesp
  \infrule[Tr-RFun]
      {\Delta, f\COL\Trfun{\tau_1}{\tau_1'}{\Delta}{\tau_2}, x\COL\tau_1\p
    M: \tau_2\q \Delta,
    f\COL\Trfun{\tau_1}{\tau_1'}{\Delta}{\tau_2}, x\COL\tau_1\Tr N\\
    (\Delta,\Gamma')=\splitTE{\Gamma}{\FV(\fixexp{f}{x}{M})}}
      {\Gamma\p
        \fixexp{f}{x}{M}:\Tfuns{\tau_1}{\tau_1'}{|\Delta|}{\tau_2}\q \Gamma'\\
        \Tr
        (\Clos{\fixexp{f}{(x,h_f)}
          {\unpack{\Delta}{h_f}
         \letexp{(r,h_2)}{\mvE}\\
          \unpack{\Delta,
            f\COL\Trfun{\tau_1}{\tau_1'}{\Delta}{\tau_2}, x\COL\tau_1}{h_2}
        (r,x,\pack{\Delta})}}{\pack{\Delta}},
        \pack{\Gamma'})}
\rulesp
  \infrule[Tr-App]
     {\Gamma(f)=\Tfuns{\tau_1}{\tau_1'}{m}{\tau_2}\andalso \Gamma(x)=\tau_1}
     {\Gamma
       \p f\,x:\tau_2\q \Gamma\\\qquad
       \Tr 
       \letexp{(r,x,h_f)}{\codeof{f}(x, \envof{f})} \letexp{f}{\Clos{\codeof{f}}{h_f}}
       (r, \pack{\Gamma})}

  \rulesp
  \infrule[Tr-RApp]
      {\Gamma(f)=\Trfun{\tau_1}{\tau_1'}{\Delta}{\tau_2}\andalso \Gamma(x)=\tau_1
      \andalso \Delta \subsetTE \Gamma\setminus\set{f,x}}
     {\Gamma
       \p f\,x:\tau_2\q \Gamma\\\qquad
       \Tr
       \letexp{h_f}{\pack{\Delta}}
       \letexp{(r,x,h'_f)}{f(x, h_f)} 
       \unpack{\Delta}{h'_f}
       (r, \pack{\Gamma})}

\caption{Transformation Rules}
\label{fig:tr-BPCFRL}
\end{figure}

Before explaining each rule, let us state how the transformation
preserves types. We define the translation of \bPCFRL{\bty} types
and type environments into 
those of \bPCF{\bty}, as follows.
\begin{align*}
  &  \toBPCF{\UNIT} = \UNIT \quad \toBPCF{\bty}=\bty\quad
  \toBPCF{(\Tref{\rty})}=\toBPCF{\rty}\quad \toBPCF{(\Tfuns{\tau_1}{}{n}{\tau_2})}
  = \bty^n \times (\toBPCF{\tau_1}\times \bty^n\to \toBPCF{\tau_2}\times \toBPCF{\tau_1}\times \bty^n)\\&
  \toBPCF{(\Trfun{\tau_1}{}{\NTE}{\tau_2})} =
  (\toBPCF{\tau_1}\times \bty^{|\NTE|}\to \toBPCF{\tau_2}\times \toBPCF{\tau_1}\times\bty^{|\NTE|})\quad \toBPCF{(x_1\COL\fty_1,\ldots,x_k\COL\fty_k)} =
  x_1\COL\toBPCF{\fty_1},\ldots,x_k\COL\toBPCF{\fty_k}.
\end{align*}
Note that, after the translation, the return type of a function now contains
the type of an argument. For example, \(\Tfuns{\Tref{\bty}}{}{2}{\UNIT}\)
becomes
\(\bty^2\times (\bty\times \bty^2\to \UNIT\times \textcolor{red}{\bty}\times\bty^2)\); this is because the function may take a reference cell or a closure
as an argument and change the value of them. For example,
\(\lambda x\COL\Tref{\BOOL}.x:=\TRUE\) of type \(\Tfuns{\Tref{\BOOL}}{}{0}\UNIT\)
should be transformed to \(\lambda x. (\Vunit,\TRUE)\) (here the store argument
is omitted for simplicity), to pass the information about the updated value of \(x\).

Our transformation rules ensure that
if \(\Gamma\p M:\tau\q \Gamma'\Tr \mvE\), then
\(\toBPCF{\Gamma}\p \mvE:\toBPCF{\tau}\times \bty^{|\Gamma'|}\).
The second element of \(\mvE\) represents the store
available after the evaluation of \(M\).
That element is generated by the idiom \(\pack{\Gamma}\), defined by:
\begin{align*}
  &  \pack{\Gamma} = \mathit{flatten}(\packaux{\Gamma}) \qquad \\  &
  \packaux{\epsilon} = ()\qquad
  \packaux{\Gamma,x\COL\bty}
  =\packaux{\Gamma,f\COL\Tfuns{\tau_1}{}{0}{\tau_2}}
  =\packaux{\Gamma,f\COL\Trfun{\tau_1}{}{\Delta}{\tau_2}} = \packaux{\Gamma}\\
   &\packaux{\Gamma,x\COL\Tref{\rty}} = (\packaux{\Gamma}, x) \qquad \packaux{\Gamma,x\COL\Tfuns{\tau_1}{}{n}{\tau_2}}
    = (\packaux{\Gamma}, \envof{x}) \mbox{ (if $n>0$)}
\end{align*}
Here, \(\packaux{\Gamma}\) constructs a nested tuple expression that represents
the value of the store, and \(\mathit{flatten}\) flattens the nested tuple.
For example, if \(\Gamma = x\COL\Tref{\bty}, y\COL\bty, f\COL
\Tfuns{\tau_1}{}{2}{\tau_2}\), then
\(\pack{\Gamma}\) is:
\( \letexp{(u_1,u_2)}{\envof{f}}{(x,u_1,u_2)}\).
We also use the idiom \(\unpack{\Gamma}{h}\mvE\) to perform the inverse of
\(\pack{\Gamma}\),
which recovers variable bindings from the tuple \(h\) representing 
the current store. For example, for \(\Gamma\) above,
\(\unpack{\Gamma}{h}\mvE\) is:
\[ \letexp{(x,u_1,u_2)}{h}\letexp{\envof{f}}{(u_1,u_2)}\mvE.\]
(Here, \(\letexp{\envof{f}}{\mvE_1}{\cdots}\) abbreviates
\(\letexp{f}{\Clos{\codeof{f}}{\mvE_1}}{\cdots}\).)

We now explain main transformation rules.
The rule \rn{Tr-Var} for variables produces, in addition to
\(x\), the store value \(\pack{\Gamma\exc{x}}\). For example, we have:
\begin{align*}
&x\COL\Tref{\bty}, y\COL\bty, f\COL
\Tfuns{\tau_1}{}{2}{\tau_2},z\COL\Tref{\bty}
\p x:\Tref{\bty}\\&
\q y\COL\bty, f\COL
\Tfuns{\tau_1}{}{2}{\tau_2},z\COL\Tref{\bty}
\Tr (x, \letexp{(u_1,u_2)}{\envof{f}}{(u_1,u_2,z)}).
\end{align*}
As the target language \bPCF{\bty} does not have mutable references,
we have to explicitly maintain the values of the store in the source term
as a tuple.
As shown in the rules \rn{Tr-Deref} and \rn{Tr-Mkref},
\(\deref{x}\) and \(\mkref{x}\) just become \(x\) (paired with the store value);
recall that the reference cell is replaced by the value stored in the cell.
If \(x\) is  a nested reference of type
\(\Tref{\Tref{\bty}}\), \(\deref{\deref{x}}\) just becomes \(x\).
Similarly, in \rn{Tr-Assign}, the reference update is just
replaced by a let-binding.

The rule \rn{T-Let} just compositionally transforms \(M_1\) and \(M_2\)
to \(\mvE_1\) and \(\mvE_2\),
but it inserts \(\unpackwo{\Gamma''}{h_1}\), so that \(\mvE_2\) is
evaluated with the current values of reference cells.

\begin{example}
  \label{ex:trns-simpleexp}
  The expression \(\letexp{x}{(y:=\NOT(y))}{!y}\)
  is transformed as follows, where \(\Gamma=y\COL\Tref{\BOOL}\).
  \[
  \footnotesize
  \infer{\Gamma\p \letexp{x}{(y:=\NOT(y))}{!y}\q \Gamma\Tr
    \letexp{(x,h_1)}{(\letexp{y}{\NOT(y)}(\Vunit,y))}
    {\letexp{y}{h_1}(y,y)}
  }
        {\Gamma\p y:=\NOT(y)\q \Gamma\Tr
          \letexp{y}{\NOT(y)}(\Vunit,y)
          &
          \Gamma,x\COL\UNIT\p !y\q \Gamma,x\COL\UNIT\Tr
          (y, y)}
        \]
        \qed
\end{example}

The rule \rn{Tr-Fun} for functions splits the current store to
the one owned by the closure (which is represented by \(\pack{\Delta}\))
and others (represented by \(\pack{\Gamma'}\)).
The function is translated to a pair consisting of its local store
and the code, where the code part takes the argument \(x\) and the current
value of the local store \(h\), recovers variable bindings from the store
(by \(\unpack{\Delta}(h)\cdots\)), and evaluates the body \(\mvE\).
The rule \rn{Tr-RFun} is similar.
\begin{example}
  \label{ex:tr-fun}
  The function \(\lambda x. x:=\NOT(!y);y!=\NOT(!x)\) is translated
  as follows, where \(\Gamma=y\COL\Tref{\BOOL},z\COL\Tref{\BOOL}\)
  and \(\Gamma' = y\COL\Tref{\BOOL},x\COL\Tref{\BOOL}\).
  \[\footnotesize
  \infer{\Gamma\p \lambda x. x:=\NOT(!y);y!=\NOT(!x):
    \Tfuns{\Tref{\BOOL}}{}{1}{\UNIT}\q
    z\COL\Tref{\BOOL}
    \Tr \mvE}
        {\Gamma'\p x:=\NOT(!y);y:=\NOT(!x):
    {\UNIT}\q
    \Gamma'
    \Tr \letexp{x}{\NOT(y)}\letexp{y}{\NOT(x)}
    {(\Vunit, (y,x))}}
  \]
  Here, (a simplified form of) \(N\) is \((\Clos{N_1}{y},z)\), where
  \(N_1\) is:
  \begin{align*}
&\lambda (x,y).
    \letexp{(r,h_1)}{(\letexp{x}{\NOT(y)}\letexp{y}{\NOT(x)}
       {(\Vunit, (y,x))})}&&\\&\qquad\quad
     \letexp{(y,x)}{h_1}(r,x,y). && \qed
  \end{align*}
\end{example}

In the rule \rn{Tr-App} for function applications,
the output of the transformation passes
the argument and \(f\)'s store to the code of \(f\),
and after the return, updates the local store of \(f\).
In the case of \rn{Tr-RApp}, \(f\) keeps only the code part;
thus before calling \(f\), it reconstructs the store part of \(f\)
by \(\pack{\Delta}\). 

\begin{example}
  \label{ex:tr-app}
  Let us consider:
  \begin{align*}
    \fixexp{f}{x}{\ifexp{*}{z:=x}{(x:=\NOT(y);f\,x;y:=x)}}.
  \end{align*}
  Under the type environment \(y\COL\Tref{\BOOL},z\COL\Tref{\BOOL}\),
  It is translated to \((\Clos{N}{(y,z)}, \Vunit)\), where \(N\) is:
{\small  
  \begin{align*}
 &   \fixexp{f}{(x,h_f)}{
      \letexp{(y,z)}{h_f}
      \ifexp{*}{\letexp{z}{x}(\Vunit, x, (y,z))}\\&\quad
            {\letexp{x}{\NOT(y)}
              \letexp{(r,x,h')}{f(x, (y,z))}
              \letexp{(y,z)}{h'}\letexp{y}{x}(r,x,(y,z))}
    }.
  \end{align*}
  }
  The store \((y,z)\) for \(f\) is packed up before the call of \(f\),
  and then \(y\) and \(z\) are updated after the call (by \(\letexp{(y,z)}{h'}\)).
  Besides, \(f\) also returns the updated value of \(x\), which is used
  for emulating \(y:=x\) by \(\letexp{y}{x}\cdots\).
  \qed
\end{example}

We now state the correctness of the transformation.
\begin{theorem}
  \label{th:main}
  If \(\emptyTE \p M:\bt\q \emptyTE \Tr \mvE \),
  then \(M\) reaches \(\FAIL\) if and only if \(\mvE\) reaches \(\FAIL\).
\end{theorem}
Though the above theorem may be intuitively clear from
the explanation and examples above, the actual proof is involved, which is
found \iffull in Appendix~\ref{sec:proof}. \else
a longer version~\cite{K25POPLfull}.
\fi
As a corollary of Theorem~\ref{th:main} and the decidability of
the reachability problem for BPCF, we have:
\begin{corollary}
  Given a term \(M\) of \BPCFRL{} such that \(\emptyTE\p M:\bt\q \emptyTE\),
  it is decidable whether \(M\) reaches \(\FAIL\).
\end{corollary}

\begin{remark}
  \label{rem:fun}
  A careful reader may have noticed that our typing rules for functions
  are too rigid, in that functions of types \(\Tfuns{\tau_1}{}{m}{\tau_2}\)
    and \(\Tfuns{\tau_1}{}{n}{\tau_2}\) are incompatible if \(m\ne n\).
    To mitigate this, one can add the following subsumption rule:
    {\small
      \infrule{\Gamma\p M:\Tfuns{\tau_1}{}{n}{\tau_2}\q \Gamma'\Tr \mvE
        \andalso m>n}
              {\Gamma\p
                M:\Tfuns{\tau_1}{}{m}{\tau_2}\q \Gamma'\Tr
                \letexp{(\Clos{N_1}{h_1},h_2)}{N}\\
              (\Clos{\lambda (x,h).
                  \letexp{(r,x,h_1')}{N_1(x,h|_n)} (r,x,h_1'\cdot c^{m-n})}
                       {h_1\cdot c^{m-n}},{h_2})}}
    The rule above just adds the dummy elements to the function's local store.
    Here, \(c^k\) represents a tuple consisting of \(k\) copies of
    an (arbitrary) constant \(c\) of type \(\bty\), \(h_1\cdot h_2\) concatenates the tuples \(h_1\) and \(h_2\),
    and \(h|_n\) constructs a tuple consisting of the first \(n\) elements of
    \(h\). \qed
\end{remark}

\subsection{Extensions}
\label{sec:disc}

The type system of \bPCFRL{\bty} may be too restrictive to be used in practice.
We have deliberately made so to keep the formalization simple.
We sketch below some extensions to relax the restriction; the formalization
of those extensions is left for future work.

With those extensions, we expect to be able 
to obtain a fairly expressive
higher-order functional language with references, whose programs
can be mapped to
 pure functional programs and then be model-checked. As in Rust, however,
programmers would have to be aware of, and obey a certain ownership/linearity principle
when using references or closures containing references. That would be the price
worth to pay, to use a fully automated, precise verification tool.

\subsubsection{Fractional Ownerships}
In the transformation above, we have represented a closure as a pair consisting of
the code and the current values of the reference cells captured by the closure.
If, however, a closure does not update the reference cells, we can omit 
the store part of the closure and allows the closure to be freely copied.
For example, \(\lambda x. \deref{y}\) is
just transformed to \(\lambda x.y\) (provided that there is no mutable copy of \(y\)
elsewhere). To support such read-only references, we can employ
fractional ownership types~\cite{Boyland03SAS,DBLP:conf/esop/TomanSSI020}.

Let us extend reference types by:
\[\rty  ::= \bty \mid \FTref{\rty}{r}.\]
Here, \(r\) ranges over the interval \([0,1]\), where \(r=1\) represents the full
ownership and \(r\) such that \(0<r<1\) represents a partial ownership, with which
the reference can be used only for dereference, not for update.
Then, the transformation rules for dereference and update become:
\infrule[Tr-Deref']{\Gamma(x)=\FTref{\rty}{r}\andalso r>0\andalso \rty=\rty_1+\rty_2}
        {\Gamma \p !x:{\rty_1} \q \Gamma\set{x\mapsto \FTref{\rty_2}{r}}
        \Tr (x,\pack{\Gamma\set{x\mapsto \FTref{\rty_2}{r}}})}

  \infrule[Tr-Assign']{
    \Gamma(y)=\FTref{\rty'}{1}\andalso \Gamma(x)=\rty\andalso  \rty=\rty_1+\rty_2
  \andalso \rty\approx \rty'}
        {\Gamma\p y:=x:\Tunit\q\Gamma\set{x\mapsto \rty_2,y\mapsto \FTref{\rty_1}{1}}\Tr
          \letexp{y}{x}(\Vunit, \pack{\Gamma\set{x\mapsto \rty_2,y\mapsto \FTref{\rty_1}{1}}})}
        Here, \(\FTE\set{x\mapsto \sigma}\) denotes the type environment obtained
        by replacing the binding on \(x\) with \(x\COL\sigma\).
        In rule \rn{Tr-Deref'} for dereference, a partial ownership \(r>0\) is sufficient, while
        in rule \rn{Tr-Assign'} for update, the full ownership on \(y\) is required.
        The condition \(\rty=\rty_1+\rty_2\) means that the ownership in \(\rty\)
        can be split into \(\rty_1\) and \(\rty_2\). For example,
        \(\FTref{\bty}{1} = \FTref{\bty}{0.3}+\FTref{\bty}{0.7}\).
        The condition \(\rty\approx \rty'\) means that
        \(\rty\) and \(\rty'\) may differ only in ownership values. For example,
        \(\FTref{\bty}{0}\approx\FTref{\bty}{1}\).

        The rule for functions becomes:
  \infrule[Tr-Fun']
      {
        \Delta, x\COL\tau_1\p M:\tau_2\q \Delta,x\COL\tau_1\Tr \mvE
    \andalso \Gamma = \Delta + \Gamma'
            }
      {\Gamma \p \lambda x.M:\Tfuns{\tau_1}{\tau_1'}{|\Delta|}\tau_2\q
        \Gamma' \Tr
              (\Clos{\lambda (x, h).\unpack{\Delta}{h}
                \letexp{(r,h')}{\mvE} \\\qquad\qquad\qquad\qquad\qquad\qquad
                \unpack{\Delta,x\COL\tau_1}{h'}(r,x,\pack{\Delta})}{\pack{\Delta}}, \pack{\Gamma'})
      }
      Here, the condition \((\Delta,\Gamma') = \splitTE{\Gamma}{\FV(\lambda x.M)}\)
      in \rn{Tr-Fun} has been replaced by
      \( \Gamma = \Delta + \Gamma'\), which allows the ownerships in \(\Gamma\)
      to be split into \(\Delta\) and  \(\Gamma'\) in an arbitrary manner.
      For example, \( y\COL\FTref{\bty}{1} =
      (y\COL\FTref{\bty}{0.5}) +(y\COL\FTref{\bty}{0.5})\).
      The definition of \(|\rty|\) used in the calculation of \(|\NTE|\) is modified as follows,
      so that only the size of values stored in mutable references
      (i.e., references with the full ownership) is counted.
      \begin{align*}
        |\FTref{\rty}{r}|=\left\{\begin{array}{ll}
     \Tsize{\rty} & \mbox{if $r=1$}\\
     |\rty| & \mbox{if $0<r<1$}\\
     0 & \mbox{if $r=0$}\\
        \end{array}\right.
  \end{align*}
      Accordingly, \(\pack{\Delta}\) and \(\unpack{\Delta}{h}\) manipulate only
      values stored in mutable references.

      With the extension above, for example, the following expression is allowed.
      \[\letexp{y}{\mkref{1}}\letexp{f}{\lambda x.!y}\letexp{g}{f}{f()+g()+!y}.
      \]
      Here, \(f\) captures the reference cell \(y\), but by giving only
      partial ownership (say, \(0.5\)) to \(f\),
      we can still use \(y\) in
      the expression \(f()+g()+!y\). Furthermore, \(f\) can be copied, as
      it has type \(\Tfuns{\UNIT}{\tau_1'}{0}\INT\).
      The expression is transformed to (with some simplification):
      \[\letexp{y}{1}\letexp{f}{\lambda x.y}\letexp{g}{f}{f()+g()+y}.
      \]
\subsubsection{Borrowing and RustHorn-style Prophecies}
\label{sec:borrow}
As indicated by the rules \rn{T-Fun} and \rn{T-RFun},
the reference cells captured by a closure can no longer be directly accessed,
even after the closure is thrown away.
This restriction excludes out, for example, the following program:
\begin{align*}
&  \letexp{x}{\mkref{\TRUE}}\\&
  \letexp{\_}{\letexp{f}{\lambda z.x:=\NOT(x)}f\Vunit}\\&
  !x.
\end{align*}
Here, the function \(f\) that captures the reference cell \(x\)
is locally defined and used, and then \(x\) is accessed.
To allow this kind of program, we can introduce
Rust-style borrowing, and allow \(f\) to temporally borrow
the ownership of \(x\), and the borrowed ownership
to be returned when \(f\)'s lifetime expires.
We can then use the idea of RustHorn~\cite{DBLP:journals/toplas/MatsushitaTK21} 
to convert the program above to:
\begin{align*}
&  \letexp{(x_c,x_f)}{(\TRUE,*)}\\&
  \letexp{\_}{\letexp{f}{\Clos{\lambda (z, (x_c,x_f)).\letexp{x_c}{\NOT(x_c)}
        (\Vunit,z,(x_c,x_f))}{(x_c,x_f)}}\\&\qquad\quad
    \letexp{(r,z,(x_c,x_f))}{\codeof{f}(\Vunit,\envof{f})}
     \Assume(x_c=x_f)}\\&
  x_f.
\end{align*}
The reference cell \(x\) is now represented by a pair \((x_c,x_f)\),
where \(x_c\) holds the current value of \(x\), and \(x_f\) holds the ``prophecy'' on
the value of \(x\) when the ownership is returned from \(f\).
Initially, \(x_f\) is set to a non-deterministic value \(*\), since the future value
is unknown at this moment.
When the ownership is indeed returned, \(x_f\) is set to \(x_c\)
by the assume statement.

To realize the above transformation, we need to extend types with
lifetime information. Following \cite{DBLP:conf/vmcai/NakayamaMSSK24},
we can extend a reference type \(\Tref{\rty}\) to
\(\BTref{\rty}{\alpha}\) and \(\LTref{\rty}{\alpha}{\beta}\),
where \(\BTref{\rty}{\alpha}\) describes a reference whose ownership is
valid only during lifetime \(\alpha\) (here, a lifetime is just a symbol that
represents an abstract notion of some time period),
and \(\LTref{\rty}{\alpha}{\beta}\) describes a reference such that 
its ownership is valid during lifetime \(\alpha\), but the ownership is
temporarily lent to its alias during lifetime \(\beta\).
For example, consider the following expression.
\begin{align*}
  &  \letexp{x}{\mkref{\FALSE}@\alpha} && \mbox{// $x\COL \BTref{\BOOL}{\alpha}$}\\
  & \NewL{\beta}; \letexp{y}{x@\beta}&& \mbox{// $x\COL\LTref{\BOOL}{\alpha}{\beta}, y\COL \BTref{\BOOL}{\beta}$}\\
  & y := \NOT(y); \EndL{\beta}; && \mbox{// 
    $x\COL\BTref{\BOOL}{\alpha}$}\\
  & \Assert(!x).
\end{align*}
The expressions \(\NewL{\beta}\) and \(\EndL{\beta}\) respectively indicate
the start and end of lifetime \(\beta\); we expect that
they will be automatically inserted by a lifetime inference algorithm, as in Rust.
On the first line, \(x\) is set to a reference that is alive during lifetime \(\alpha\).
On the second line, \(y\) is created as an alias of \(x\), and the ownership of
\(x\) is lent to \(y\) during liftime \(\beta\). The lifetime \(\beta\) expires
after the third line, and the type of \(x\) goes back to \(\BTref{\BOOL}{\alpha}\).

The rules for borrowing (\(x@\beta\) above), dereference, and \(\EndL{\beta}\) would
look like:\footnote{For simplicity, we do not consider nested references here.}
\infrule[Tr-Borrow]{\FTE(x)=\BTref{\bty}{\alpha}\andalso \FTE'=\FTE\set{x\mapsto \LTref{\bty}{\alpha}{\beta}}}
        {\FTE \p x@\beta: \BTref{\bty}{\beta}\q
          \FTE' \\
          \Tr
          \letexp{(x_c,x_\alpha)}{x}\letexp{x_\beta}{*}\letexp{x}{(x_\beta,x_\alpha)}
          \letexp{y}{(x_c,x_\beta)}(y,\pack{\FTE'})}
\infrule[Tr-Deref'']{\Gamma(x)=\BTref{\bty}{\alpha}}
        {\Gamma \p !x:\bty \q \Gamma  \Tr (x,\pack{\Gamma})}
\infrule[Tr-End]{\mbox{$\alpha$ does not occur in $\Gamma\Up\alpha$}}
        {\Gamma \p \EndL{\alpha} \q \Gamma\Up\alpha  \Tr
          (\Assumes{\Gamma}{\alpha},\pack{\Gamma\Up\alpha})}

        In \rn{Tr-Borrow}, the type of \(x\) is updated from \(\BTref{\bty}{\alpha}\)
        to \(\LTref{\bty}{\alpha}{\beta}\). As indicated in the output of transformation,
        a value of type \(\BTref{\bty}{\alpha}\) is a pair
        \((x_c,x_\alpha)\) where \(x_c\) is the current value and \(x_\alpha\) is the future
        value of \(x\)  (called the prophecy value) at the end of lifetime \(\alpha\). 
        A value of type \(\LTref{\bty}{\alpha}{\beta}\) is a pair \((x_\beta,x_\alpha)\)
        where \(x_\beta\) is the future value of \(x\) at the end of lifetime \(\beta\).
        In \rn{Tr-Deref''}, we require that the type of \(x\) is of the form \(\BTref{\bty}{\alpha}\), not of the form \(\LTref{\bty}{\alpha}{\beta}\). A reference of the latter type
        is not accessible at the moment, as the ownership has not been returned.

        In \rn{Tr-End}, \(\FTE\Up{\alpha}\) denotes the type environment obtained
        by replacing type bindings of the form \(x\COL\LTref{\bty}{\beta}{\alpha}\)
        with \(x\COL\BTref{\bty}{\beta}\), and 
         removing type bindings of the form \(x\COL\BTref{\bty}{\alpha}\)
        and \(x\COL\BTfuns{\tau_1}{}{n}{\alpha}{\tau_2}\). Here, 
        \(\BTfuns{\tau_1}{}{n}{\alpha}{\tau_2}\) describes a function that is valid
        only during lifetime \(\alpha\).
        \(\Assumes{\FTE}{\alpha}\) sets the prophecy values for \(\alpha\)
        to the current values by assume statements. For example,
        if \(\FTE=x\COL \BTref{\bty}{\alpha}, y\COL \LTref{\bty}{\alpha}{\beta},
        f\COL \BTfuns{\UNIT}{}{2}{\alpha}{\UNIT}\),
        then  \(\Assumes{\FTE}{\alpha}\) is:
        \begin{align*}
          \letexp{(x_c,x_\alpha)}{x}\letexp{(h,\_)}{f}
          \letexp{((z_c,z_\alpha),(w_c,w_\alpha))}{h}\\
          \Assume(x_c=x_\alpha\land z_c=z_\alpha\land w_c=w_\alpha).
        \end{align*}

        Finally, the rule for function creations would be:
  \infrule[Tr-Fun'']
      {
        \Delta, x\COL\tau_1\p M:\tau_2\q \Delta,x\COL\tau_1\Tr \mvE
    \andalso (\Delta,\Gamma') = \BsplitTE{\Gamma}{\FV(\lambda x.M)}{\alpha}
            }
      {\Gamma \p \lambda x.M@\alpha:\BTfuns{\tau_1}{\tau_1'}{|\Delta|}{\alpha}\tau_2\q
              \Gamma'\\\Tr
              (\Clos{\lambda (x, h).\unpack{\Delta}{h}
                \letexp{(r,h')}{\mvE} \unpack{\Delta,x\COL\tau_1}{h'}(r,x,\pack{\Delta})}{\pack{\Delta}}, \pack{\Gamma'})
      }
      Here, the only difference from \rn{Tr-Fun} is the condition
      \(\BsplitTE{\Gamma}{\FV(\lambda x.M)}{\alpha}\), which
      returns \((\Delta, \Gamma')\) where \(\Delta\)
      represents the ownerships borrowed from \(\Gamma\), and
      \(\Gamma'\) represents the rest of the ownerships.
      For example, \(\BsplitTE{x\COL\BTref{\bty}{\beta}}{\set{x}}{\alpha}
      = (x\COL\BTref{\bty}{\alpha}, x\COL\LTref{\bty}{\beta}{\alpha})\).
      Recall the example given at the beginning of this subsection:
      \begin{align*}
        &  \letexp{x}{\mkref{\TRUE}@\beta}\\&
        \NewL{\alpha};
    \letexp{\_}{\letexp{f}{\lambda z.x:=\NOT(x)@\alpha}f\Vunit}\\&
        \EndL{\alpha};
  !x.
      \end{align*}
      Here, we have added lifetime annotations.
      Thanks to the rule \rn{Tr-Fun''}, after \(f\) is created, the caller environment
      still keeps the ownership for \(x\) as \(\LTref{\BOOL}{\beta}{\alpha}\),
      and when \(\alpha\) exprires, the type of \(x\) goes back to
      \(x\COL\BTref{\BOOL}{\alpha}\); hence the dereference \(!x\) is allowed.
\subsubsection{Product and Sum Types}
\label{sec:sumtype}
It is not difficult to extend \bPCFRL{\bty} with product and sum types.
The syntax of types can be extended by:
\[\rty  ::= \bty \mid \Tref{\rty}\mid \rty_1\times \rty_2\mid \rty_1+\rty_2.\]
The rules for sum types would be:
\infrule[Tr-In]{\FTE(x)=\rty_i\andalso i\in\set{1,2}}{\FTE\p \Inj_i(x): \rty_1+\rty_2\q \FTE\exc{x}\Tr
  (\Inj_i(x),\pack{\FTE\exc{x}})}
\infrule[Tr-Match]{\FTE(x)=\rty_1+\rty_2\andalso
  \FTE\excr{x}{\rty_1},y\COL\rty_1\p M_1:\nty\q \FTE'\Tr \mvE_1\andalso
  \FTE\excr{x}{\rty_2},z\COL\rty_2\p M_2:\nty\q \FTE'\Tr \mvE_2
}
        {\FTE\p \MATCH\;x\;\WITH\;\Inj_1(y)\to M_1|\Inj_2(z)\to M_2\q
          \FTE'\\
          \Tr \MATCH\;x\;\WITH\;\Inj_1(y)\to \mvE_1|\Inj_2(z)\to \mvE_2}
As before, reference cells are replaced by their contents by the transformation; for example,
\((\mkref{1}, \mkref (2,\mkref{3}))\) in the source language becomes
\((1,(2,3))\) in the target language.

Using sum types (or variant types), we can express a stateful object.
For example, the following code defines a function \(\mathit{newcounter}\)
to create a counter object with two methods \(\mathit{Inc}\) and \(\mathit{Read}\),
creates a counter \(c\) with the initial value \(n\), increments it, and then
reads it. (Here, for readability, we wrote \(\mathit{Inc}\) and \(\mathit{Read}\)
for \(\Inj_1\) and \(\Inj_2\) respectively.)
\begin{align*}
&  \letexp{\mathit{newcounter}}{\lambda \mathit{init}. 
    \letexp{r}{\mkref{\mathit{init}}}\\&\qquad\qquad\qquad\qquad
    \letexp{\mathit{f}}{\lambda x.
      \mathtt{match}\ x\ \mathtt{with}\ \mathit{Inc}() \rightarrow (r := !r+1;0) \;|\;  \mathit{Read}() \rightarrow !r}\\&\qquad\qquad\qquad\qquad
f}\\
  &\letexp{c}{\mathit{newcounter}\,n}
  c(\mathit{Inc}()); c(\mathit{Read}()).
\end{align*}
The expression can be converted to a \bPCFRL{\INT} expression
by our transformation.

Another way to encode a stateful object
would be to represent it as a tuple of function closures, as in:
\begin{align*}
&  \letexp{\mathit{newcounter}}{\lambda \mathit{init}. 
    \letexp{r}{\mkref{\mathit{init}.}}\\&\qquad\qquad\qquad\qquad
    \letexp{\mathit{inc}}{\lambda x.r := !r+1}\\&\qquad\qquad\qquad\qquad
\letexp{\mathit{read}}{\lambda x.!r}
    (inc,read)}\\
  &\letexp{(\mathit{inc},\mathit{read})}{\mathit{newcounter}\;n}
  \mathit{inc}();\mathit{inc}();\mathit{read}()
\end{align*}
This is, however, not typable in our type system.
The program defines a function \(\mathit{newcounter}\),
which creates a counter with two functions \(\mathit{inc}\),
and \(\mathit{read}\) to increment/read the counter. The counter value is
stored in the reference cell \(r\), which is shared by the two functions, hence
rejected by our type system.
It is left for future work to extend the type system so that a record or a tuple of
closures is treated linearly as a whole (so that there cannot be multiple aliases
for the record of closures or its elements).

\subsubsection{Higher-Order References}
We have so far considered only ground references,
where functions cannot be stored.
To deal with higher-order references, we need to extend function types.
Recall that \(n\) in the function type \(\Tfuns{\tau_1}{\tau_1'}{n}\tau_2\)
represents the size of the local store of the function; in other words,
the local store can be represented as a value of type \(\bty^n\).
If we allow higher-order references, then the local store
should also contain function values. Thus, we need to consider
function types of the form \(\Tfuns{\tau_1}{\tau_1'}{\tau_3}\tau_2\),
where \(\tau_3\) represents the type of the local store;
\(\Tfuns{\tau_1}{\tau_1'}{n}\tau_2\) is just a special case where \(\tau_3=\bty^n\).
For example, the function created by the following code would have
type \(\Tfuns{\UNIT}{}{\Tfuns{\UNIT}{}{\INT}{\INT}}{\INT}\)
\begin{align*}
  \letexp{r}{\mkref{1}}\letexp{f}{\mkref{\lambda x.r:=!r+1;!r}}
         {\lambda y\COL\UNIT.!f()}.
\end{align*}
\begin{remark}
  Even if we allow higher-order references,
  Landin's knot (which is a standard trick for implementing recursion
  by using higher-order references) cannot be expressed due to the control of
  ownerships by the type system.
Consider:
\begin{align*}
  \letexp{r}{\mkref{\lambda x.x}}
  r := (\lambda x.!r\,x);
  !r(),
\end{align*}
which falls into an infinite loop.
The part \(r := \lambda x.!r\,x\) is not typable, as the ownership of \(r\)
is duplicated between \(r :=\cdots\) and \(\lambda x.!r\,x\).
\end{remark}

 \subsection{Preliminary Experiments}
\label{sec:exp}
Our transformation from \bPCFRL{\bty} to \bPCF{\bty} can be applied to
automated verification of higher-order functional programs with
references (that conform to the typing system of \bPCFRL{\INT})
 as follows.\footnote{
  An alternative approach would be to apply predicate abstraction first to
  obtain a \BPCFRL{} program, apply our transformation to obtain
  a \BPCF{} program, and then invoke a higher-order model checker for
  \BPCF{}. The approach mentioned above is however easier to implement, as
  we can reuse off-the-shelf automated verification tools.}
\begin{itemize}
\item Step 1: Check that a given program is typable, and
  apply our transformation to remove reference cells.
\item Step 2: Use an off-the-shelf automated verification tool (like \mochi{}~\cite{KSU11PLDI,SUK13PEPM})
  for (pure) higher-order functional programs to verify the resulting program.
\end{itemize}

To evaluate the effectiveness of the above approach, we have prepared
twelve small, but tricky OCaml programs with references, and manually performed
Step~1 to obtain OCaml programs without references.
We then used \mochi{}\footnote{\url{https://github.com/hopv/MoCHi/tree/develop}}
  to verify the resulting programs.
The experiments were conducted on a machine with
Intel i7-8650U CPU 1.90GHz and 16GB memory.

The experimental results are summarized in Table~\ref{tab:exp}.
The column ``result'' shows the result of verification (which indicates whether
the assertion failure may occur), and the column ``time'' shows the running time of
\mochi{}, as shown by the web interface.

Benchmark programs used in the experiments are described as follows.
The exact code is found in
\iffull Appendix~\ref{sec:benchmark}. \else
a longer version~\cite{K25POPLfull}. \fi
\begin{itemize}
\item \texttt{repeat\_ref}: This creates a closure by \(\letexp{x}{\mkref{0}}{\lambda y.x:=!x+1;!x}\),
  repeatedly calls it \(n\) times, and asserts that the \(n\)-th call returns \(n\).
\item \texttt{repeat\_localref}: This creates a closure by \({\lambda y.\letexp{x}{\mkref{0}}x:=!x+1;!x}\),
  repeatedly calls it \(n\) times, and asserts that the \(n\)-th call returns \(1\).
\item \texttt{inc\_before\_rec}: This defines a recursive function
  \begin{verbatim}
  let x = ref 0 
  let rec f n = if n=0 then !x else (x := !x+1; f(n-1))
\end{verbatim}
  and asserts that the return value of \(f\,n\) is \(n\).
\item \texttt{inc\_after\_rec}: Similar to \texttt{inc\_before\_rec}, but the else-clause is replaced with
  \texttt{f(n-1); x := !x+1; !x}.
\item \texttt{borrow}: This consists of the following code, which is an integer version of the code discussed in Section~\ref{sec:borrow}.
\begin{verbatim}
    let x = ref n in
    let _ = let f() = x := !x+1 in f(); f() in assert(!x=n+2)
\end{verbatim}
The reference \(x\) is first captured by \(f\), but released and reused after the lifetime of \(f\) expires.
\item  \texttt{counter}: This is the example of counters discussed in
  Section~\ref{sec:sumtype}.
\end{itemize}
The program named \texttt{xx\_ng} is an unsafe version of \texttt{xx}, obtained
by replacing the assertion with a wrong one.

As the table shows, \mochi{} could correctly verify/reject all the translated programs
(except that a minor problem occurred for \texttt{counter\_ng} during the counterexample generation phase,
due to a bug of \mochi{}, not our translation). We have also tested \mochi{} against the source
programs containing references, but it could verify none of them. This confirms the effectiveness
of our approach. An implementation of the translator and larger experiments are left for future work.
\begin{table}
  \caption{Experimental Results}
  \label{tab:exp}
  \begin{center}
\begin{tabular}{|l|l|l|}
  \hline
  program & result & time (sec.) \\
  \hline
  \texttt{repeat\_ref} & safe & 5.539 \\
  \hline
  \texttt{repeat\_ref\_ng} & unsafe & 1.674\\
  \hline
  \texttt{repeat\_localref} & safe & 1.356 \\
  \hline
  \texttt{repeat\_localref\_ng} & unsafe & 1.800\\
  \hline
  \texttt{inc\_before\_rec} & safe & 1.784\\
  \hline
  \texttt{inc\_before\_rec\_ng} & unsafe & 0.526\\
  \hline
  \texttt{inc\_after\_rec} & safe & 1.737\\
  \hline
  \texttt{inc\_after\_rec\_ng} & unsafe & 0.538\\
  \hline
  \texttt{borrow} & safe & 0.379\\
  \hline
  \texttt{borrow\_ng} & unsafe & 0.651\\
  \hline
  \texttt{counter} & safe & 0.373\\
  \hline
  \texttt{counter\_ng} & unsafe & 0.650\\
  \hline
\end{tabular}
\end{center}
\end{table}
 \section{Related Work}
\label{sec:rel}
\paragraph{On (un)decidability of extensions of finitary PCF}

As already mentioned, the undecidability of the simply-typed \(\lambda\)-calculus
with exceptions has been proven to be undecidable even without recursion,
by Lillibridge~\cite{DBLP:journals/lisp/Lillibridge99}.
More precisely, he used exceptions carrying a second-order function of type
\((\UNIT\to\UNIT)\to(\UNIT\to\UNIT)\) to encode the untyped \(\lambda\)-calculus.
In contrast, we have shown that in the presence of recursion,
even exceptions carrying a first-order function (of type \(\UNIT\to\UNIT\))
leads to undecidability. For algebraic effects, the basic ideas of our proof are
attributed to Dal Lago and Ghyselen~\cite{DBLP:journals/pacmpl/LagoG24},
and our proof is just a streamlined version of their undecidability proof. Nevertheless,
the simplicity of our proof would be useful for better understanding the nature of
undecidability of the calculus with algebraic effects.

For \BPCFL{}, as clear from our proof, the undecidability comes from
the primitive for equality check, rather than the creation of a fresh symbol itself.
In fact, Kobayashi~\cite{Kobayashi09POPL} has shown the decidability of
resource usage analysis, which can be considered the reachability problem
for an extension of finitary PCF with the primitives for creating fresh symbols
(called resources in \cite{Kobayashi09POPL}) and accessing each symbol (but
not for comparing them).

For an extension with references, previous undecidability proofs are found in
\cite{DBLP:journals/jacm/JonesM78} (Theorem 3) and \cite{DBLP:journals/apal/Ong04} (Lemma 39),
which encoded Post's tag systems~\cite{Minsky} instead of Minsky machines.
Besides the reachability problem, the program equivalence/approximation problem has
been actively studied for languages with references~\cite{DBLP:journals/apal/Ong04,DBLP:journals/tcs/Murawski05,DBLP:journals/tocl/Murawski05,DBLP:conf/fossacs/MurawskiW05,DBLP:conf/icalp/MurawskiOW05,DBLP:conf/lics/Murawski08}. Murawski~\cite{DBLP:journals/tocl/Murawski05}
has shown that program equivalence for Idealized Algol is undecidable at order 4
even without recursion. 
Jaber~\cite{DBLP:journals/pacmpl/Jaber20} implemented an automated tool for
proving contextual equivalence of higher-order programs with references.

\paragraph{Ownerships for automated verification of programs with references}
Toman et al.~\cite{DBLP:conf/esop/TomanSSI020} proposed a refinement type
system for first-order imperative languages. They used 
fractional ownerships~\cite{Boyland03SAS} to ensure that when there is a mutable
reference, there is no other alias, and used that condition to enable strong updates
of refinement types. Matsushita et al.~\cite{DBLP:journals/toplas/MatsushitaTK21}
used Rust-style ownerships and reduced the reachability problem for a subset of
Rust to the satisfiability problem for Constrained Horn Clauses. Both
methods can alternatively be formalized as a reduction from first-order imperative
programs to functional programs without references~\cite{DBLP:conf/vmcai/NakayamaMSSK24,DBLP:conf/icfem/DenisJM22}. Those studies, however, considered
only first-order programs, and did not deal with closures containing references.

 \section{Conclusion}
\label{sec:conc}

We have given simple proofs of the undecidability of the reachability
problem for various extensions of finitary PCF. Although those undecidability results
were known or folklore, the simplicity and uniformity of our proofs
would be useful for better understanding the nature of undecidability.
We have also given a type-based characterization of a decidable
fragment of the extension of finitary PCF with Boolean references, and
shown a translation from the fragment into finitary PCF without
references.  We have shown that non-trivial higher-order programs with references
can be verified by combining our translation and an automated verification tool for 
functional programs without references.
Future work includes a redesign of a functional language with (Rust-style and/or fractional) ownerships,
and employ our transformation to enable fully automated verification of higher-order programs with
references.
 
\subsubsection*{Acknowledgment}
  We thank anonymous referees for useful comments,
  and Ugo Dal Lago, Andrzej Murawski, Taro Sekiyama, Tachio Terauchi, Takeshi Tsukada, and Hiroshi Unno for discussions and/or information on the literature.
This work was supported by
JSPS KAKENHI Grant Number JP20H05703.


\iffull
\newpage
\appendix
\section*{Appendix}
\section{Proof of Theorem~\ref{th:main}}
\label{sec:proof}
This section gives a proof of Theorem~\ref{th:main}.
For a map \(f\), we write \(\dom(f)\) for the domain of \(f\).
We write \(f\proj{S}\) and \(f\setminus S\) for the map obtained from \(f\)
by restricting its domain to \(S\) and \(\dom(f)\setminus S\) respetively.
We often write \(f\setminus x\) for \(f\setminus\set{x}\).
When \(\dom(f)\cap \dom(g)\), we write \(f\uplus g\) for \(f\cup g\).

We first define the operational semantics of \bPCFRL{\bty} and \bPCF{\bty}.
For the technical convenience of the correspondence between the source and target programs,
we use big-step semantics, as opposed to small-step semantics used in Sections~\ref{sec:pre} and \ref{sec:undecidability}.

For \bPCFRL{\bty}, we introduce the relation \((M,R,H)\eval \mathit{ans}\) where \(\mathit{ans}\) is of
the form \((v, R', H')\) or \(\FAIL\), and \(v\) ranges over the set defined by:
\[v ::= c \mid \ell \mid \Clos{\lambda x.M}{R}\mid \Clos{\fixexp{f}{x}{M}}{R}\]
Here, \(c\) ranges over the set of constants, \(\ell\) ranges over a set of locations,
\(R\) is an environment, which maps variables to values, and \(H\) is a heap, which maps locations to the set consisting of locations and constants.
It is defined by the rules in Figure~\ref{fig:os-bpcfrl}. Actually, we have additional rules for \(\FAIL\);
for each rule that has an assumption of the form \((M,R,H)\eval (v,R',H')\),
we have a rule that says \((M,R,H)\eval \FAIL\) implies the conclusion also evaluates to \(\FAIL\).
For example, For \rn{RE-Let}, we have two additional rules.
\infrule[RE-LetF1]{(M_1,R,H)\eval \FAIL}
        {(\letexp{x}{M_1}{M_2},R,H)\eval \FAIL}
\infrule[RE-LetF2]{(M_1,R,H)\eval (v_1,R_1,H_1)\andalso (M_2,R_1\set{x\mapsto v_1}, H_1)\eval \FAIL}
        {(\letexp{x}{M_1}{M_2},R,H)\eval \FAIL}
The semantics is mostly standard, except the following points.
\begin{itemize}        
\item We assume that each variable is implicitly annotated with types, and in rule \rn{RE-Var},
  we refer to it, and restrict \(R\) if the type \(\tau\) of \(x\) is not sharable.
  \(R\exc{x^\tau}\) is defined as \(R\) if \(\Sharable(\tau)\), and \(R\setminus{x}\) otherwise. Because of this change of an environment,
  the evaluation result contains the \(R\)-component.
\item In \rn{RE-Op} and \rn{E-Op}, \(\sem{\Op{k}}\) denotes the mathematical function
  represented by \(\Op{k}\).
\item For \(\bty=\INT\), \(\TRUE\) and \(\FALSE\)
  in the rules for conditionals repsectively mean non-zero integers
  and \(0\).
\item When evaluating an expression containing bound variables, we implicitly assume
  that \(\alpha\)-renaming is applied as necessary to avoid variable clash.
\item \(\splitEnv{R}{S}\) used in \rn{RE-Fun} and \rn{RE-RFun} is analogous to \(\splitTE{\Gamma}{S}\) used in the typing rules;
  it splits the environment \(R\) into \(R_1\) on \(S\) and the rest of the environment \(R_2\), though the bindings on 
  variables of sharable types is shared between \(R_1\) and \(R_2\).
  \(\splitEnv{R}{S}\) is defined as \((R_1,R_2)\), where:
  \begin{align*}
  &  R_1 = \set{x^\tau\mapsto v\in R \mid x\in S}\\
   & R_2 = \set{x^\tau\mapsto v\in R \mid x\not\in S\lor \Sharable(\tau)}.
  \end{align*}
\end{itemize}        
\begin{figure}[thbp]
  \infrule[RE-Fail]{}{(\FAIL, R, H)\eval \FAIL}
  \infrule[RE-Var]{}{(x^\tau, R, H)\eval (R(x^\tau), R\exc{x^\tau},H)}

\rulesp        
\infrule[RE-Const]{}{(c,R,H)\eval (c, R, H)}
\rulesp        
\infrule[RE-Const]{\sem{\Op{k}}(R(x_1),\ldots,R(x_k))=c}
        {(\Op{k}(x_1,\ldots,x_k),R,H)\eval (c, R, H)}

\rulesp        
\infrule[RE-Deref]{}{(!x,R,H)\eval (H(R(x)), R\excr{x^{\Tref{\rty}}}{\rty}, H)}

\rulesp        
\infrule[RE-Mkref]{\mbox{$\ell$ fresh}}{(\mkref{x},R,H)\eval (\ell, R\exc{x^\rty}, H\set{\ell\mapsto R(x)})}

\rulesp        
\infrule[RE-Assign]{}{(x:=y,R,H)\eval (\Vunit, R\exc{x^\rty}, H\set{R(x)\mapsto R(y)})}

\rulesp        
\infrule[RE-Let]{(M_1,R,H)\eval (v_1,R_1,H_1)\andalso (M_2,R_1\set{x\mapsto v_1}, H_1)\eval (v_2,R_2,H_2)}
        {(\letexp{x}{M_1}{M_2},R,H)\eval (v_2,R_2\setminus{x},H_2)}
\rulesp        
\infrule[RE-IfT]{R(x)=\TRUE\andalso (M_1,R,H)\eval (v_1,R_1,H_1)}
        {(\ifexp{x}{M_1}{M_2},R,H)\eval (v_1,R_1,H_1)}
        
\rulesp        
\infrule[RE-IfF]{R(x)=\FALSE\andalso (M_2,R,H)\eval (v_2,R_2,H_2)}
        {(\ifexp{x}{M_1}{M_2},R,H)\eval (v_2,R_2,H_2)}
\rulesp        
\infrule[RE-Fun]{(R_1,R_2)=\splitEnv{R}{\FV(\lambda x.M)}}
        {(\lambda x.M,R,H)\eval (\Clos{\lambda x.M}{R_1}, R_2, H)}
\rulesp        
\infrule[RE-RFun]{(R_1,R_2)=\splitEnv{R}{\FV(\lambda x.M)}}
        {(\fixexp{f}{x}{M},R,H)\eval (\Clos{\fixexp{f}{x}{M}}{R_1}, R_2, H)}
\rulesp        
\infrule[RE-App1]{R(f)=\Clos{\lambda y.M}{R_f}\andalso (M, R_f\set{y\mapsto R(x)}, H)\eval (v, R',H')}
        {(f\,x,R,H)\eval (v, R,H')}
\rulesp        
\infrule[RE-App2]{R(f)=\Clos{\fixexp{g}{y}{M}}{R_f}\\
   ([f/g]M, R_f\set{f\mapsto \Clos{\fixexp{g}{y}{M}}{R_f}, y\mapsto R(x)}, H)\eval (v, R',H')}
        {(f\,x,R,H)\eval (v, R,H')}
\caption{Operational semantics of \bPCFRL{\bty}}
\label{fig:os-bpcfrl}
\end{figure}        

The big-step semantics of \bPCF{\bty} (extended with tuples) is defined in Figure~\ref{fig:os-bpcf}.
Here, \(w\) ranges over the set of values defined by
\[w ::= c\mid \ell\mid \Clos{\lambda x.\mvE}{S}\mid \Clos{\fixexp{f}{x}{\mvE}}{S}\mid
(w_1,\ldots,w_k).\]
\(S\) denotes an environment; unlike \bPCFRL{\bty}, we do not have a heap-component as \bPCF{\bty} do not have reference cells.
As for the case of \bPCFRL{\bty}, we have omitted evaluation rules for \(\FAIL\) (except the base case \rn{E-Fail}).

\begin{figure}[thbp]
  \infrule[E-Fail]{}{(\FAIL, S)\eval \FAIL}
  \infrule[E-Var]{}{(x, S)\eval S(x)}

\rulesp        
\infrule[E-Const]{}{(c,S)\eval c}
\rulesp        
\infrule[E-Const]{\sem{\Op{k}}(S(x_1),\ldots,S(x_k))=c}
        {(\Op{k}(x_1,\ldots,x_k),S)\eval c}

\rulesp        
\infrule[E-Tuple]{(\mvE_i,S)\eval w_i \mbox{ for each $i\in\set{1,\ldots,k}$}}
        {((\mvE_1,\ldots,\mvE_k),S)\eval (w_1,\ldots,w_k)}
\rulesp        
\infrule[E-Proj]{(\mvE,S)\eval (w_1,\ldots,w_k) \andalso 1\le j\le k}
        {(\mvE.j, S)\eval w_j}
\rulesp        
\infrule[E-Let]{(\mvE_1,S)\eval w_1\andalso (\mvE_2,S\set{x\mapsto w_1})\eval w_2}
        {(\letexp{x}{\mvE_1}{\mvE_2},S)\eval w_2}
\rulesp      
\infrule[E-IfT]{S(x)=\TRUE\andalso (\mvE_1,S)\eval w}
        {(\ifexp{x}{\mvE_1}{\mvE_2},S)\eval w}
        
\rulesp        
\infrule[E-IfF]{S(x)=\FALSE\andalso (\mvE_2,S)\eval w}
        {(\ifexp{x}{\mvE_1}{\mvE_2},S)\eval w}
\rulesp        
\infrule[E-Fun]{}
        {(\lambda x.\mvE,S)\eval \Clos{\lambda x.\mvE}{S}}
\rulesp        
\infrule[E-RFun]{}
        {(\fixexp{f}{x}{\mvE},S)\eval \Clos{\fixexp{f}{x}{\mvE}}{S}}
\rulesp        
\infrule[E-App1]{
  (\mvE,S)\eval v\andalso 
  S(f)=\Clos{\lambda y.\mvE_1}{S_f}\andalso (\mvE_1, S_f\set{y\mapsto v})\eval w}
        {(f\,\mvE,S)\eval w}
\rulesp        
\infrule[E-App2]{
(\mvE,S)\eval v\andalso 
  S(f)=\Clos{\fixexp{g}{y}{\mvE_1}}{S_f}\andalso ([f/g]\mvE_1,
  S_f\set{f\mapsto \Clos{\fixexp{g}{y}{\mvE_1}}{S_f},y\mapsto v})\eval w}
        {(f\,\mvE,S)\eval w}

\caption{Operational semantics of \bPCF{\bty}}
\label{fig:os-bpcf}
\end{figure}        

We now re-state Theorem~\ref{th:main} more formally.
\begin{theorem}
\label{th:restatement-main}  
  If \(\emptyTE \p M:\bt\q \emptyTE \Tr \mvE \),
  then \((M,\emptyset,\emptyset)\eval \FAIL\)  if and only if \((\mvE,\emptyset)\eval \FAIL\).
\end{theorem}

To prove the theorem, we need to introduce auxiliary relations.
First, we define the type judgment \(\p (R,H):\Gamma\) by the rules in Figure~\ref{fig:typing-runtime}.
The judgment means that the pair of the environment \(R\) and
the store \(H\) together provides an environment conforming to \(\Gamma\).
\begin{figure}
  \infrule[TyR-Unit]{}{\p (\Vunit,\emptyset):\UNIT}
  \rulesp
  \infrule[TyR-Const]{}{\p (c^\bty,\emptyset):\bty}
  \rulesp
\infrule[TyR-Ref]{\p (v,H):\rty}{\p (\ell,H\set{\ell\mapsto v}):\Tref{\rty}}
  \rulesp
  \infrule[TyR-Fun]{\p (R,H):\Gamma\andalso \Gamma,x\COL\tau_1\p M:\tau_2\q \Gamma,x\COL\tau_1\andalso |\Gamma|=m}
          {\p (\Clos{\lambda x.M}{R}, H):\Tfuns{\tau_1}{}{m}{\tau_2}}
  \rulesp
  \infrule[TyR-RFun]{\p (R,H):\Gamma\andalso \Gamma,f\COL\Trfun{\tau_1}{}{\Gamma}{\tau_2},x\COL\tau_1\p M:\tau_2\q \Gamma,f\COL\Trfun{\tau_1}{}{\Gamma}{\tau_2},x\COL\tau_1\andalso |\Gamma|=m}
          {\p (\Clos{\fixexp{f}{x}{M}}{R}, H):\Tfuns{\tau_1}{}{m}{\tau_2}}
\infrule[TyR-Env0]{}{\p (R,H):\emptyTE}
        \infrule[TyR-Env1]{\p (R,H):\Gamma\qquad \p (v,H'):\tau}
                  {\p (R\set{x\mapsto v}, H\uplus H'):(\Gamma,x\COL\tau)}
       \infrule[TyR-Env2]{\p (R,H):\Gamma \qquad
         \Gamma,f\COL\Trfun{\tau_1}{}{\Gamma}{\tau_2},x\COL\tau_1\p M:\tau_2\q
         \Gamma,f\COL\Trfun{\tau_1}{}{\Gamma}{\tau_2},x\COL\tau_1\andalso
       R_f\proj{\dom(\Gamma)} = R\proj{\dom(\Gamma)}}
       {\p (R\set{f\mapsto \Clos{\fixexp{f}{x}{M}}{R_f}},H):(\Gamma,f\COL\Trfun{\tau_1}{}{\Gamma}{\tau_2})}
\caption{Typing run-time states.}
\label{fig:typing-runtime}          
\end{figure}

\begin{lemma}[type preservation]
  \label{lem:type-preservation}
  If \(\p (R,H):\Gamma\), \(\Gamma\p M:\tau\q \Gamma'\) and
  \((M,R,H)\eval (v,R',H')\), then \(\p (R'\set{x\mapsto v}, H'):(\Gamma',x\COL\tau)\),
  where \(x\) is a fresh variable (i.e., \(x\not\in\dom(R)\)).
\end{lemma}

We prepare some lemmas.
\begin{lemma}
  \label{lem:sharable-heap}
  If \(\p (v, H):\tau\) and \(\Sharable(\tau)\), then \(\p (v,\emptyset):\tau\).
  Similarly, if \(\p (R,H):\Gamma\) and \(|\Gamma|=0\), then
  \(\p (R,\emptyset):\Gamma\).
\end{lemma}
\begin{proof}
  This follows by simultaneous induction on the derivations of
  \(\p (v, H):\tau\) and \(\p (R,\emptyset):\Gamma\).
\end{proof}

\begin{lemma}
  \label{lem:postTE}
  If \(\Gamma \p M:\tau\q \Gamma'\), then \(\Gamma'\subseteq \Gamma\).
\end{lemma}
\begin{proof}
  This follows by straightforward induction on the derivation of \(\Gamma \p M:\tau\q \Gamma'\).
\end{proof}

\begin{lemma}
  \label{lem:postR}
  If \(\p (R,H):\Gamma\), \(\Gamma\p M:\tau\q \Gamma'\) and
  \((M,R,H)\eval (v,R',H')\), then
  \(R\proj{\dom(\Gamma')}=R'\proj{\dom(\Gamma')}\).
\end{lemma}
\begin{proof}
  The proof proceeds by induction on the derivation of
  \((M,R,H)\eval (v,R',H')\), with case analysis on the last rule used.
  We discuss only the case for \rn{RE-Let}, as the other cases are straightforward.
  If \((M,R,H)\eval (v,R',H')\) is derived by using \rn{RE-Let}, we have:
  \begin{align*}
    &  M=\letexp{x}{M_1}{M_2}\\&
    (M_1,R,H)\eval (v_1,R'',H'')\qquad
    (M_2,R''\set{x\mapsto v_1}, H'')\eval (v,R',H')\\&
    \Gamma\p M_1:\tau_1\q \Gamma_1 \qquad
    \Gamma_1,x\COL\tau_1\p M_2:\tau_2\q \Gamma_2\qquad \Gamma'=\Gamma_2\setminus{x}.
  \end{align*}
  By the induction hypothesis, we have
  \begin{align*}
    & R''\proj{\dom(\Gamma_1)}=R\proj{\dom(\Gamma_1)}
    \qquad R''\proj{\dom(\Gamma')} = R'\proj{\dom(\Gamma')}.
  \end{align*}
  By Lemma~\ref{lem:postR}, we have \(\Gamma'\subseteq \Gamma_1\subseteq \Gamma\). Thus, we have \(R'\proj{\dom(\Gamma')}=R''\proj{\dom(\Gamma')} =R\proj{\dom(\Gamma')}\) as required.
\end{proof}
The following two lemmas means that the evaluation of a term is independent of
the ``irrelevant'' part of a heap.
\begin{lemma}
  \label{lem:heap-strengthening}
  Suppose \(\p (R,H_1):\Gamma\) and \(\Gamma \p M:\tau\q \Gamma'\),
  with \(H=H_1\uplus H_2\). If \((M,R,H)\eval (v,R',H')\),
  then \((M,R,H_1)\eval (v,R',H_1')\) and \(H'=H'_1\uplus H_2\) for some \(H_1'\).
\end{lemma}
\begin{proof}
  This follows by induction on the derivation of \((M,R,H)\eval (v,R',H')\).
\end{proof}
\begin{lemma}
  \label{lem:heap-weakening}
  If 
  \((M,R,H_1)\eval (v,R',H_1')\) and \(\dom(H_2)\cap \dom(H_1')=\emptyset\),
  then \((M,R,H_1\uplus H_2)\eval (v,R',H_1'\uplus H_2)\).
\end{lemma}
\begin{proof}
  This follows by straightforward induction on the derivation of \((M,R,H_1)\eval (v,R',H_1')\).
\end{proof}
\begin{lemma}
  \label{lem:heap-type-weakening}
  If \((R,H):\Gamma\) with \(R\subseteq R'\) and \(H\subseteq H'\), then
  \((R',H'):\Gamma\).
\end{lemma}
\begin{proof}
  This follows by straightforward induction on the derivation of
  \((R,H):\Gamma\).
\end{proof}
We are now ready to prove Lemma~\ref{lem:type-preservation}.
\begin{proof}[Proof of Lemma~\ref{lem:type-preservation}]
  This follows by induction on the derivation of \((M,R,H)\eval (v,R',H')\), with case analysis
  on the last rule used.
  \begin{itemize}
  \item Case \rn{RE-Var}:
    In this case, we have:
    \begin{align*}
      M=y^\tau\qquad R'=R\exc{y^\tau}\qquad v=R(y^\tau)\qquad H'=H\qquad \Gamma'=\Gamma\exc{y}.
    \end{align*}
    By the assumption \(\p (R,H):\Gamma\), we have:
    \begin{align*}
      \p (R(y), H_1):\tau\qquad \p (R\setminus{y},H_2):\Gamma\setminus{x}
      \qquad H=H_1\uplus H_2.
    \end{align*}
    If \(\Sharable(\tau)\), then
    we may assume \(H_1=\emptyset\) by Lemma~\ref{lem:sharable-heap}.
    Thus, by \(\p (R,H):\Gamma\) and \(\p (v,\emptyset):\tau\),
    we have \(\p (R'\set{x\mapsto v},H'):(\Gamma',x\COL\tau)\).
    If \(\neg\Sharable(\tau)\), then
    by \(\p (R(y), H_1):\tau\) and \(\p (R\setminus{y},H_2):\Gamma\setminus{x}\),
    we have \(\p (R'\set{x\mapsto v},H'):(\Gamma',x\COL\tau)\) as required.
  \item Case \rn{RE-Const}: The result follows immediately, as \(M=v=c\),
    \(R'=R\), \(H'=H\), and \(\Gamma'=\Gamma\).
  \item Case \rn{RE-Op}: The result follows immediately, as \(R'=R\), \(H'=H\),
    \(\tau=\bty\), \(\Gamma'=\Gamma\), and 
    \(v=\sem{\Op{k}}(R(x_1),\ldots,R(x_k))\) is a constant of type \(\bty\).
  \item Case \rn{RE-Deref}:
    In this case, we have:
    \begin{align*}
      & M=!y\qquad H(R(x))=v\qquad \Gamma(y)=\Tref{\tau}\\&
R'=R\excr{y}{\tau}\qquad H'=H\qquad \Gamma'=\Gamma\excr{y}{\tau}.
    \end{align*}
    If \(\Sharable(\tau)\) (i.e. $\tau=\bty$), then 
    \(\p (R'\set{x\mapsto v},H'):(\Gamma',x\COL\tau)\) follows immediately from
    \(\p (R,H):\Gamma\) and \((v,\emptyset):\tau\).
    
    Suppose \(\neg\Sharable(\tau)\).
    By the condition \(\p (R,H):\Gamma\) with \(\Gamma(y)=\Tref{\tau}\),
    we also have
    \begin{align*}
      \p (v,H_1):\tau \qquad \p (R\setminus{x},H_2):\Gamma\setminus y\qquad
      H_1\uplus H_2=H.
    \end{align*}
    we have \(\p (R'\set{x\mapsto v},H_1\uplus H_2):(\Gamma',x\COL\tau)\),
    i.e., \(\p (R'\set{x\mapsto v},H'):(\Gamma',x\COL\tau)\) as required.
  \item Case \rn{RE-Mkref}:
    In this case, we have:
    \begin{align*}
    &  M=\mkref{y}\qquad v=\ell\qquad \tau=\Tref{\rty}
      \qquad R'=R\exc{y^\rty}\\& H'=H\set{\ell\mapsto R(y)}
      \qquad \Gamma'=\Gamma\exc{y}.
    \end{align*}
    If \(\Sharable(\rty)\) (i.e., \(\rty=\bty\)), then 
    by \(\p (\ell, \set{\ell\mapsto R(y)}):\Tref{\bty}\) and \(\p (R',H):\Gamma'\),
    we have \(\p (R'\set{x\mapsto \ell},H'):(\Gamma',x\COL\Tref{\rty})\)
    as required.
    
    Suppose \(\neg\Sharable(\rty)\). By \(\p (R,H):\Gamma\) and \(\Gamma(y)=\rty\),
    we have:
    \begin{align*}
      \p (R(y), H_1):\rty \qquad \p (R\setminus{y},H_2):\Gamma\setminus{y}
      \qquad H=H_1\uplus H_2.
    \end{align*}
    By the condition \(\p (R(y), H_1):\rty\), we have
    \(\p (\ell, H_1\set{\ell\mapsto R(y)}):\Tref{\rty}\).
    Therefore, we have:
    \(\p (R\setminus{y}{x\mapsto \ell}, H_1\set{\ell\mapsto R(y)}\uplus H_2):
    \Gamma',x\COL\Tref{\rty}\), i.e.,
    \(\p (R'\set{x\mapsto \ell},H'):(\Gamma',x\COL\tau)\) as required.
  \item Case \rn{RE-Assign}:
    In this case, we have:
    \begin{align*}
     & M=y:=z\qquad v=\Vunit\qquad \tau=\UNIT
      \qquad R'=R\exc{z}\qquad H'=H\set{R(y)\mapsto R(z)}\\&
      \Gamma'=\Gamma\exc{z}\qquad \Gamma(y)=\Tref{\rty}\qquad \Gamma(z)=\rty.
    \end{align*}
    By the assumption \(\p (R,H):\Gamma\), we have:
    \begin{align*}
      \p (R\setminus{y,z}, H_0):\Gamma\setminus{y,z}\qquad
      \p (R(y), H_y):\Tref{\rty}\qquad
      \p (R(z), H_z):\rty.
\end{align*}
    The last condition implies
    \(\p (R(y),H_z\set{R(y)\mapsto R(z)}):\Tref{\rty}\).
    Thus, we have
    \(\p (R\setminus{z}, H_0\uplus H_z\set{R(y)\mapsto R(z)}): \Gamma\setminus{z}\).
    If \(\neg\Sharable(\rty)\), then 
    \(\p (R', H'):\Gamma'\) follows by Lemma~\ref{lem:heap-type-weakening}.
    Thus, we have \(\p (R'\set{x\mapsto v},H'):\Gamma',x\COL\tau\) as required.
    
    If \(\Sharable(\rty)\), i.e., \(\rty=\bty\), then 
    from \(\p (R\setminus{z}, H_0\uplus H_z\set{R(y)\mapsto R(z)}): \Gamma\setminus{z}\), we obtain \(\p (R, H_0\uplus H_z\set{R(y)\mapsto R(z)}):
    \Gamma\).
    Thus, by Lemma~\ref{lem:heap-type-weakening},
    we have \(\p (R'\set{x\mapsto v},H'):\Gamma',x\COL\tau\) as required.
    \item Case \rn{RE-Let}:
    In this case, we have:
    \begin{align*}
      & M= \letexp{y}{M_1}{M_2} \qquad 
      (M_1,R,H)\eval (v_1,R_1,H_1)\\&
     (M_2,R_1\set{y\mapsto v_1},H_1)\eval (v,R_2,H')
      \qquad R'=R_2\setminus{y}\\&
      \Gamma\p M_1:\tau'\q \Gamma_1\qquad \Gamma_1,y\COL\tau'\p M_2:\tau\q\Gamma_2\qquad \Gamma'=\Gamma_2\setminus y.
    \end{align*}
    By applying the induction hypothesis to \((M_1,R,H)\eval (v_1,R_1,H_1)\),
    we obtain \(\p (R_1\set{y\mapsto v_1},H_1):(\Gamma_1,y\COL\tau_1)\).
    By applying the induction hypothesis to \((M_2,R_1\set{y\mapsto v_1},H_1)\eval (v,R_2,H')\), we obtain \(\p (R_2\set{x\mapsto v},H'):(\Gamma_2,x\COL\tau)\),
    which implies \(\p (R'\set{x\mapsto v},H'):(\Gamma',x\COL\tau)\), as required.
    \item Case \rn{RE-IfT}:
    In this case, we have:
    \begin{align*}
      & M= \ifexp{y}{M_1}{M_2}\qquad
      R(y)=\TRUE\qquad (M_1,R,H)\eval (v,R',H')\\&
      \Gamma(y)=\bty\qquad \Gamma\p M_1:\tau\q \Gamma'.
    \end{align*}
    By applying the induction hypothesis to \((M_1,R,H)\eval (v,R',H')\),
    we obtain \((R'\set{x\mapsto v},H'):(\Gamma',x\COL\tau)\) as required.
  \item Case \rn{RE-IfF}: Similar to the case \rn{RE-IfT} above.
  \item Case \rn{RE-Fun}: 
    In this case, we have:
    \begin{align*}
      & M= \lambda y.M_1\qquad (R_1,R')=\splitEnv{R}{\FV(M)} \qquad
      v=\Clos{\lambda y.M_1}{R_1}\qquad H'=H\\&
      (\Gamma_1,\Gamma')=\splitTE{\Gamma}{\FV(M)}\qquad
      \Gamma_1,y\COL\tau_1\p M_1:\tau_2\q \Gamma_1,y\COL\tau_1\qquad
      \tau=\Tfuns{\tau_1}{\tau_1'}{|\Gamma_1|}{\tau_2}.
    \end{align*}
    By the conditons \(\p (R,H):\Gamma\),
    \((R_1,R')=\splitEnv{R}{\FV(M)}\) and
    \((\Gamma_1,\Gamma')=\splitTE{\Gamma}{\FV(M)}\), we have:
    \begin{align*}
    \p (R_1,H_1):\Gamma_1\qquad \p (R',H_2):\Gamma'\qquad (H'=)H = H_1\uplus H_2
    \end{align*}
    for some \(H_1\) and \(H_2\).
    By \(\p (R_1,H_1):\Gamma_1\) and
    \(\Gamma_1,y\COL\tau_1\p M_1:\tau_2\q \Gamma_1,y\COL\tau_1\),
    we have
    \(\p (v, H_1):\tau\).
    Thus, we have \(\p (R'\set{x\mapsto v},H'):(\Gamma',x\COL\tau)\) as required.
  \item Case \rn{RE-RFun}: 
    In this case, we have:
    \begin{align*}
      & M= \fixexp{f}{y}{M_1}\qquad (R_1,R')=\splitEnv{R}{\FV(M)} \qquad
      v=\Clos{M}{R_1}\qquad H'=H\\&
      (\Gamma_1,\Gamma')=\splitTE{\Gamma}{\FV(M)}\qquad
      \Gamma_1,f\COL\Trfun{\tau_1}{\tau_1'}{\Gamma_1}{\tau_2},
      y\COL\tau_1\p M_1:\tau_2\q \Gamma_1,f\COL\Trfun{\tau_1}{\tau_1'}{\Gamma_1}{\tau_2}, y\COL\tau_1\\&
      \tau=\Tfuns{\tau_1}{\tau_1'}{|\Gamma_1|}{\tau_2}.
    \end{align*}
    By the assumption \(\p (R,H):\Gamma\), we have:
    \begin{align*}
      \p (R_1,H_1):\Gamma_1 \qquad \p (R',H_2):\Gamma'\qquad H'=H_1\uplus H_2.
    \end{align*}
    By \(\p (R_1,H_1):\Gamma_1\) and the typing of \(M_1\),
    we have \(\p (v, H_1):\tau\). Thus, we have
    \(\p (R'\set{x\mapsto v},H'):(\Gamma',x\COL\tau)\) as required.
  \item Case \rn{RE-App1}:
    In this case, we have:
    \begin{align*}
      & M= f\,z\qquad R(f)=\Clos{\lambda y.M_1}{R_f}
      \qquad (M,R_f\set{y\mapsto R(z)},H)\eval (v,R'',H')\qquad R'=R\\&
      \Gamma(f)=\Tfuns{\tau_1}{}{m}{\tau}\qquad \Gamma(z)=\tau_1\qquad \Gamma'=\Gamma.
    \end{align*}
    By the conditions \(\p (R,H):\Gamma\) and \(R(f)=\Clos{\lambda y.M_1}{R_f}\),
    we have:
    \begin{align*}
     & \p (\Clos{\lambda y.M_1}{R_f},H_1):\Tfuns{\tau_1}{}{m}{\tau}
      \qquad \p (R\setminus\set{f,z}, H_0):\Gamma\setminus\set{f,z}
      \\& \p (R(z),H_2): \tau_1 \qquad H=H_0\uplus H_1\uplus H_2.
    \end{align*}
    Furthermore, by \(\p (\Clos{\lambda y.M_1}{R_f},H_1):\Tfuns{\tau_1}{}{m}{\tau}\),
    we have:
    \begin{align*}
      \p (R_f,H_1):\Gamma_1\qquad \Gamma_1,y\COL\tau_1\p M_1:\tau\q
      \Gamma_1,y\COL\tau_1\qquad |\Gamma_1|=m.
    \end{align*}
    Thus, we also have \(\p (R_f\set{y\mapsto R(z)}, H_1\uplus H_2):(\Gamma_1,y\COL\tau_1)\).
    By applying Lemma~\ref{lem:heap-strengthening} to
    \((M,R_f\set{y\mapsto R(z)},H)\eval (v,R'',H')\), we obtain
    \((M,R_f\set{y\mapsto R(z)},H_1\uplus H_2)\eval (v,R'',H'')\),
    with \(H'=H_0\uplus H''\).
    By the induction hypothesis and Lemma~\ref{lem:postR}, we have:
    \begin{align*}
     & \p (R''\set{x\mapsto v},H''):(\Gamma_1,y\COL\tau_1,x\COL\tau)\\&
      R''\proj{\dom(\Gamma_1)\cup\set{y}} =
      (R_f\set{y\mapsto R(z)})\proj{\dom(\Gamma_1)\cup\set{y}}.
    \end{align*}
    Thus, we also have
    \begin{align*}
     & \p (R_f, H_1'):\Gamma_1\qquad
      \p (R(z),H_2'):\tau_1\qquad
      \p (v, H_3'): \tau\\& H''=H_1'\uplus H_2'\uplus H_3'.
    \end{align*}
    The first condition also implies
    \(\p (\Clos{\lambda y.M_1}{R_f},H_1'):\Tfuns{\tau_1}{}{m}{\tau}\).
    By these conditions and \(\p (R\setminus\set{f,z}, H_0):\Gamma\setminus\set{f,z}\),
    we have:
    \begin{align*}
      \p (R\set{x\mapsto v}, H_0\uplus H_1'\uplus H_2'\uplus H_3'):(\Gamma,x\COL\tau).
    \end{align*}
    By \(H_0\uplus H_1'\uplus H_2'\uplus H_3'=H_0\uplus H''=H'\)
    with \(R'=R\) and \(\Gamma'=\Gamma\), we have
    \(\p (R'\set{x\mapsto v}, H'):(\Gamma',x\COL\tau)\) as required.
   \item Case \rn{RE-App2}: 
    In this case, we have:
    \begin{align*}
      & M= f\,z\qquad R(f)=\Clos{\fixexp{f}{y}{M_1}}{R_f}\\&
       (M_1,R_f\set{f\mapsto R(f),y\mapsto R(z)},H)\eval (v,R'',H')\qquad R'=R.
    \end{align*}
          The assumption \(\Gamma\p M:\tau\q \Gamma'\) must have been derived from either \rn{T-App} or
          \rn{T-RApp}.
          In the former case, the proof is similar to the case for \rn{RE-App} above.
          Thus, we discuss only the case where \(\Gamma\p M:\tau\q \Gamma'\) has been derived by \rn{T-RApp}. In this case, we have:
          \begin{align*}
           & \Gamma(f)=\Trfun{\tau_1}{}{\Delta}{\tau}\qquad \Gamma(z)=\tau_1\qquad
            \Gamma\setminus\set{f,z} = \Delta,\Gamma_0\\&
            \Gamma'=\Gamma.
          \end{align*}
          By \(\p (R,H):\Gamma\), we have:
          \begin{align*}
            & \p (R,H_0):\Gamma_0 \qquad (R,H_1):\Delta
            \qquad (R(z),H_2):\tau_1\qquad H=H_0\uplus H_1\uplus H_2\\
            & \Delta,f\COL\Trfun{\tau_1}{}{\Delta}{\tau},y\COL\tau_1
            \p M_1\q \Delta,f\COL\Trfun{\tau_1}{}{\Delta}{\tau},y\COL\tau_1\\&
            R_f\proj{\dom(\Delta)}=R\proj{\dom{\Delta}}.
          \end{align*}
    By Lemma~\ref{lem:heap-strengthening} with
    \(   (M_1,R_f\set{f\mapsto R(f),y\mapsto R(z)},H)\eval (v,R'',H')\) and
    \(R_f\proj{\dom(\Delta)}=R\proj{\dom{\Delta}}\),
    we obtain
    \begin{align*}
      (M_1,R\set{f\mapsto R(f),y\mapsto R(z)},H_1\uplus H_2)\eval (v,R''',H'')
      \qquad H' = H_0\uplus H''.
    \end{align*}
    Thus, by the induction hypothesis, we have
    \begin{align*}
      \p (R'''\set{x\mapsto v},H''):(\Delta,f\COL\Trfun{\tau_1}{}{\Delta}{\tau},y\COL\tau_1,x\COL\tau).
      \end{align*}
    By Lemma~\ref{lem:postR}, we also have:
    \begin{align*}
      \p (R\set{f\mapsto R(f),y\mapsto R(z)},H''):(\Delta,f\COL\Trfun{\tau_1}{}{\Delta}{\tau},y\COL\tau_1).
    \end{align*}
    Together with the condition \(\p (R,H_0):\Gamma_0\), we obtain
    \begin{align*}
      \p (R\set{x\mapsto v},H_0\uplus H''):(\Gamma,x\COL\tau),
    \end{align*}
   i.e., \(\p (R\set{x\mapsto v},H'):(\Gamma,x\COL\tau)\), as required.
  \end{itemize}
\end{proof}

We next define relations between run-time states of the source and target languages.
We introduce the relation \((M,R,H)\con_{\Gamma,\tau} (\mvE,S)\) on expressions
(where \(\Gamma\) and \(\tau\) describe the type environment and type of \(M\)),
the relation \((v, R, H)\con_{\Gamma,\tau} ((w,h),S)\) on values (where \(h\) is a tuple representing
a store), and the relation \((R,H)\con_{\Gamma} S\) on run-time environments and heaps, as defined in Figure~\ref{fig:con}.
The relation \((R,H)\con_{\Gamma} S\) is defined by induction on \(\Gamma\) (here, recall that \(\Gamma\) is a sequence,
but \(R,H\), and \(S\) are finite maps, or sets of bindings).
Here, \(\Pack{\Gamma}(S)\) describes the tuple obtained by evaluating \(\pack{\Gamma}\) under the environment \(S\).
For example, \(\Pack{x\COL\BOOL, y\COL\Tref{\BOOL},z\COL\Tref{\BOOL}}(\set{x\mapsto \TRUE, y\mapsto \FALSE, z\mapsto \TRUE})
= (\FALSE,\TRUE)\).
Similarly, \(\Unpack{\Gamma}{h}\) describes the environment obtained by evaluating \(\unpackwo{\Gamma}{h}\).
For example, \(\Unpack{x\COL\BOOL, y\COL\Tref{\BOOL},z\COL\Tref{\BOOL}}{\FALSE,\TRUE} =
\set{y\mapsto \FALSE, z\mapsto \TRUE})\). (Note that \(x\) is irrelevant, as its type is sharable; only non-sharable types
are relevant to those operations.)
In the notation of the form \(S_1\cdot S_2\) used in \rn{Con-Fun}, \rn{Con-RFun}, and \rn{Con-Val},
the domains of \(S_1\) and \(S_2\) may overlap, and
in that case, \(S_2\) overwrites the overlapped bindings of \(S_1\). For example, \(\set{x\mapsto \TRUE}\cdot \set{x\mapsto \FALSE} =
\set{x\mapsto \FALSE}\).
In the rule \rn{Con-Fun}, the condition \((R_1,H_1)\con_{\Gamma_1} S_1\cdot \Unpack{\Gamma_1}{h_1}\)
requires that the environment and the local store of the closure are consistent between the source and target.
\begin{figure}[tp]
\infrule[Con-Unit]{}
        { (\Vunit,H)\con_{\UNIT} \Vunit}
\rulesp
\infrule[Con-Const]{}
        { (c^\bty,H)\con_{\bty} c^\bty}
\rulesp
        
\infrule[Con-Ref]{(v,H)\con_{\rty} w}
        { (\ell,H\set{\ell\mapsto v})\con_{\Tref{\rty}}
          w}
\rulesp
        
\infrule[Con-Fun]{
  \Gamma,x\COL\tau_1\p M:{\tau_2}\q \Gamma,x\COL\tau_1\Tr N\\
      (R,H)\con_{\Gamma} S\cdot \Unpack{\Gamma}{h}\\
  N_f = \lambda (x,h).\unpack{\Gamma}{h}
  \letexp{(r,h')}{N}\unpack{\Gamma,x\COL\tau_1}{h'}(r,x,\pack{\Gamma})}
        { (\Clos{\lambda x.M}{R},H)\con_{\Tfuns{\tau_1}{}{n}{\tau_2}}
   \Clos{\Clos{N_f}{S}}{h}}
        
\rulesp
\infrule[Con-RFun]{ 
  \Gamma,f\COL\Trfun{\tau_1}{\tau_1'}{\Gamma}{\tau_2}, x\COL\tau_1
  \p M:{\tau_2}\q
  \Gamma, f\COL\Trfun{\tau_1}{\tau_1'}{\Gamma}{\tau_2}, x\COL\tau_1\Tr N\\
      (R,H)\con_{\Gamma} S\cdot \Unpack{\Gamma}{h}\\
  N_f = \fixexp{f}{(x,h)}{\unpack{\Gamma}{h}
  \letexp{(r,h')}{N}\\\qquad\qquad\qquad\qquad\qquad\qquad\unpack{\Gamma,f\COL\Trfun{\tau_1}{\tau_1'}{\Gamma}{\tau_2},x\COL\tau_1}{h'}(r,x,\pack{\Gamma})}}
        { (\Clos{\fixexp{f}{x}{M}}{R},H)\con_{\Tfuns{\tau_1}{}{n}{\tau_2}}
          \Clos{\Clos{N_f}{S}}{h}}

\rulesp

\infrule[Con-Env0]{}
        { (R,H)\con_{\emptyTE} S}
\rulesp

\infrule[Con-Env1]{(R,H_1)\con_{\Gamma} S\andalso
  (v,H_2)\con_{\tau} w\\
  \mbox{$\tau$ is not of the form $\Trfun{\tau_1}{}{\Delta}{\tau_2}$}}
        { (R\set{x\mapsto v},H_1\uplus H_2)\con_{\Gamma,x\COL\tau} S\set{x\mapsto w}}
\rulesp
\infrule[Con-Env2]{(R,H)\con_{\Gamma} S\andalso
  (\Clos{M}{R_1},H)\con_{\Tfuns{\tau_1}{}{|\Gamma|}{\tau_2}} \Clos{\Clos{\mvE}{S_1}}{\Pack{\Gamma}(S)}\andalso
 R_1\subseteq R\\ S(y)=S_1(y)\mbox{ for each $y$ such that $y\COL\tau\in \Gamma$ and $\Sharable(\tau)$}}
        { (R\set{x\mapsto \Clos{M}{R_1}},H)\con_{\Gamma,x\COL\Trfun{\tau_1}{}{\Gamma}{\tau_2}} S\set{x\mapsto \Clos{\mvE}{S_1}}}
\rulesp

\infrule[Con-Exp]{\Gamma\p M:\tau\q \Gamma'\Tr N\andalso
           (R,H)\con_{\Gamma} S}
        { (M,R,H)\con_{\Gamma,\tau} (N, S)}

\rulesp
\infrule[Con-ValState]{H=H_1\uplus H_2\andalso (v,H_1)\con_{\tau} w\andalso
    (R,H_2)\con_{\Gamma} S\cdot \Unpack{\Gamma}{h}}
        {(v,R,H)\con_{\Gamma,\tau} ((w,h),S)}
\caption{Correspondence of run-time states.}
\label{fig:con}
\end{figure}

We prepare the following lemma.
\begin{lemma}
  \label{lem:con}
  \begin{itemize}
   \item If \((v,H)\con_{\tau}w\) and \(\Sharable(\tau)\), then \((v,\emptyset)\con_{\tau}w\).
   \item If \((R,H)\con_{\Gamma} S\) with \(R\subseteq R', H\subseteq H'\) and \(S\subseteq S'\),
     then \((R',H')\con_{\Gamma}S'\).
   \item Suppose \(\splitTE{\Gamma}{V}=(\Gamma_1,\Gamma_2\) and
\(\splitEnv{R}{V}=(R_1,R_2\).
If \((R,H)\con_{\Gamma}S\), then \((R_1,H_1)\con_{\Gamma_1} S\) and
\((R_2,H_2)\con_{\Gamma_2} S\) with \(H=H_1\uplus H_2\) for some \(H_1\) and \(H_2\).
    \end{itemize}
\end{lemma}
\begin{proof}
  The first two properties are trivial by the definitions of \(\con_{\tau}\) and \(\con_{\Gamma}\).
  To check the last property, let \(\Gamma_{1,1}=\set{x\COL\tau\in \Gamma_1\mid \neg\Sharable(\tau)}\),
    \(\Gamma_{1,2}=\set{x\COL\tau\in \Gamma_1\mid \Sharable(\tau)}\),
    \(R_{1,1} = \set{x^\tau\mapsto v\in R_1\mid \neg\Sharable(\tau)}\), and
    \(R_{1,2} = \set{x^\tau\mapsto v\in R_2\mid \neg\Sharable(\tau)}\).
    Since \(\Gamma\) is a shuffle of \(\Gamma_{1,1}\) and \(\Gamma_2\), by the assumption \((R,H)\con_{\Gamma}S\),
    we have \((R_{1,1}, H_1)\con_{\Gamma_{1,1}} S\) and \((R_2,H_2)\con_{\Gamma_2}S\) with
    \(H=H_1\uplus H_2\) for some \(H_1\) and \(H_2\). Since \(R_{1,2}\) is a subset of \(R_2\),
    we also have \((R_{1,2},H_3)\con_{\Gamma_{1,2}} S\) for some \(H_3\). By the first property of this lemma,
    we have \((R_{1,2},\emptyset)\con_{\Gamma_{1,2}} S\). Thus, we have \((R_1,H_1)\con_{\Gamma_1} S\) as requird.
\end{proof}
We now state and prove a main lemma.
\begin{lemma}
\label{lem:main1}  
  Suppose \(\p (R,H):\Gamma\), \(\Gamma\p M:\tau\q \Gamma'\Tr \mvE\), and
  \((M,R,H)\con_{\Gamma,\tau} (\mvE,S)\). 
  \begin{enumerate}
  \item If \((M,R,H)\eval (v,R',H')\), then
    \((\mvE,S)\eval (w,h)\) and \((v,R',H')\con_{\Gamma',\tau} ((w,h),S)\) for some \(w\) and \(h\).
  \item If    \((M,R,H)\eval \FAIL\), then \((\mvE,S)\eval \FAIL\).
  \end{enumerate}
\end{lemma}
\begin{proof}
  The proof proceeds by simultaneous induction on the derivation of
    \((M,R,H)\eval (v,R',H')\) or \((M,R,H)\eval \FAIL\), with case analysis on the last rule used.
    \begin{itemize}
    \item Case \rn{RE-Fail}:
    In this case, we have:
    \begin{align*}
      &    M=\FAIL\qquad
      & \mvE = (\FAIL,\pack{\Gamma}).
    \end{align*}
      Thus, we have \((\mvE,S)\eval \FAIL\), as required.
  \item Case \rn{RE-Var}:
    In this case, we have:
    \begin{align*}
      &    M=x\qquad v=R(x) 
      \qquad R'=R\exc{x}\qquad H'=H\qquad \Gamma'=\Gamma\exc{x}\\
      & \mvE = (x,\pack{\Gamma}) \qquad (R,H)\con_{\Gamma} S.
    \end{align*}
    Let \(w=S(x)\) and \(h=\Pack{\Gamma'}(S)\).
    Then we have \((\mvE,S)\eval (w,h)\).
    It remains to show \((v,R',H')\con_{\Gamma',\tau} ((w,h),S)\).
    By the condition \((R,H)\con_{\Gamma} S\),
    we have \((R(x), H_1) \con_{\tau} S(x)\) and \((R\setminus{x},H_2)\con_{\Gamma\setminus{x}} S\),
    with \(H=H_1\uplus H_2\).
    By \((R(x), H_1) \con_{\tau} S(x)\), we have \((v,H_1)\con_{\tau} w\).
    We also have \(S\cdot \Unpack{\Gamma'}{h} = S\cdot \Unpack{\Gamma'}{\Pack{\Gamma'}(S)}=S\).
Case analysis on whether \(\Sharable(\tau)\).
    \begin{itemize}
    \item If \(\Sharable(\tau)\), then \(R'=R\).
      By Lemma~\ref{lem:con} and \((v, H_1) \con_{\tau} w\), we have \((v,\emptyset)\con_{\tau}w\).
      Together with \((R,H)\con_{\Gamma} S\) and \(S\cdot \Unpack{\Gamma'}{h} = S\), we obtain
      \((v,R',H')\con_{\Gamma,\tau} ((w,h),S)\) as required.
\item If \(\neg\Sharable(\tau)\), then 
      \(R'=R\setminus{x}\) and \(\Gamma'=\Gamma\setminus{x}\).
      By the conditions \((v,H_1)\con_{\tau}w\) and  \((R\setminus{x},H_2)\con_{\Gamma\setminus{x}} S\) with \(H=H_1\uplus H_2\),
      we have \((R'\set{y\mapsto v},H)\con_{\Gamma',y\COL\tau} S\set{y\mapsto w}\),
      as required.
      \end{itemize}
  \item Case \rn{RE-Const}:
    In this case, we have:
      \begin{align*}
      &    M=v=c
      \qquad H=H'\qquad \Gamma'=\Gamma\qquad R'=R\\
      & \mvE = (c,\pack{\Gamma}) \qquad (R,H)\con_{\Gamma} S.
    \end{align*}
      Let \(w=c\), and \(h=\Pack{\Gamma}(S)\).
      Then, we have \((\mvE,S)\eval (w,h)\) and
\((v,R',H')\con_{\Gamma',\tau} ((w,h),S)\) as required.
      \item Case \rn{RE-Op}:
    In this case, we have:
      \begin{align*}
      &    M=\Op{k}(x_1,\ldots,x_k) \qquad v=\sem{\Op{k}}(R(x_1),\ldots,R(x_k))
      \qquad H=H'\qquad \Gamma'=\Gamma\qquad R'=R\\
      & \mvE = (\Op{k}(x_1,\ldots,x_k),\pack{\Gamma}) \qquad (R,H)\con_{\Gamma} S.
    \end{align*}
      Let \(w=v=\sem{\Op{k}}(R(x_1),\ldots,R(x_k))\), and \(h=\Pack{\Gamma}(S)\).
      Then, we have \((\mvE,S)\eval (w,h)\) and
\((v,R',H')\con_{\Gamma',\tau} ((w,h),S)\) as required.
    \item Case \rn{RE-Deref}:
    In this case, we have:
      \begin{align*}
      &    M=!x \qquad v = H(R(x))
      \qquad H=H'\qquad \Gamma'=\Gamma\excr{x}{\rty}\qquad R'=R\excr{x}{\rty} \\
      & \mvE = (x,\pack{\Gamma}) \qquad (R,H)\con_{\Gamma} S\qquad \tau=\rty.
      \end{align*}
      By the condition \((R,H)\con_{\Gamma} S\), we also have:
      \begin{align*}
        (v,H_1)\con_{\rty} S(x)\qquad (R\setminus{x},H_2)\con_{\Gamma\setminus x} S
        \qquad H_1\set{R(x)\mapsto v}\uplus H_2=H.
      \end{align*}
      Let \(w=S(x)\) and \(h=\Pack{\Gamma\excr{x}{\rty}}(S)\).
      Then, we have \((w,H_1)\con_{\tau} v\) and
      \((R',H_2)\con_{\Gamma'} S = S\cdot \Unpack{\Gamma}{h}\).
      If \(\neg\Sharable(\rty)\), then we have
      \((v,R',H') = (v,R',H_1\uplus H_2)\con_{\Gamma',\tau} ((w,h),S)\) as required.
      If \(\Sharable(\rty)\), then we have \((v,\emptyset)\con_{\rty} w\).
      Thus, from \((R,H)\con_{\Gamma} S\), we obtain
            \((v,R',H') \con_{\Gamma',\tau} ((w,h),S)\) as required.
      \item Case \rn{RE-Mkref}: 
    In this case, we have:
      \begin{align*}
      &    M=\mkref{x} \qquad v = \ell
        \qquad H'=H\set{\ell\mapsto R(x)}\qquad \Gamma'=\Gamma\exc{x^\rty}\qquad R'=R\exc{x^\rty}\\
      & \mvE = (x,\pack{\Gamma}) \qquad (R,H)\con_{\Gamma} S   \qquad \tau=\Tref{\rty} .
      \end{align*}
      Let \(w=S(x)\) and \(h=\Pack{\Gamma'}(S)\).
      Then we have  \((\mvE,S)\eval (w,h)\) and:
      \begin{align*}
        &(v,H_1\set{\ell\mapsto R(x)})\con_{\tau} w \qquad
        (R\setminus x, H_2)\con_{\Gamma\setminus x} S\qquad
        S = S\cdot \Unpack{\Gamma'}{h}.
      \end{align*}
      If \(\neg\Sharable(\rty)\), then we have
      \((v,R',H')\con_{\Gamma',\tau} ((w,h),S)\) as required.
      If \(\Sharable(\rty)\), i.e., \(\rty=\bty\), then we have \((v,\set{\ell\mapsto R(x)})\con_{\tau} w\). Thus, from \((R,H)\con_{\Gamma}S\), we obtain
      \((v,R',H')\con_{\Gamma',\tau} ((w,h),S)\) as required.
      \item Case \rn{RE-Assign}: 
    In this case, we have:
      \begin{align*}
      &    M=x:=y \qquad v = \Vunit
        \qquad H'=H\set{R(x)\mapsto R(y)}\qquad \Gamma'=\Gamma\exc{y^\rty}\qquad
        R'=R\exc{y^\rty}\\
        & \mvE = \letexp{x}{y}(\Vunit,\pack{\Gamma}) \qquad (R,H)\con_{\Gamma} S
        \qquad \tau=\UNIT.
      \end{align*}
      Let \(w=\Vunit\) 
      and \(h=\Pack{\Gamma}(S\set{x\mapsto R(y)})\). Then,
      we have \((\mvE,S)\eval (w,h)\).
      Furthermore, we have:
      \begin{align*}
    &    (v,\emptyset)\con_{\tau} w \qquad
        (R(x), H_x)\con_{\Tref{\rty}} S(x)\qquad
        (R(y), H_y)\con_{\rty} S(y)\\&
        (R\setminus\set{x,y},H_0)\con_{\Gamma\setminus\set{x,y}} S\qquad
        H=H_0\uplus H_x\uplus H_y\qquad
        S\set{x\mapsto R(y)} = S\cdot \Unpack{\Gamma}{h}.
      \end{align*}
      From \((R(y), H_y)\con_{\rty} S(y)\), we obtain
      \((R(x), H_y\set{x\mapsto R(y)})\con_{\Tref{\rty}} S(y)\).
      Thus, we have
      \((R',H')\con_{\Gamma'}
      S\set{x\mapsto R(y)}\), which implies
      \((v,R',H')\con_{\Gamma',\tau} ((w,h),S)\) as required.
     \item Case \rn{RE-Let}: 
    In this case, we have:
      \begin{align*}
        &    M=\letexp{x}{M_1}{M_2} \\
        & (M_1,R,H)\eval (v_1,R_1,H_1)
        \qquad (M_2,R_1\set{x\mapsto v_1},H_1)\eval (v,R',H')\\
        & \Gamma\p M_1:\tau'\q \Gamma_1\Tr \mvE_1 \qquad
        \Gamma_1,x\COL\tau'\p M_2:\tau\q \Gamma_2\Tr \mvE_2\\
        & \Gamma' = \Gamma_2\setminus{x}\qquad \mvE = \letexp{(x,h_1)}{\mvE_1}\unpack{\Gamma_1}{h_1}\mvE_2\\
        & (R,H)\con_{\Gamma} S.
      \end{align*}
      By the induction hypothesis, we have:
      \begin{align*}
    &    (\mvE_1,S)\eval (w_1,h_1)\qquad (v_1,R_1,H_1)\con_{\Gamma_1,\tau'}((w_1,h_1),S)
      \end{align*}
      for some \(w_1\) and \(h_1\).
      Thus, we have \[(M_2,R_1\set{x\mapsto v_1},H_1)\con_{\Gamma_1,x\COL\tau'}
      (\mvE_2,(S\cdot\Unpack{\Gamma_1}{h_1})\set{x\mapsto w_1}).\]
      Let \(S_1=(S\cdot\Unpack{\Gamma_1}{h_1})\set{x\mapsto w_1}\).
      By the induction hypothesis, we have:
      \begin{align*}
    &    (\mvE_2,S_1)\eval (w,h)\qquad (v,R',H')\con_{\Gamma_2,\tau}((w,h),S_1)
      \end{align*}
      for some \(w\) and \(h\).
      By the second condition, we have
      \begin{align*}
      &  H' = H_1\uplus H_2 \qquad (v,H_1) \con_{\tau} w \\&
        (R',H_2) \con_{\Gamma_2} S_1\cdot \Unpack{\Gamma_2}{h} = S\cdot \Unpack{\Gamma_2}{h}\set{x\mapsto v_1}.
      \end{align*}
   As \(\Gamma'=\Gamma\setminus x\), we have
      \[(R',H_2)\con_{\Gamma'} 
       S\cdot \Unpack{\Gamma_2}{h},\]
      which implies \((v,R',H')\con_{\Gamma',\tau} ((w,h),S)\).
      From \((\mvE_1,S)\eval (w_1,h_1)\) and \((\mvE_2,S_1)\eval (w,h)\),
we also obtain \((\mvE,S)\eval (w,h)\) as required.
     \item Case \rn{RE-IfT}: 
    In this case, we have:
      \begin{align*}
        &    M=\ifexp{x}{M_1}{M_2} \\
        & R(x)=\TRUE\qquad (M_1,R,H)\eval (v,R',H')\\
        & \Gamma\p M_1:\tau\q \Gamma'\Tr \mvE_1 \\
        & \mvE = \ifexp{x}{\mvE_1}{\mvE_2}\\
        & (R,H)\con_{\Gamma} S.
      \end{align*}
      By the conditions \(R(x)=\TRUE\) and
      \((R,H)\con_{\Gamma} S\), we have \(S(x)=\TRUE\).
      By the induction hypothesis, there exist \(w\) and \(h\) such that
      \begin{align*}
        (\mvE_1,S)\eval (w,h)\qquad (v,R',H')\con_{\Gamma',\tau} ((w,h),S).
      \end{align*}
      Thus, we have
        \((\mvE,S)\eval (w,h)\) and \((v,R',H')\con_{\Gamma',\tau} ((w,h),S)\) as required.
     \item Case \rn{RE-IfF}: Similar to the case \rn{RE-IfT}.
     \item Case \rn{RE-RFun}: 
    In this case, we have:
      \begin{align*}
        &    M=\fixexp{f}{x}{M_1}\qquad \tau = \Tfuns{\tau_1}{}{|\Gamma_1|}{\tau_2}\qquad
        (\Gamma_1,\Gamma')=\splitTE{\Gamma}{\FV(M)}\\
        & v = \Clos{\fixexp{f}{x}{M_1}}{R_1}\qquad (R_1,R')=\splitEnv{R}{\FV(M)}
        \qquad H'=H\\
        & \Gamma_1,f\COL\Trfun{\tau_1}{}{\Gamma_1}{\tau_2},x\COL\tau_1 \p M_1:\tau_2\q \Gamma_1,f\COL\Trfun{\tau_1}{}{\Gamma_1}{\tau_2},x\COL\tau_1 \Tr \mvE_1 \\
        & \mvE = (\Clos{\fixexp{f}{(x, h)}{\unpack{\Gamma_1}{h}
                \letexp{(r,h')}{\mvE_1} \\&\qquad
                \unpack{\Gamma_1,f\COL\Trfun{\tau_1}{}{\Gamma_1}{\tau_2},x\COL\tau_1}{h'}(r,x,\pack{\Gamma_1})}}{\pack{\Gamma_1}}, \pack{\Gamma'})\\
        & (R,H)\con_{\Gamma} S.
      \end{align*}
      Let \(h\) be \(\Pack{\Gamma'}(S)\),  and \(w\) be:
      \begin{align*}
        & \Clos{\Clos{\fixexp{f}{(x, h)}{\unpack{\Gamma_1}{h}
                \letexp{(r,h')}{\mvE_1} \\&\qquad\qquad\qquad
                \unpack{\Gamma_1,x\COL\tau_1}{h'}(r,x,\pack{\Gamma_1})}}{S}}{\Pack{\Gamma_1}(S)}.
      \end{align*}
      Then we have \((\mvE,S)\eval (w,h)\).
      By the condition \((R,H)\con_{\Gamma}S\), 
      \((\Gamma_1,\Gamma')=\splitTE{\Gamma}{\FV(M)}\),
      and \((R_1,R')=\splitEnv{R}{\FV(M)}\),
      we have
      \((R_1,H_1)\con_{\Gamma_1}S\) and
      \((R',H_2)\con_{\Gamma'}S\) for some \(H_1\) and \(H_2\) such that \(H=H_1\uplus H_2\).
      By the condition \((R_1,H_1)\con_{\Gamma_1}S\), we have
      \((v, H_1)\con_{\tau} w\).
      Thus, we have
      \((R'\set{y\mapsto v},H)\con_{\Gamma',y\COL\tau} S\set{y\mapsto w})\),
      which implies \((v,R',H')\con_{\Gamma',\tau} ((w,h),S)\) as required.
     \item Case \rn{RE-Fun}: Similar to the case \rn{RE-RFun} above.
     \item Case \rn{RE-App1}: 
       In this case, we have:
      \begin{align*}
        &    M=f\,x \qquad R=R_0\set{x\mapsto v_x, f\mapsto \Clos{\lambda y.{M_1}}{R_1}}\qquad R'=R\\
        & (M_1,R_1\set{y\mapsto v_x},H)\eval (v,R_1',H')\\
        & \mvE = \letexp{(r,x,h_f)}{\codeof{f}(x, \envof{f})} \letexp{f}{\Clos{\codeof{f}}{h_f}}   (r, \pack{\Gamma})\\
         & (R, H)\con_{\Gamma} S.
      \end{align*}
      By the conditions \((R, H)\con_{\Gamma} S\) and
      \(R=R_0\set{x\mapsto v_1, f\mapsto \Clos{\lambda y.{M_1}}{R_1}}\),
      \begin{align*}
        &  (R_0,H_0)\con_{\Gamma\setminus\set{f,x}} S_0
        \qquad (\Clos{\lambda y.{M_1}}{R_1}, H_1)\con_{\Tfuns{\tau_1}{}{|\Gamma_1|}{\tau_2}} S(f)\qquad (v_x,H_2)\con_{\tau_1} S(x)\\&
          H=H_0\uplus H_1\uplus H_2.
      \end{align*}
      By the condition \((\Clos{\lambda y.{M_1}}{R_1}, H_1)\con_{\Tfuns{\tau_1}{}{|\Gamma_1|}{\tau_2}} S(f)\), we also have:
      \begin{align*}
        &  (R_1,H_1)\con_{\Gamma_1} S_1\cdot \Unpack{\Gamma_1}{h_f}\\
        & \Gamma_1,y\COL\tau_1\p M_1:\tau\q \Gamma_1,y\COL\tau_1\Tr \mvE_1\\
         & N_f = \lambda{(y,h)}.{\unpack{\Gamma_1}{h} \letexp{(r,h')}{N_1}\unpack{\Gamma_1,y\COL\tau_1}{h'}(r,y,\pack{\Gamma_1})}\\
        & S(f) = \Clos{\Clos{N_f}{S_1}}{h_f}.
      \end{align*}
      By the conditions
      \(\Gamma_1,y\COL\tau_1\p
      M_1:\tau\q \Gamma,y\COL\tau_1\Tr \mvE_1\), \((v_x,H_2)\con_{\tau_1} S(x)\),
      and
      \((R_1,H_1)\con_{\Gamma_1} S_1\cdot \Unpack{\Gamma_1}{h_f}\), 
      we have 
      \((M_1,R_1\set{y\mapsto R(x)},H_1\uplus H_2)\con_{\Gamma_1,y\COL\tau_1}
      (\mvE_1, S_1\cdot \Unpack{\Gamma_1}{h_f}\set{y\mapsto S(x)})\).
      By applying Lemma~\ref{lem:heap-strengthening} to
      \((M_1,R_1\set{y\mapsto R(x)},H)\eval (v,R_1',H')\), we have
      \begin{align*}
        & (M_1,R_1\set{y\mapsto R(x)},H_1\uplus H_2)\eval (v,R_1',H'')
        & H' = H_0\uplus H''
      \end{align*}
      for some \(H''\).
      Thus, by applying the induction hypothesis,
      we have:
      \begin{align*}
    &  (\mvE_1, S_1\cdot \Unpack{\Gamma_1}{h_f}\set{y\mapsto S(x)})\eval (w_1,h_1)\\&
      (v,R_1',H'')\con_{\Gamma_1,y\COL\tau_1,\tau} ((w_1,h_1), S_1\set{y\mapsto S(x)})
      \end{align*}
      for some \(w_1\) and \(h_1\).
      We have \((\codeof{f}(x,\envof{f}), S)\eval (w_1, S''(y), \Pack{\Gamma_1}(S''))\)
      for \(S''=S_1\set{y\mapsto S(x)}\cdot \Unpack{\Gamma_1, y\COL\tau_1}{h_1}\).
      Therefore, we have
      \((\mvE,S)\eval (w_1,\Pack{\Gamma}(S'))\)
      for \(S'=S''\set{x\mapsto S''(y),\envof{f}\mapsto \Pack{\Gamma_1}(S'')}\) (here we have omitted some irrelevant bindings
      in \(S'\)).
      It remains to check \((v,R,H')\con_{\Gamma}((w_1,\Pack{\Gamma}(S')),S)\).
      By the condition \((v,R_1',H'')\con_{\Gamma_1,y\COL\tau_1,\tau} ((w_1,h_1),S''')\) where \(S'''=
      S_1\set{y\mapsto S(x)}\),
      we have:
      \begin{align*}
        &   H''=H_1'\uplus H_2'\uplus H_3'\qquad 
        (R_1',H_1')\con_{\Gamma_1} S'''\cdot \Unpack{\Gamma_1,y\COL\tau_1}{h_1}\\       &
        (R_1'(y),H_2')\con_{\tau_1} (S'''\cdot \Unpack{\Gamma_1,y\COL\tau_1}{h_1})(y)\\ &
        (v,H_3')\con_{\tau} w_1.
      \end{align*}
      By the conditions \((v,H_3')\con_{\tau} w_1\) and
      \((R_0,H_0)\con_{\Gamma\setminus\set{f,x}} S_0\) with
      \(H' = H_0\uplus H''=H_0\uplus H_1'\uplus H_2'\uplus H_3'\),
      it remains to show:
      \begin{align*}
        & (R(f),H_1')\con_{\Tfuns{\tau_1}{}{|\Gamma_1|}{\tau_2}} (S\cdot S')(f) \qquad
        (R(x),H_2')\con_{\tau_2} (S\cdot S')(x).
      \end{align*}
      The second condition follows from
      \((R_1'(y),H_2')\con_{\tau_1} (S'''\cdot \Unpack{\Gamma_1,y\COL\tau_1}{h_1})(y)\),
      because \(R_1'(y) = (R_1\set{y\mapsto R(x)})(y) = R(x))\),
      and \((S\cdot S')(x) = S''(y)=(\set{y\mapsto S(x)}\cdot \Unpack{\Gamma_1, y\COL\tau_1}{h_1})(y) = (S'''\cdot \Unpack{\Gamma_1,y\COL\tau_1}{h_1})(y)\).
      The remaining condition to check (i.e., the first condition above) is:
      \begin{align*}
        (\Clos{\lambda y.M_1}{R_1}, H_1')
        \con_{\Tfuns{\tau_1}{}{|\Gamma_1|}{\tau_2}}
        \Clos{\Clos{N_f}{S_1}}{\Pack{\Gamma_1}(S'')}.
      \end{align*}
      This follows from \((R_1',H_1')\con_{\Gamma_1} S'''\cdot \Unpack{\Gamma_1,y\COL\tau_1}{h_1}\), because
      \begin{align*}
      S'''\cdot \Unpack{\Gamma_1,y\COL\tau_1}{h_1}
      = S_1\set{y\mapsto S(x)}\cdot \Unpack{\Gamma_1,y\COL\tau_1}{h_1}
      = S''.
      \end{align*}
    \item Case \rn{RE-App2}:
      In this case, 
      we have:
      \begin{align*}
        &    M=f\,x \qquad R=R_0\set{x\mapsto v_x, f\mapsto \Clos{\fixexp{f}{y}{M_1}}{R_1}}\qquad R'=R\\
        & (M_1,R_1\set{f\mapsto \Clos{\fixexp{f}{y}{M_1}}{R_1},y\mapsto v_x},H)\eval (v,R_1',H')\\
         & (R, H)\con_{\Gamma} S.
      \end{align*}
      The assumption \(\Gamma\p M:\tau\q \Gamma'\Tr \mvE\) must have been derived from either \rn{Tr-App} or
      \rn{Tr-RApp}. In the former case, the proof is similar to the case \rn{RE-App} above.
      Thus, we discuss only the case where \(\Gamma\p M:\tau\q \Gamma'\Tr \mvE\) has been derived by \rn{Tr-RApp}.
      In that case, we have:
      \begin{align*}
        & \mvE = \letexp{h_f}{\pack{\Delta}}
       \letexp{(r,x,h'_f)}{f(x, h_f)} \\&\qquad
       \unpack{\Delta}{h'_f}
       (r, \pack{\Gamma})\\&
       \Gamma(f)=\Trfun{\tau_1}{}{\Delta}{\tau}\qquad \Gamma(x)=\tau_1\qquad \Gamma\setminus\set{f,x}=\Delta,\Gamma_0.
      \end{align*}
      By the condition \((R,H)\con_{\Gamma}S\), we have:
      \begin{align*}
        & (R_0,H_0)\con_{\Gamma_0} S\qquad (R_0,H_1)\con_{\Delta} S\\
        & (\Clos{\fixexp{f}{y}{M_1}}{R_1},H_1)\con_{\Tfuns{\tau_1}{}{|\Delta|}{\tau}} \Clos{S(f)}{\Pack{\Delta}(S)}
        \qquad R_1\subseteq R_0\\&
        (v_x, H_1)\con_{\tau_1} S(x) \qquad H = H_0\uplus H_1\uplus H_2\qquad
        H'=H_0\uplus H''\\&
         (M_1,R_1\set{f\mapsto \Clos{\fixexp{f}{y}{M_1}}{R_1},y\mapsto v_x},H_1\uplus H_2)\eval (v,R_1',H'').
      \end{align*}
      By the second condition, we also have:
      \begin{align*}
        &  S(f)=\Clos{\mvE_f}{S_1}\\&
        \mvE_f = \fixexp{f}{(y,h)}{\unpack{\Delta}{h}
          \letexp{(r,h')}{\mvE_1}\\&\qquad\qquad\qquad\qquad\unpack{\Delta,f\COL\Trfun{\tau_1}{\tau_1'}{\Delta}{\tau_2},y\COL\tau_1}{h'}(r,y,\pack{\Delta})}\\&
        \Delta,f\COL\Trfun{\tau_1}{}{\Delta}{\tau},y\COL\tau_1\p M_1:\tau\q
        \Delta,f\COL\Trfun{\tau_1}{}{\Delta}{\tau},y\COL\tau_1\Tr \mvE_1\\&
        (R_1,H_1)\con_{\Delta} S_1'\\&
        S_1'=S_1\cdot \Unpack{\Delta}{\Pack{\Delta}(S)}.
      \end{align*}
      The last condition, \((v_x, H_1)\con_{\tau_1} S(x)\) etc. imply:
      \begin{align*}
     &   (R_1\set{f\mapsto \Clos{\fixexp{f}{y}{M_1}}{R_1},y\mapsto v_x}, H_1\uplus H_2)\\&\qquad
        \con_{\Delta,f\COL\Trfun{\tau_1}{}{\Delta}{\tau},y\COL\tau_1} S_1'\set{f\mapsto \Clos{N_f}{S_1},y\mapsto S(x)}.
      \end{align*}
      By applying the induction hypothesis to
      \((M_1,R_1\set{f\mapsto \Clos{\fixexp{f}{y}{M_1}}{R_1},y\mapsto v_x},H_1\uplus H_2)\eval (v,R_1',H'')\), we obtain:
      \begin{align*}
      &  (\mvE_1, S'_1\set{f\mapsto \Clos{N_f}{S_1},y\mapsto S(x)})
        \eval (w, h')\\
        & (v, R_1', H'')\con_{\Delta,f\COL\Trfun{\tau_1}{}{\Delta}{\tau},y\COL\tau_1,\tau} ((w, h'),S_1'\set{f\mapsto \Clos{N_f}{S_1},y\mapsto S(x)}).
      \end{align*}
      Let \(h=\Pack{\Gamma}(S')\) where
      \begin{align*}
        & S' = S\set{x\mapsto S_1''(y)}\cdot \Unpack{\Delta}{\Pack{\Delta}(S''_1)}\\&
         S_1''=S_1'\set{f\mapsto \Clos{N_f}{S_1},y\mapsto S(x)}\cdot
        \Unpack{\Delta,f\COL\Trfun{\tau_1}{}{\Delta}{\tau},y\COL\tau_1}{h'}.
      \end{align*}
      Then, we have \((\mvE,S)\eval (w, h)\).
      It remains to check \((v,R',H')\con_{\Gamma,\tau} ((w,h),S)\).
      By the condition \((v, R_1', H'')\con_{\Delta,f\COL\Trfun{\tau_1}{}{\Delta}{\tau},y\COL\tau_1,\tau} ((w, h'),S_1'\set{f\mapsto \Clos{N_f}{S_1},y\mapsto S(x)})\), we have:
      \begin{align*}
        & (v, H_1')\con_{\tau} w\\&
        (R_1',H_2')\con_{\Delta,f\COL\Trfun{\tau_1}{}{\Delta}{\tau}} S_1'\set{f\mapsto \Clos{N_f}{S_1}}\cdot\Unpack{\Delta,f\COL\Trfun{\tau_1}{}{\Delta}{\tau},y\COL\tau_1}{h'}\\&
        (v_x,H_3')\con_{\tau_1}(S_1'\cdot\Unpack{\Delta,f\COL\Trfun{\tau_1}{}{\Delta}{\tau},y\COL\tau_1}{h'})(y)\\&
        H''=H_1'\uplus H_2'\uplus H_3'.
      \end{align*}
      Thus, to prove \((v,R',H')\con_{\Gamma,\tau} ((w,h),S)\), it suffices to show
      \begin{align*}
        &   (R',H_2')\con_{\Delta,f\COL\Trfun{\tau_1}{}{\Delta}{\tau}} S\cdot \Unpack{\Gamma}{h}\\&
         (v_x,H_3')\con_{\tau_1}(S\cdot \Unpack{\Gamma}{h})(x).
      \end{align*}
      The second condition follows by:
      \begin{align*}
        (S\cdot \Unpack{\Gamma}{h})(x) &= S'(x) = S_1''(y) = (S_1'\cdot\Unpack{\Delta,f\COL\Trfun{\tau_1}{}{\Delta}{\tau},y\COL\tau_1}{h'})(y).
      \end{align*}
      The first condition follows from the fact:
      \begin{align*}
     &  ( S_1'\set{f\mapsto \Clos{N_f}{S_1}}\cdot\Unpack{\Delta,f\COL\Trfun{\tau_1}{}{\Delta}{\tau},y\COL\tau_1}{h'})\proj{\dom(\Delta)\cup\set{f}} \\&=(S\cdot \Unpack{\Gamma}{h})\proj{\dom(\Delta)\cup\set{f}}.
      \end{align*}
      Indeed, for \(z\COL\tau_z\in \Delta,f\COL\Trfun{\tau_1}{}{\Delta}{\tau}\) such that \(\Sharable(\tau_z)\),
      the values of the maps on the both sides coincide with \(S(z)\), and
      for \(z\COL\tau_z\in \Delta,f\COL\Trfun{\tau_1}{}{\Delta}{\tau}\) such that \(\neg \Sharable(\tau_z)\),
      the values of the maps on the both sides coincide with \(\Unpack{\Delta,f\COL\Trfun{\tau_1}{}{\Delta}{\tau}}{h'}\).
    \item Inductive cases for \(\FAIL\): In this case, by the induction hypothesis, a subterm of \(\mvE\) evaluates to \(\FAIL\),
      from which \((\mvE,S)\eval \FAIL\) is obtained.
  \end{itemize}
\end{proof}
 
\nk{Revised up to this point.}
The following is the converse of the lemma above.
\begin{lemma}
\label{lem:main2}  
  Suppose \(\p (R,H):\Gamma\) and \((M,R,H)\con_{\Gamma,\tau} (\mvE,S)\). 
  \begin{enumerate}
  \item If \((\mvE,S)\eval w\), then \(w\) is of the form \((w_1,h)\), with
    \((M,R,H)\eval (v,R',H')\) and \((v,R',H')\con_{\Gamma',\tau} ((w_1,h),S)\) for some \(v, R'\) and \(H'\).
  \item If  \((\mvE,S)\eval \FAIL\), then \((M,R,H)\eval \FAIL\).
  \end{enumerate}
\end{lemma}
\begin{proof}
  By the assumption \((M,R,H)\con_{\Gamma,\tau} (\mvE,S)\), we have:
  \begin{align*}
    \Gamma\p M:\tau\q \Gamma'\Tr \mvE\qquad (R,H)\con_{\Gamma} S.
  \end{align*}
  The proof proceeds by induction on the derivation of
  \((\mvE,S)\eval w\) or \((\mvE,S)\eval \FAIL\), with case analysis on the rule used for
  \(\Gamma\p M:\tau\q \Gamma'\Tr \mvE\). 
  \begin{itemize}
  \item Case \rn{Tr-Fail}:
    In this case, we have:
    \begin{align*}
      & \mvE = (\FAIL, \pack{\Gamma}) \qquad M = \FAIL\qquad (\mvE,S)\eval \FAIL.
    \end{align*}
    Thus, we have \((M,R,H)\eval \FAIL\) as required.
  \item Case \rn{Tr-Unit}:
    In this case, we have:
    \begin{align*}
      & \mvE = (\Vunit, \pack{\Gamma}) \qquad M = \Vunit \qquad \tau=\UNIT\qquad \Gamma'=\Gamma.
    \end{align*}
    Suppose \((\mvE,S)\eval w\). Then it must be the case that \(w=(w_1,h)\) for
    \(w_1=\Vunit\) and \(h=\Pack{\Gamma}(S)\).
    Let \(v=\Vunit\), \(R'=R\), and \(H'=H\).
    Then we have \((M,R,H)\eval (v,R',H')\). We also have
    \((v,\emptyset)\con_{\tau} w_1\) and \((R,H)\con_{\Gamma'} S=S\cdot \Unpack{\Gamma}{h}\).
    Thus, we have \((v,R',H')\con_{\Gamma',\tau} ((w_1,h),S)\) as required.
  \item Case \rn{Tr-Const}: Similar to the case \rn{Tr-Unit} above.
  \item Case \rn{Tr-Op}:
    In this case, we have:
    \begin{align*}
      & \mvE = (\Op{k}(x_1,\ldots,x_k), \pack{\Gamma}) \qquad
      M = \Op{k}(x_1,\ldots,x_k) \qquad \tau=\bty\qquad \Gamma'=\Gamma.
    \end{align*}
    Suppose \((\mvE,S)\eval w\). Then it must be the case that \(w=(w_1,h)\) for
    \(w_1=\sem{\Op{k}}(S(x_1),\ldots,S(x_k))\) and \(h=\Pack{\Gamma}(S)\).
    Let \(v=w_1\), \(R'=R\), and \(H'=H\).
    Then we have \((M,R,H)\eval (v,R',H')\). We also have
    \((v,\emptyset)\con_{\tau} w_1\) and \((R,H)\con_{\Gamma'} S=S\cdot \Unpack{\Gamma}{h}\).
    Thus, we have \((v,R',H')\con_{\Gamma',\tau} ((w_1,h),S)\) as required.
  \item Case \rn{Tr-Var}:
    In this case, we have:
    \begin{align*}
      & \mvE = (x, \pack{\Gamma'}) \qquad M = x \qquad \Gamma'=\Gamma\exc{x}.
    \end{align*}
    Thus, \(w\) must be \((S(x), \Pack{\Gamma'}(S))\).
    Let \(w_1=S(x)\), \(h=\Pack{\Gamma'}(S)\), \(v=R(x), R'=R\exc{x}\), and \(H'=H\). Then we have
    \((M,R,H)\eval (v,R',H')\).
    It remains to check \((v,R',H')\con_{\Gamma',\tau} ((w_1,h),S)\).
    We perform case analysis on \(\Sharable(\tau)\).
    \begin{itemize}
    \item  Case \(\Sharable(\tau)\):
      By the condition \((R,H)\con_{\Gamma} S\), we have
     \((v,\emptyset)\con_{\tau} w_1\) and 
     \((R',H')\con_{\Gamma'} S = S\cdot \Unpack{\Gamma'}{h}\).
     Thus we have \((v,R',H')\con_{\Gamma',\tau} ((w_1,h),S)\) as required.
   \item Case \(\neg\Sharable(\tau)\):
     By the condition \((R,H)\con_{\Gamma} S\), we have:
     \begin{align*}
       (v,H_1)\con_{\tau} w_1\qquad (R',H_2)\con_{\Gamma'} S = S\cdot \Unpack{\Gamma'}{h}.
       \qquad H=H_1\uplus H_2.
     \end{align*}
      Thus, we have \((v,R',H')\con_{\Gamma',\tau} ((w_1,h),S)\) as required.
      \end{itemize}
  \item Case \rn{Tr-Deref}:
    In this case, we have:
    \begin{align*}
      & \mvE = (x, \pack{\Gamma}) \qquad M = \deref{x}\qquad \Gamma'=\Gamma\excr{x}{\rty}\qquad
       \Gamma(x)=\Tref{\rty}.
    \end{align*}
   Suppose \((\mvE,S)\eval w\). Then it must be the case that \(w=(w_1,h)\) for
    \(w_1=S(x)\) and \(h=\Pack{\Gamma}(S)\).
    Let \(v=H(R(x))\), \(R'=R\excr{x}{\rty}\), and \(H'=H\).
  Then we have \((M,R,H)\eval (v,R',H')\).
    By the condition \((R,H)\con_{\Gamma} S\), we have:
    \begin{align*}
      (v,H_1)\con_{\rty} w_1\qquad (R\setminus x, H_0)\con_{\Gamma\setminus x}S
     \qquad H_0\uplus H_1\set{R(x)\mapsto v}=H.
    \end{align*}
  If \(\Sharable(\rty)\), then we have 
    \((R,H)\con_{\Gamma'} S=S\cdot \Unpack{\Gamma}{h}\), which imply
    \((v,R',H')\con_{\Gamma',\tau} ((w_1,h),S)\) as required.
    
    If \(\neg\Sharable(\rty)\), then we have
    \((R', H_0)\con_{\Gamma'}S=S\cdot \Unpack{\Gamma}{h}\).
    By \((v,H_1)\con_{\rty} w_1\) and Lemma~\ref{lem:con}, we have
    \((v,H_1\set{R(x)\mapsto v})\con_{\rty} w_1\).
    We have thus \((v,R',H')\con_{\Gamma',\tau} ((w_1,h),S)\) as required.
  \item Case \rn{Tr-Mkref}:
    In this case, we have:
    \begin{align*}
      & \mvE = (x, \pack{\Gamma}) \qquad M = \mkref{x}\qquad \Gamma'=\Gamma\exc{x}\qquad
       \Gamma(x)={\rty}\qquad \tau=\Tref{\rty}.
    \end{align*}
   Suppose \((\mvE,S)\eval w\). Then it must be the case that \(w=(w_1,h)\) for
    \(w_1=S(x)\) and \(h=\Pack{\Gamma}(S)\).
    Let \(v=\ell\), \(R'=R\exc{x}\), and \(H'=H\set{\ell\mapsto R(x)}\) (where \(\ell\) is fresh).
    Then we have \((M,R,H)\eval (v,R',H')\).
    By the condition \((R,H)\con_{\Gamma} S\), we have:
    \begin{align*}
      (R(x), H_1)\con_{\rty} w_1\qquad (R\setminus x, H_0)\con_{\Gamma\setminus x} S
      \qquad H_0\uplus H_1=H.
    \end{align*}
    The first condition implies \((v,H_1\set{\ell\mapsto w})\con_{\Tref{\rty}} w_1\).
    If \(\Sharable(\rty)\), then
    we have \((v,\set{\ell\mapsto w})\con_{\Tref{\rty}} w_1\),
    which, together with
    \((R',H)\con_{\Gamma'} S=S\cdot \Unpack{\Gamma}{h}\), implies
    \((v,R',H')\con_{\Gamma',\tau} ((w_1,h),S)\) as required.
    If \(\neg\Sharable(\rty)\), then
    we have \((R', H_0)\con_{\Gamma'} S=S\cdot \Unpack{\Gamma}{h}\) and
    \((v,H_1\set{\ell\mapsto R(x)})\con_{\Tref{\rty}} w_1\). Thus,
   we also have  \((v,R',H') = (v,R',H_0\uplus H_1\set{\ell\mapsto R(x)}\con_{\Gamma',\tau} ((w_1,h),S)\) as required.
  \item Case \rn{Tr-Assign}:
    In this case, we have:
    \begin{align*}
      & \mvE = \letexp{y}{x}(\Vunit, \pack{\Gamma}) \qquad M = (y:=x)\qquad \Gamma'=\Gamma\exc{x}\\
      & \Gamma(y)=\Tref{\rty} \qquad \Gamma(x) = \rty.
    \end{align*}
    Suppose \((\mvE,S)\eval w\). Then, \(w=(\Vunit,h)\) for
    \(h=\pack{\Gamma'}(S\set{y\mapsto S(x)}))\).
    Let \(v = \Vunit\), \(R'=R\exc{x}\), and \(H'=H\set{R(x)\mapsto R(y)}\).
    Then we have \((M,R,H)\eval (v,R',H')\).
    By the condition \((R,H)\con_{\Gamma} S\), we have
    \begin{align*}
     & (R(y),H_y)\con_{\Tref{\rty}} S(y)\qquad
      (R(x),H_x)\con_{\rty} S(x)\qquad
      (R\setminus\set{x,y},H_0)\con_{\Gamma\setminus\set{x,y}} S\\&
      \qquad H=H_0\uplus H_x\uplus H_y.
    \end{align*}
    By the second and third conditions, we have
    \((R\setminus{x},H_0\set{R(y)\mapsto R(x)}\uplus H_x)\con_{\Gamma\setminus x}
    S\set{y\mapsto S(x)}\).

    If \(\Sharable(\rty)\), then we also have \((R(x),\emptyset)\con_{\rty} S(x)\);
    thus
    we have
    \((R,H_0\set{R(y)\mapsto R(x)}\uplus H_x)\con_{\Gamma}
    S\set{y\mapsto S(x)}= S\cdot \Unpack{\Gamma'}{h}\).
    By Lemma~\ref{lem:con}, we have
    \((v,R',H')\con_{\Gamma',\tau} ((\Vunit,h),S)\), as required.

    If \(\neg\Sharable(\rty)\), then
   by  \((R\setminus{x},H_0\set{R(y)\mapsto R(x)}\uplus H_x)\con_{\Gamma\setminus x}
        S\set{y\mapsto S(x)}\)
        and Lemma~\ref{lem:con},
        we have
        \((R',H')\con_{\Gamma'}
        S\set{y\mapsto S(x)}= S\cdot \Unpack{\Gamma'}{h}\).
    Thus, we have
     \((v,R',H')\con_{\Gamma',\tau} ((\Vunit,h),S)\), as required.
  \item Case \rn{Tr-Let}:
    In this case, we have:
    \begin{align*}
      & \mvE = \letexp{(x,h_1)}{\mvE_1}{\unpack{\Gamma_1}{h_1}\mvE_2}\\&
      M = \letexp{x}{M_1}{M_2}\qquad \Gamma'=\Gamma_2\setminus{x}\\&
      \Gamma\p M_1:\tau'\q \Gamma_1\Tr \mvE_1\qquad \Gamma_1,x\COL\tau'\p M_2:\tau\q \Gamma_2\Tr \mvE_2.
    \end{align*}
    Suppose \((\mvE,S)\eval w\). Then we have:
    \begin{align*}
      (\mvE_1,S)\eval (w', h_1)\qquad (\mvE_2,S\set{x\mapsto w'}\cdot \Unpack{\Gamma_1}{h_1})\eval w.
    \end{align*}
    By applying the induction hypothesis to \((\mvE_1,S)\eval (w', h_1)\), we obtain:
    \begin{align*}
      (M_1,R,H)\eval (v_1,R'',H'')\qquad (v_1,R'',H'')\con_{\Gamma_1,\tau} ((w',h_1), S)
    \end{align*}
    for some \(v_1\), \(R''\) and \(H''\). The second condition implies:
    \begin{align*}
      (v_1, H''_1)\con_{\tau} w'\qquad (R'',H''_2)\con_{\Gamma_1} S\cdot \Unpack{\Gamma''}{h_1}
      \qquad H''=H''_1\uplus H''_2,
    \end{align*}
    which further implies 
    \begin{align*}
      (R''\set{x\mapsto v_1},H'')\con_{\Gamma_1,x\COL\tau} S\set{x\mapsto w'}\cdot \Unpack{\Gamma_1}{h_1}.
    \end{align*}
    By applying the induction hypothesis to \((\mvE_2,S\set{x\mapsto w'}\cdot \Unpack{\Gamma_1}{h_1})\eval w\), we obtain:
    \begin{align*}
    &  (M_2,R\set{x\mapsto v_1},H'')\eval (v,R',H')\qquad (v,R',H')\con_{\Gamma_2,\tau}
      ((w_1,h), S)\\& w = (w_1,h)
    \end{align*}
    for some \(v,R'\), and \(H'\).
    Thus, we have \((M,R)\eval (v,R',H')\) and \((v,R',H')\con_{\Gamma',\tau}
      ((w_1,h), S)\) as required.
  \item Case \rn{Tr-If}:
    In this case, we have:
    \begin{align*}
      & \mvE = \ifexp{x}{\mvE_1}{\mvE_2}\qquad
      M = \letexp{x}{M_1}{M_2}\\&
      \Gamma\p M_1:\tau\q \Gamma'\Tr \mvE_1\qquad
      \Gamma\p M_2:\tau\q \Gamma'\Tr \mvE_2\qquad \Gamma(x)=\BOOL.
    \end{align*}
    Suppose \((\mvE,S)\eval w\). Then \(S(x)\in\set{\TRUE,\FALSE}\).
    Suppose \(S(x)=\TRUE\); the case where \(S(x)=\FALSE\) is similar.
    Then we have:
    \begin{align*}
      (\mvE_1,S)\eval w.
    \end{align*}
    By the induction hypothesis, we have:
    \begin{align*}
      (M_1,R,H)\eval (v,R',H')\qquad (v,R',H')\con_{\Gamma',\tau} ((w_1,h),S)\qquad w=(w_1,h).
    \end{align*}
    By the conditions \((R,H)\con_{\Gamma}S\) and \(S(x)=\TRUE\), we have \(R(x)\).
    Thus, we have \((M,R,H)\eval (v,R',H')\) and
    \((v,R',H')\con_{\Gamma',\tau} ((w_1,h),S)\) as required.
   \item Case \rn{Tr-Fun}:
     In this case, we have
     \begin{align*}
       & \mvE =
       (\Clos{\lambda (x, h).\unpack{\Gamma_1}{h}
                \letexp{(r,h')}{\mvE_1} \\&\qquad\qquad\qquad
                \unpack{\Gamma_1,x\COL\tau_1}{h'}(r,x,\pack{\Gamma_1})}{\pack{\Gamma_1}}, \pack{\Gamma_2})\\
       & M = \lambda x.M_1\\
       & \Gamma_1,x\COL\tau_1\p M_1:\tau_2\q \Gamma_1,x\COL\tau_1\Tr N_1\\
       & (\Gamma_1,\Gamma')=\splitTE{\Gamma}{\FV(\lambda x.M_1)}\\
       & \tau = \Tfuns{\tau_1}{}{|\Gamma_1|}{\tau_2}\\
       & (R, H)\con_{\Gamma} S.
     \end{align*}
     Suppose \((\mvE,S)\eval w\). Then, we have:
     \begin{align*}
       & w = (w_1, \Pack{\Gamma'}(S))\\
       & w_1 = \Clos{\Clos{\lambda (x, h).\unpack{\Gamma_1}{h}
                \letexp{(r,h')}{\mvE_1} \\&\qquad\qquad\qquad\qquad\qquad
                \unpack{\Gamma_1,x\COL\tau_1}{h'}(r,x,\pack{\Gamma_1})}{S}}{\Pack{\Gamma_1}(S)}
     \end{align*}
     Let \(v=\Clos{\lambda x.M_1}{R}\), \(R'=R\), and \(H'=H\). Then we have
     \((M,R,H)\eval (v,R',H')\).
     To prove \((v,R',H')\con_{\Gamma',\tau} ((w_1,\Pack{\Gamma'}(S)),S)\),
     we need to show
     \begin{align*}
       (v,H_1) \con_{\tau} w_1\qquad (R',H_2)\con_{\Gamma'} S\cdot \Unpack{\Gamma'}{\Pack{\Gamma'}(S)}\qquad
       H = H_1\uplus H_2
\end{align*}
     for some \(H_1\) and \(H_2\).
     By Lemma~\ref{lem:con} and \((R,H)\con_{\Gamma}S\), we have
     \begin{align*}
    &   (R_1,H_1)\con_{\Gamma_1}S \qquad (R_2,H_2)\con_{\Gamma'}S \qquad  H=H_1\uplus H_2\\& (R_1,R_2)=\splitEnv{\FV(M)}{R}.
     \end{align*}
     By Lemma~\ref{lem:con}, the first two conditions imply \((R,H_1)\con_{\Gamma_1}S\)
     and \((R,H_2)\con_{\Gamma_2}S\),
      from which
      we obtain \((v,H_1)\con_{\tau}w_1\) and \((R',H_2)\con_{\Gamma'} S\cdot \Unpack{\Gamma'}{\Pack{\Gamma'}(S)}(=S)\)
      as required.
    \item Case \rn{Tr-RFun}: Similar to the case \rn{Tr-Fun} above.
     \item Case \rn{Tr-App}:
       In this case, we have:
       \begin{align*}
   &      \mvE =   \letexp{(r,x,h_f)}{\codeof{f}(x, \envof{f})} \letexp{f}{\Clos{\codeof{f}}{h_f}}
       (r, \pack{\Gamma})\\
       & M = f\,x \qquad \Gamma(f)=\Tfuns{\tau_1}{}{m}{\tau}\qquad \Gamma(x)=\tau_1.
       \end{align*}
       By the condition \((R,H)\con_{\Gamma}S\), we also have:
       \begin{align*}
        & (R(f),H_1)\con_{\Tfuns{\tau_1}{}{m}{\tau}} S(f) \qquad
         (R(x), H_2)\con_{\tau_1} S(x)\qquad (R,H_0)\con_{\Gamma\setminus{f,x}} S\\&
         H=H_0\uplus H_1\uplus H_2.
       \end{align*}
       Furthermore, by the condition \((R(f),H_1)\con_{\Tfuns{\tau_1}{}{m}{\tau}} S(f)\), we have:\footnote{
       We omit to discuss the case where \(\mvE_f\) is a recursive function, which is a little more tedious but similar to
       the case where \(\mvE_f\) is a non-recursive function.}
       \begin{align*}
         & R(f) = \Clos{\lambda y.{M_1}}{R_1}\qquad \Gamma_1,x\COL\tau_1\p
         M_1:\tau_2\q \Gamma_1,x\COL\tau_1\Tr N_1\\
           &
          (R_1,H_1)\con_{\Gamma_1}S_1\cdot\Unpack{\Gamma_1}{h_1}\\
         & S(f) =\Clos{\Clos{N_f}{S_1}}{h_1}\\
        & N_f = \lambda {(y,h)}.\unpack{\Gamma_1}{h}
           \letexp{(r,h_2)}{\mvE_1} \\&\qquad\qquad\qquad
          \unpack{\Gamma_1,y\COL\tau_1}{h_2}(r,y,\pack{\Gamma_1}).
       \end{align*}
       Suppose \((\mvE,S)\eval w\). Then, we have:
       \begin{align*}
           & (S_1\set{y\mapsto S(x)}\cdot \Unpack{\Gamma_1}{h_1}, N_1)\eval (w_1,h_2)\\
          & w = (w_1, \Pack{\Gamma}(S\set{\envof{f}\mapsto \Pack{\Gamma_1}(S_1'),x\mapsto S_1'(y)})\\
          & S_1' = S_1\set{y\mapsto S(x)}\cdot \Unpack{\Gamma_1,y\COL\tau_1}{h_2}.
       \end{align*}
       By the above conditions, we have
       \begin{align*}
     &    (M_1, R_1\set{y\mapsto R(x)}, H_1\uplus H_2)
       \con_{\Gamma_1, y\COL\tau_1} (\mvE_1, S_1\cdot \Unpack{\Gamma_1}{h_1}\set{y\mapsto S(x)}).
       \end{align*}
       By the induction hypothesis, we have
       \begin{align*}
     &    (M_1, R_1\set{y\mapsto R(x)}, H_1\uplus H_2)\eval (v,R_1',H_1')\\
       &  (v,R_1',H_1') \con_{\Gamma_1, y\COL\tau_1,\tau} ((w_1,h_2), S_1\cdot \Unpack{\Gamma_1}{h_1}\set{y\mapsto S(x)}).
       \end{align*}
       for some \(v,R_1'\), and \(H_1'\). Let \(R'=R\) and \(H'=H_1'\uplus H_0\).
       By Lemma~\ref{lem:heap-weakening} and
       \((M_1, R_1\set{y\mapsto R(x)}, H_1\uplus H_2)\eval (v,R_1',H_1')\),
        we have
        \((M_1, R_1\set{y\mapsto R(x)}, H)\eval (v,R_1',H')\).
       Thus, we have \((M,R,H)\eval (v,R',H')\).

       It remains to check \((v,R',H')\con_{\Gamma,\tau}((w_1,h),S)\),
       where \(h = \Pack{\Gamma}(S\set{\envof{f}\mapsto \Pack{\Gamma_1}(S_1'),x\mapsto S_1'(y)})\).
       By the condition \((v,R_1',H_1') \con_{\Gamma_1, y\COL\tau_1,\tau} ((w_1,h_2), S_1\cdot \Unpack{\Gamma_1}{h_1}\set{y\mapsto S(x)})\) we have:
       \begin{align*}
         &  H_1'=H_{1,1}'\uplus H_{1,2}'\qquad
         (v, H_{1,1}')\con_{\tau} w_1\\& (R_1',H_{1,2}')\con_{\Gamma_1,y\COL\tau_1}
         S_1\cdot \Unpack{\Gamma_1}{h_1}\set{y\mapsto S(x)}\cdot \Unpack{\Gamma_1,y\COL\tau_1}{h_2}= S'_1.
\end{align*}
       By the latter condition and Lemma~\ref{lem:postR}, we have:
       \begin{align*}
         (R_1\set{y\mapsto R(x)}, H_{1,2}')\con_{\Gamma_1,y\COL\tau_1} S_1',
       \end{align*}
       which implies:
       \begin{align*}
         (R_1, H_{1,3}')\con_{\Gamma_1} S_1' \qquad
         (R(x), H_{1,4}')\con_{\tau_1} S_1'(y)\qquad H_{1,2}'=H_{1,3}'\uplus H_{1,4}'.
       \end{align*}
       By the first condition, we have:
       \begin{align*}
       (R(f), H_{1,3}') = (\Clos{\lambda y.M_1}{R_1}, H_{1,3}') \con_{\Tfuns{\tau_1}{}{m}{\tau}}
         \Clos{\Clos{N_f}{S_1}}{\Pack{\Gamma_1}(S_1')}.
       \end{align*}
       Together with \((R(x), H_{1,4}')\con_{\tau_1} S_1'(y)\) and \((R,H_0)\con_{\Gamma\setminus{f,x}} S\),
       we obtain:
       \begin{align*}
       (R, H_0\uplus H_{1,2}') \con_{\Gamma} S\set{\envof{f}\mapsto \Pack{\Gamma_1}(S_1'),x\mapsto S_1'(y)} = S\cdot \Unpack{\Gamma}{h}.
       \end{align*}
       Thus, we have
       \((v, R, H_{1,1}'\uplus H_0\uplus H_{1,2}')\con_{\Gamma,\tau} ((w_1,h),S)\), i.e.,
       \((v, R', H')\con_{\Gamma,\tau} ((w_1,h),S)\), as required.
\item Case \rn{Tr-RApp}:
       In this case, we have:
       \begin{align*}
   &      \mvE =   \letexp{h_f}{\pack{\Delta}}
       \letexp{(r,x,h'_f)}{f(x, h_f)} 
       \unpack{\Delta}{h'_f}
       (r, \pack{\Gamma})
\\
& M = f\,x \qquad \Gamma(f)=\Trfun{\tau_1}{}{\Delta}{\tau}\qquad \Gamma(x)=\tau_1
\qquad \Delta \subseteq \Gamma.
       \end{align*}
       Suppose \((N,S)\eval w\). Then we have:
       \begin{align*}
       &  h_1 = \Pack{\Delta}(S)\qquad
         (f(x,h_f), S\set{h_f\mapsto h_1})\eval (w_1,w_x,h_2)\\
         & w=(w_1,h)\qquad
         h=\Pack{\Gamma}(S\set{x\mapsto w_x}\cdot
         \Unpack{\Delta}{h_2}).
       \end{align*}
       By the condition \((R,H)\con_{\Gamma}S\) and \(\Delta\subseteq \Gamma\), we also have:
       \begin{align*}
         & H = H_0\uplus H_1\uplus H_2 \qquad
         (R,H_0)\con_{\Gamma_0}S\qquad
         (R\setminus x,H_1)\con_{\Delta,f\COL\Trfun{\tau_1}{}{\Delta}{\tau}}S\\&
(R(x),H_2)\con_{\tau_1} S(x) \qquad \Gamma = \Delta,f\COL\Trfun{\tau_1}{}{\Delta}{\tau},
         x\COL\tau_1,\Gamma'.
       \end{align*}
       By the above conditions and
       \((f(x,h_f), S\set{h_f\mapsto h_1})\eval (w_1,w_x,h_2)\), we have:
       \begin{align*}
         &  S(f)=\Clos{N_f}{S_1}\qquad R(f)=\Clos{\fixexp{f}{y}{M_1}}{R_1}\\&
         N_f = \fixexp{f}{(y,h)}{\unpack{\Delta}{h}
           \letexp{(r,h')}{\mvE_1}\\&\qquad\qquad\qquad\qquad\unpack{\Delta,f\COL\Trfun{\tau_1}{\tau_1'}{\Delta}{\tau},y\COL\tau_1}{h'}(r,y,\pack{\Delta})}\\&
         (R_1,H_1)\con_{\Delta} S_1\cdot \Unpack{\Delta}{h_1}\qquad R_1\subseteq R\\&
         \Delta,f\COL\Trfun{\tau_1}{\tau_1'}{\Delta}{\tau}, y\COL\tau_1
  \p M_1:{\tau}\q
  \Delta, f\COL\Trfun{\tau_1}{\tau_1'}{\Delta}{\tau}, y\COL\tau_1\Tr \mvE_1\\&
  (\mvE_1, S_1'')\eval
  (w_1,h_3)\\&
  S_1''=S_1\set{f\mapsto \Clos{N_f}{S_1},y\mapsto S(x)}\cdot \Unpack{\Delta}{h_1}\\&
  S_1' = S_1\set{f\mapsto \Clos{N_f}{S_1},y\mapsto S(x)}\cdot\Unpack{\Delta,f\COL\Trfun{\tau_1}{\tau_1'}{\Delta}{\tau},y\COL\tau_1}{h_3}\\&
  w_x=S_1'(y)\qquad h_2 = \Pack{\Delta}(S_1').
       \end{align*}
       By the conditions  \((R_1,H_1)\con_{\Delta} S_1\cdot \Unpack{\Delta}{h_1}\), etc.,
       we have:
       \begin{align*}
      &   (R_1\set{f\mapsto R(f),y\mapsto R(x)},H_1\uplus H_2)
         \con_{\Delta,f\COL\Trfun{\tau_1}{\tau_1'}{\Delta}{\tau},y\COL\tau_1} S_1''.
       \end{align*}
       Thus, by the induction hypothesis, we have:
       \begin{align*}
         &   (M_1,R_1\set{f\mapsto R(f),y\mapsto R(x)},H_1\uplus H_2)\eval (v,R_1',H_1')\\&
         (v,R_1',H_1')\con_{\Delta,f\COL\Trfun{\tau_1}{\tau_1'}{\Delta}{\tau},y\COL\tau_1}
         ((w_1,h_3),S_1'').
       \end{align*}
       Let \(R'=R\) and \(H'=H_0\uplus H_1'\). Then we have \((M,R,H)\eval (v,R',H')\).
       
       By the condition \((v,R_1',H_1')\con_{\Delta,f\COL\Trfun{\tau_1}{\tau_1'}{\Delta}{\tau},y\COL\tau_1,\tau}
       ((w_1,h_3),S_1'')\), we have:
       \begin{align*}
         & H_1' = H_{1,1}'\uplus H_{1,2}' \qquad (v,H_{1,1}')\con_{\tau} w_1 \\&
         (R_1',H_{1,2}')\con_{\Delta,f\COL\Trfun{\tau_1}{\tau_1'}{\Delta}{\tau},y\COL\tau_1} S_1''\cdot
\Unpack{\Delta,f\COL\Trfun{\tau_1}{\tau_1'}{\Delta}{\tau},y\COL\tau_1}{h_3} = S_1'.
       \end{align*}
       By the last condition, we have:
       \begin{align*}
      &   (R_1',H_{1,3}')\con_{\Delta,f\COL\Trfun{\tau_1}{\tau_1'}{\Delta}{\tau}} S_1'\qquad
         (R(x), H'_{1,4})\con_{\tau_1}w_x\\&
         H_{1,2}'=H'_{1,3}\uplus H'_{1,4}.
       \end{align*}
       Together with the condition \((R,H_0)\con_{\Gamma_0}S\),
  we obtain:
       \begin{align*}
         (R',H_0\uplus H_{1,2}')\con_{\Gamma} S\set{x\mapsto w_x}\cdot S_1' = S\cdot \Unpack{\Gamma}{h}.
       \end{align*}
       Thus, we have \((v,R',H_{1,1}'\uplus H_0\uplus H_{1,2}')\con_{\Gamma} ((w_1,h),S)\), i.e.,
       \((v,R',H')\con_{\Gamma} ((w_1,h),S)\), as required.
  \end{itemize}
\end{proof}
 Theorem~\ref{th:restatement-main} (i.e., Theorem~\ref{th:main}) is an immediate corollary
of Lemmas~\ref{lem:main1} and \ref{lem:main2} above.
 \section{Benchmark Programs}
\label{sec:benchmark}
Below, we give benchmark programs used in the experiments in
Section~\ref{sec:exp}. The comment part is the source program using reference cells,
and the uncommented part is the translated program, which was feeded to \mochi{}.
\begin{itemize}
\item \texttt{repeat\_ref}:
\begin{verbatim}
(*
let x = ref 0 in
let f() = x:= !x+1; !x in
let rec repeat n g =
  if n=1 then g()
  else g(); repeat (n-1) g
let m = * in
if m>0 then 
  assert(repeat m f = m)
 *)

let x=0 
let f = (x, fun x-> let x=x+1 in (x, x)) 
let rec repeat n (h, g) =
  if n=1 then let (r, _) = g h in r
  else let (r, h) = g h in
       repeat (n-1) (h, g)
let main m =
  if m>0 then
    assert(repeat m f = m)
\end{verbatim}
\item \texttt{repeat\_ref\_ng}: obtained from \texttt{repeat\_ref} by replacing the assertion with\\
   \texttt{assert(repeat m f = 1)}.
  
 \item \texttt{repeat\_localref}:
 \begin{verbatim}
 (*
let f() = let x = ref 0 in x:= !x+1; !x in
let rec repeat n g =
  if n=1 then g()
  else g(); repeat (n-1) g
let m = * in
if m>0 then 
  assert(repeat m f = 1)
 *)

let f = ((), fun _ -> let x=0 in let x=x+1 in (x, ()))
let rec repeat n (h, g) =
  if n=1 then let (r, _) = g h in r
  else let (r, h) = g h in
       repeat (n-1) (h, g)
let main m =
  if m>0 then
    assert(repeat m f = 1)
\end{verbatim}
\item \texttt{repeat\_localref\_ng}: obtained from \texttt{repeat\_localref} by replacing the assertion with
   \texttt{assert(repeat m f = m)}.
\item \texttt{inc\_before\_rec}:
\begin{verbatim}   
(*
let x = ref 0 
let rec f n =
  if n=0 then !x
  else (x := !x+1; f(n-1))
let main n =
 if n>=0 then assert(f n = n)
 *)

let x = 0
let rec f x n =
  if n=0 then (x, x)
  else let x=x+1 in f x (n-1)
let main n = if n>=0 then assert(let (r,_)=f x n in r=n)
\end{verbatim}
\item \texttt{inc\_before\_rec\_ng}: obtained from \texttt{inc\_before\_rec} by replacing the assertion with
   \texttt{assert(f n = n+1)}.
\item \texttt{inc\_after\_rec}:
\begin{verbatim}   
(*
let x = ref 0 
let rec f n =
  if n=0 then !x
  else (f(n-1); x := !x+1; !x)
let main n =
 if n>=0 then assert(f n = n)
 *)

let x = 0
let rec f x n =
  if n=0 then (x, x)
  else let (_,x) = f x (n-1) in let x=x+1 in (x,x)
let main n = if n>=0 then assert(let (r,_)=f x n in r=n)
\end{verbatim}
\item \texttt{inc\_after\_rec\_ng}: obtained from \texttt{inc\_after\_rec} by replacing the assertion with
   \texttt{assert(f n = n+1)}.
\item \texttt{borrow}:
\begin{verbatim}   
(*
let n = * 
let x = ref n 
let main() =
  let _ = let f() = x := !x+1 in
             f(); f()
  in assert(!x=n+2)
 *)

let rec loop() = loop()
let n = Random.int 0               
let x_c = n                
let x_f = Random.int 0
let x = (x_c,x_f) 
let assume b = if b then () else loop()
let main() =
  let _ = let f (x_c, x_f) = let x_c=x_c+1 in (x_c, x_f) in
            let x=f x in
            let (x_c,x_f)=f x in assume(x_c=x_f)
  in assert(x_f=n+2)        
\end{verbatim}
\item \texttt{borrow\_ng}: obtained from \texttt{borrow} by replacing the assertion with
   \texttt{assert(!x = n+1)}.
\item  \texttt{counter}:
\begin{verbatim}   
(*
type msg = Inc of unit | Read of unit
let newc init =
  let r = ref init in
  let f msg =
   match msg with
        Inc() -> (r := !r+1; 0)
      | Read() -> !r
  in f
let main n =
 let c = newc n in
   c(Inc()); assert(c(Read())=n+1)
 *)

type msg = Inc of unit | Read of unit
let newc init =
  let r = init in
  let f r msg =
    match msg with
      Inc() -> let r = r+1 in (0,r)
    | Read() -> (r, r)
  in (r, f)
let main n =
  let (h,c) = newc n in
   let (_,h)=c h (Inc()) in
   let (x,h)=c h (Read()) in assert(x=n+1)
\end{verbatim}
\item \texttt{counter\_ng}: obtained from \texttt{counter} by replacing the assertion with\\
   \texttt{assert(c(Read())=n+2)}.
\end{itemize}
 \fi
\end{document}